\documentclass[pra,onecolumn,superscriptaddress]{article}

\usepackage[T1]{fontenc}
\usepackage[utf8]{inputenc}
\usepackage{amsmath,amsfonts,amsthm,amssymb}
\usepackage{mathtools}
\usepackage{subeqnarray}
\usepackage{setspace}
\usepackage{graphics,graphicx,color}
\usepackage{url}
\usepackage{enumitem}
\usepackage{float} 
\usepackage{multirow}
\usepackage{hhline}
\usepackage[margin=1.25in]{geometry}
\usepackage{braket}
\usepackage{dsfont} 
\usepackage{authblk}
\usepackage{multirow}
\usepackage{hhline}
\usepackage[makeroom]{cancel}
\usepackage{ifthen}
\usepackage{comment}
\usepackage{breakcites}
\usepackage[dvipsnames]{xcolor}
\usepackage{microtype}
\usepackage{tcolorbox}

\newlist{notation}{description}{1}
\setlist[notation]{
    style=multiline,
    align=left,
    font=\normalfont,
    labelwidth=7em,
    leftmargin=9em,
    labelsep=1em,
    itemsep=0.4\baselineskip,
    parsep=0pt,
    topsep=0.5\baselineskip
}

\usepackage{mleftright}
\renewcommand{\left}{\mleft}
\renewcommand{\right}{\mright}

\usepackage{sectsty}
\allsectionsfont{\boldmath}

\makeatletter
\renewcommand*\l@section[2]{%
  \ifnum \c@tocdepth >\z@
    \addpenalty\@secpenalty
    \addvspace{1.0em \@plus\p@}%
    \setlength\@tempdima{1.5em}%
    \begingroup
      \parindent \z@
      \rightskip \@pnumwidth
      \parfillskip -\@pnumwidth
      \leavevmode
      \bfseries\boldmath
      \advance\leftskip\@tempdima
      \hskip -\leftskip
      #1\nobreak\hfil
      \nobreak\hb@xt@\@pnumwidth{\hss #2\kern-\p@\kern\p@}\par
    \endgroup
  \fi
}
\makeatother

\usepackage[nottoc]{tocbibind}

\definecolor{ddgreen}{rgb}{.05,.4,.05}
\definecolor{damethyst}{rgb}{0.4, 0.2, 0.6}
\usepackage[colorlinks, linkcolor=damethyst,citecolor=ddgreen]{hyperref}

\usepackage[nameinlink]{cleveref}

\usepackage{mathrsfs}

\usepackage{tabularray}

\newtheorem{theorem}{Theorem}
\newtheorem{cor}{Corollary}
\newtheorem{definition}{Definition}

\newtheorem{lem}{Lemma}
\newtheorem{fact}{Fact}
\newtheorem{prop}{Proposition}
\newtheorem*{prop*}{Proposition}

\newtheorem{remark}[theorem]{Remark}

\theoremstyle{definition}

\usepackage{keytheorems}
\newkeytheoremstyle{boxthm}
  {
    noteseparator={ },
    notebraces={(}{)},
    headpunct={.},
    bodyfont=\normalfont,
    tcolorbox={
    coltitle=black,
    toptitle=2pt,
    colback=OliveGreen!5,
  colframe=OliveGreen!25,
  boxrule=1pt,
  arc=6pt,
  left=6pt,right=6pt,top=6pt,bottom=6pt
    }, 
  }
\newkeytheorem{protocol}[style=boxthm]

\newkeytheoremstyle{minorboxthm}
  {
    noteseparator={ },
    notebraces={(}{)},
    headpunct={.},
    bodyfont=\normalfont,
    tcolorbox={
    coltitle=black,
    toptitle=2pt,
    colback=Gray!5,
  colframe=Gray!25,
  boxrule=1pt,
  arc=6pt,
  left=6pt,right=6pt,top=6pt,bottom=6pt
    }, 
  }
\newkeytheorem{minorprotocol}[name=Protocol,style=minorboxthm,sibling=protocol]

\usepackage[scr=boondoxo,scrscaled=1.05]{mathalfa}

\DeclareSymbolFont{yhlargesymbols}{OMX}{yhex}{m}{n} \DeclareMathAccent{\reallywidehat}{\mathord}{yhlargesymbols}{"62}

\newcommand{\C}{\mathbb{C}}

\newcommand{\1}{\mathbf 1}

\newcommand{\Ext}{\mathrm{Ext}}

\newcommand{\Var}{\mathrm{Var}}

\newcommand{\unif}{\textnormal{unif}}

\newcommand{\Int}{\mathrm{int}}

\newcommand{\frechet}{\textnormal{D}}

\newcommand{\beq}{\begin{eqnarray}}
\newcommand{\eeq}{\end{eqnarray}}

\newcommand\norm[1]{\left\Vert#1\right\Vert}

\newcommand{\poly}{\mathsf{poly}}

\DeclareMathOperator{\E}{\mathbf{E}}

\newcommand{\R}{\mathcal{R}}

\newcommand{\Avg}{\textnormal{\textbf{Avg}}}

\renewcommand{\epsilon}{\varepsilon}
\newcommand{\eps}{\epsilon}
\newcommand{\bs}{\boldsymbol}

\newcommand{\op}{\mathrm{op}}

\newcommand{\ketbra}[2]{\ket{#1}\!\!\bra{#2}}

\usepackage{tikz}

\newcommand{\tr}{\mbox{\textnormal{tr}}\,}

\definecolor{amethyst}{rgb}{0.6, 0.4, 0.8}

\newcommand{\anote}[1]{\textcolor{blue}{[Andrea: #1]}}
\newcommand{\jnote}[1]{\textcolor{teal}{[Joe: #1]}}
\newcommand{\jlnote}[1]{\textcolor{purple}{[Jerry: #1]}}

\title{Robust certification of almost all quantum states\\from single-qubit measurements}
\title{Robust quantum state certification and\\uncertainty principles for
total influence}

\author{Andrea Coladangelo\thanks{Email: \texttt{coladan@cs.washington.edu}}} 
\author{Jerry Li\thanks{Email: \texttt{jerryzli@cs.washington.edu}}} 
\author{Joseph Slote\thanks{Email: \texttt{jslote@cs.washington.edu}}}
\affil{Paul G. Allen School of Computer Science and Engineering, University of Washington}
\date{}

\newcommand{\wh}{\widehat}
\newcommand{\wt}{\widetilde}
\newcommand{\mc}{\mathcal}

\newcommand{\tsum}{{\textstyle\sum}}
\newcommand{\Inf}{\mathbf{Inf}}
\newcommand{\Bias}{\mathrm{Bias}}

\usepackage{tikz}

\newcommand{\gobackuparrow}{%
  \mathrel{%
    \tikz[baseline=-.13ex, x=1ex, y=1ex]{
      \draw[->, line width=0.6pt]
        (0,0)
        arc[start angle=270, end angle=180, radius=0.8]
        -- (-0.8,1.6);
    }%
  }%
}

\usepackage{fancyhdr}
\crefname{prop}{proposition}{propositions}
\Crefname{prop}{Proposition}{Propositions}

\begin{document}

\maketitle
\begin{abstract}

We show that nonadaptive single-qubit Pauli measurements suffice to test whether an unknown $n$-qubit state $\rho$ is $\eps$-close to or $\mc O(\eps)$-far from an ideal target state $\ket{\psi}$,
for all but a $2^{-\Omega(n)}$ fraction of target states.
The test uses $\mc O(\eps^{-2}\log(1/\delta))$ copies of $\rho$ to achieve confidence $1-\delta$, which is information-theoretically optimal even among protocols with arbitrary joint measurements.

The main technical innovation is an uncertainty principle for weighted generalizations of the total influence of Boolean functions.
As a simple example, the unweighted variant states that $\Inf[f]+\Inf[\wh f\,] = \Omega(n)$, which is a natural hypercube analogue of the Heisenberg uncertainty principle (here $\,\wh\cdot\,$ denotes the $2^{-n/2}$-normalized Fourier transform).
The weighted case generalizes $\Inf[\,\cdot\,]$ and $\Inf[\,\wh \cdot\,]$ to Dirichlet energies associated with Glauber dynamics for certain dual measures on the cube.
    
\end{abstract}

\newpage
\hypertarget{mytoc}{}
\setcounter{tocdepth}{2}
\tableofcontents

\newpage

\fancyfoot[L]{\footnotesize\hyperlink{mytoc}{$\gobackuparrow$\,\textsf{Contents}}}
\fancyfoot[R]{\footnotesize\hyperlink{gloss}{\textsf{Glossary}\,\raisebox{-1pt}{\rotatebox[origin=c]{180}{$\gobackuparrow$}}}}

\section{Introduction}

A basic task in quantum information is \textit{state certification}: testing whether a copy of an unknown state is close to a pure target state $\ket{\psi}$ (or far from it). This task has broad applications in quantum science, from benchmarking quantum devices to verifying the implementation of quantum algorithms~\cite{Badescu2019,ZhuHayashi,KlieschRoth,Huang2025}. With no constraints on measurement capabilities, the  optimal measurement is well-understood: it is simply the projective measurement $\{\ketbra{\psi}{\psi},\mathbf 1-\ketbra{\psi}{\psi}\}$.
However, implementing this measurement is in general just as difficult as preparing $\ket{\psi}$ itself, and there is significant interest in finding state certification protocols that only require simple measurements---ideally of only a few qubits at a time.
Remarkably, recent work has shown that single-qubit measurements suffice.

The first such result was by Huang, Preskill, and Soleimanifar \cite{Huang2025}, who proposed a protocol consisting of non-adaptive single-qubit Pauli measurements, which can certify all but an exponentially-small fraction of $n$-qubit target states.
In subsequent, independent works, Gupta, He, and O'Donnell~\cite{guptaHeODonnell2025}, and Li and Zhu~\cite{li2025universal} proposed somewhat more intricate protocols consisting of single-qubit measurements that can certify \textit{any} target state.
\cite{guptaHeODonnell2025} also showed that adaptivity is \textit{necessary} to get a single-qubit measurement protocol that can certify any state. In a more recent work, Abdul Sater et al.~\cite{sater2025efficientcertificationintractablequantum} describe a fully-efficient protocol (using only efficiently accessible information about the target state) to certify a useful, but restricted, class of states with Pauli measurements.

Unfortunately, all of the above protocols share a limitation, which is that they are non-robust: they can only certify whether the unknown state $\rho$ is $\mc O(\eps/n)$-close or at least $\eps$-far from the $n$-qubit target state $\ket{\psi}$. Said another way, these protocols only succeed in positively certifying lab states that are vanishingly close to the ideal state as the number of qubits $n$ grows, making them difficult to apply in practice: if an experimentalist wishes to double their system size (\textit{i.e.}, double $n$), the test will only remain usable if they simultaneously double their preparation~accuracy.
Moreover, one can show that this non-robustness is not just a limitation of the analysis, but is inherent to their protocols.\footnote{See, for example, \cite[Appendix A]{CLSW26}}

More recently, Coladangelo et al.\ \cite{CLSW26} gave two protocols that make a substantial step towards removing this limitation.
Their first protocol involves single-qubit measurements along with one $\log(n)$-qubit measurement and achieves \emph{constant} robustness: it can certify whether $\rho$ is $\eps/C$-close to or $\eps$-far from $\ket{\psi}$, for a universal constant $C\approx 2$, independent of $n$.
Their second protocol requires only single-qubit measurements, but only achieves $\log(n)$ robustness (\textit{i.e.},\ $C = \Theta(\log n)$).


The authors of \cite{CLSW26} left open the tantalizing possibility that one could achieve the best of both worlds: a certification protocol for almost all target states that achieves \emph{constant} robustness using only \emph{single-qubit} measurements (non-adaptively).
In particular, they conjectured \cite[Conjecture 1]{CLSW26} that a modification of their protocol (which amounts to running the \cite{Huang2025} test in two bases) should suffice.
In this work we prove their conjecture in the affirmative.  We show the following.
\begin{theorem}[Constant robustness with single-qubit measurements]
\label{thm:informal-1}
There exists an efficient protocol using single-qubit measurements that, for all but a $2^{-\Omega(n)}$ fraction of pure $n$-qubit target states $\ket{\psi}$, satisfies the following. Given oracle access to a classical description of the target state $\ket{\psi}$ (via the model in \Cref{def:oracle}), and a single copy of an $n$-qubit mixed state $\rho$,
\begin{itemize}
    \item (Completeness) Outputs $\mathsf{accept}$ with probability at least $\bra{\psi}\rho\ket{\psi}$.
    \item (Soundness) Outputs $\mathsf{reject}$ with probability at least
    $\frac{1-\braket{\psi|\rho|\psi}}{C} - o(1)$, for a universal constant $C>1$ (\textit{i.e.}\ approximately the infidelity up to a constant factor).
\end{itemize}
Moreover, the single-qubit measurements are either \emph{all} in the $X$-basis or \emph{all} in the $Z$-basis, except for one final (non-Pauli) adaptive measurement.
\end{theorem}

The completeness-soundness gap can be amplified via sequential repetition using more copies, and the small amount of adaptivity can be completely removed using classical shadow tomography.
As a result, \Cref{thm:informal-1} implies the following corollary: a copy-optimal test using exclusively non-adaptive single-qubit Pauli measuremerents, with arbitrarily low failure probability.
\begin{cor}[Optimal copy-complexity with with non-adaptive single-qubit Pauli measurements]
\label{cor:1}
    There exists an efficient protocol using only non-adaptive single-qubit Pauli measurements that, for all but a $2^{-\Omega(n)}$ fraction of pure $n$-qubit target states $\ket{\psi}$, satisfies the following. Given oracle access to a classical description of the target state $\ket{\psi}$ (via the model in \Cref{def:oracle}), parameters $\epsilon, \delta > 0$, and $O\left( \epsilon^{-2} \log (1/\delta)\right)$ copies of an arbitrary state $\rho$, 
    \begin{itemize}
        \item (completeness) Outputs \emph{\textsf{accept}} with probability at least $1 - \delta$ if $\bra{\psi} \rho \ket{\psi} \geq 1 - \frac{\epsilon}{C}$.
        \item (soundness) Outputs \emph{\textsf{reject}} with probability at least $1 - \delta$, if $\bra{\psi} \rho \ket{\psi} \leq 1 - \epsilon$.
    \end{itemize}
for some constant $C>1$.

Moreover, the single-qubit measurements are either \emph{all} in the $X$-basis or \emph{all} in the $Z$-basis, except for one qubit that is (non-adaptively) measured in a random Pauli basis.
\end{cor}

We pause here to make some comments on these results.
Concretely, the results resolve Conjecture~1 of~\cite{CLSW26}, and demonstrate that the two-basis test proposed in~\cite{CLSW26} yields optimal state certification guarantees (up to constants) for almost all states. 
Note that the restriction that the method works for \emph{almost} all states is inherent (and not a limitation of the analysis) since the test is non-adaptive, and~\cite{guptaHeODonnell2025} demonstrates that adaptivity is required to certify all states.

On a more technical level, our main conceptual contribution is a novel proof technique for analyzing the soundness of the two-basis test, which significantly departs from the previous proof strategy of~\cite{CLSW26}.
In particular, we draw a novel connection between the soundness of this test and new uncertainty principles for total influence of Boolean functions and, especially, more general weighted versions of total influence.
These uncertainty principles generalize the natural hypercube analog of the famous Heisenberg uncertainty principle, as we explain in \Cref{sec:ups}.

\begin{figure}
\centering
\begin{tblr}{
    colspec={Q[l] Q[l,wd=3.6cm] Q[l,wd=1.6cm] Q[l, wd=1.8cm] Q[l] Q[l]},
    row{1} = {f,font=\itshape},
  vlines={white,1.5pt},
  hlines={white, 1.5pt},
hline{2} = {1}{-}{0.4pt, black},
hline{2} = {2}{-}{0.4pt, black},
hline{2} = {3}{-}{.5pt, white},
hline{6} = {3}{-}{0.4pt, gray},
hline{6} = {3}{-}{1.5pt, white},
}
Reference      & Largest measurement required & Copy complexity & Fraction of target states & Adaptivity & Robustness\\
\cite{Huang2025}       & \SetCell{YellowGreen!20}Single-qubit (Pauli)         & $\mc O(n^4)$            & $1-2^{-\Omega(n)}$    & \SetCell{YellowGreen!20}No       & $\mc O(1/n^2)$          \\
\cite{guptaHeODonnell2025}       & \SetCell{YellowGreen!20}Single-qubit (arbitrary)     & $\mc O(n)$            & \SetCell{YellowGreen!20}All           & Yes      & $\mc O(1/n)$          \\
\cite{CLSW26}      & \SetCell{YellowGreen!20}Single-qubit (arbitrary)     & $\mc O(\log^2 n)$        & $1-2^{-\Omega(n)}$    & Yes      & $\mc O(1/\log n)$     \\
\cite{CLSW26}      & $\mathrm{log}(n)$-qubit (Clifford)      & \SetCell{YellowGreen!20}$\mc O(1)$            & $1-2^{-\Omega(n)}$    & \SetCell{YellowGreen!20}No       & \SetCell{YellowGreen!20}$\Omega(1)$ \\
This work & \SetCell{YellowGreen!20}Single-qubit (Pauli)         & \SetCell{YellowGreen!20}$\mc O(1)$            & $1-2^{-\Omega(n)}$    & \SetCell{YellowGreen!20}No       & \SetCell{YellowGreen!20}$\Omega(1)$       
\end{tblr}

    \caption{Comparing this work to previous protocols certifying all or almost all target states.
    Green cells denote the best possible for the given column; note that it is impossible for all columns to be green simultaneously due to the adaptivity no-go theorem of \cite{guptaHeODonnell2025}.
    Accuracy $\eps$ and confidence $\delta$ are assumed to be constants in the quantities above.
    A $t$-robust protocol accepts states that are $t\eps$-close to the target states and rejects those which are $\eps$-far. A concurrent update to a work of Du et al.~\cite{du2025certifying}, which we discuss below, also achieves the desiderata that we achieve in this work.} 
    \label{fig:comparing-protocols}
\end{figure}

\paragraph{Concurrent work.} Shortly before this work was posted, we were informed by the authors of \cite{du2025certifying} of an update to their paper, giving a test that also achieves constant robustness for almost all target states using Pauli measurements. We briefly elaborate on the differences between our respective tests (for the unfamiliar reader, these differences should be more understandable after reading our technical overview). 

In short, the test in \cite{du2025certifying} runs the \cite{Huang2025} test in a uniformly random Pauli basis without randomizing the location of the ``leftover'' qubit. Our test is even simpler (and its soundness was left as Conjecture 1 in \cite{CLSW26}): it runs the \cite{Huang2025} test in one of only two bases, either all-$Z$ or all-$X$, but randomizes the location of the leftover qubit. It is this two-basis picture that leads to the connections with uncertainty principles for total influences that we uncover here. As a result, the requirements on the oracle access to the target state needed for our test are in some sense lighter: we require queries to the amplitudes of 1-qubit states post-selected in only two bases, whereas the test of \cite{du2025certifying} requires an oracle that can answer such queries in $3^{n-1}$ potential Pauli bases.

The authors of \cite{du2025certifying} also make progress on one of the questions that we mention in \Cref{sec:outlook}, namely finding examples of natural families of states that are in the ``good'' set of target states: they show that random graph states, a subclass of stabilizer states, are in the ``good'' set.

\subsection*{Acknowledgements and AI use}
A.C. is thankful for support from the Google Research Scholar program.
J.L. acknowledges support from the National Science Foundation through the NSF AI Institute for Foundations of Machine Learning (IFML), under Award No. CCF-2505865, and the Institute for Foundations of Data Science (IFDS), under Award No. 2023166.
ChatGPT 5.4--6 was consulted throughout the research process to strategize proof approaches and conduct extensive numerical explorations.

\section{Technical Overview}
\subsection{The test}
We first describe the (slightly adaptive) single-shot test that achieves the guarantees of \Cref{thm:informal-1}. We will then describe the nonadaptive version that achieves the guarantees of \Cref{cor:1}.
\begin{minorprotocol}[One-shot (slightly adaptive) robust state certification]
\label{prot:adaptive}
\textbf{Input:} One copy of an $n$-qubit lab state $\rho$; and access to the $n$-qubit target state $\ket{\psi}$ (via the access model described in \Cref{sec:access-model}).
\medskip

\textbf{1. Sample a measurement.} Sample the following:
    \begin{itemize}[itemsep=0pt, topsep=6pt]
        \item A basis $\mathsf B\in\{\mathsf X,\mathsf Z\}$ with equal probability,
        \item A holdout qubit $j\in[n]$ with uniform probability.
    \end{itemize}
\medskip

\noindent \textbf{2. Measure the bulk of the lab state.} Measure all the qubits $k\neq j$ in basis $\mathsf B$, obtaining outcome string $b\in\{0,1\}^{n-1}$ and 1-qubit post-measurement state $\rho_b^{(j,\mathsf B)}$ on the holdout qubit $j$.
\medskip

\noindent \textbf{3. Compute target.} Compute (a classical description of) $\ket{\psi_{b}^{(j,\mathsf B)}}$, the 1-qubit state on the holdout qubit $j$ resulting from measuring qubits $k\neq j$ of the target state $\ket{\psi}$ in basis $\mathsf B$ and obtaining outcome $b$. (The classical description of $\ket{\psi_{b}^{(j,\mathsf B)}}$ is computed via the access model described in \Cref{sec:access-model}).
\medskip 

\noindent\textbf{4. Measure the holdout qubit.}
Perform the projective measurement $\big\{\ketbra{\psi_{b}^{(j,\mathsf B)}}{\psi_{b}^{(j,\mathsf B)}}, \mathbf 1 - \ketbra{\psi_{b}^{(j,\mathsf B)}}{\psi_{b}^{(j,\mathsf B)}}\big\}$ on $\rho_b^{(j,\mathsf B)}$.
Accept if the first outcome is obtained.
\end{minorprotocol}
What we show is that this protocol accepts with probability at least $\bra{\psi}\rho\ket{\psi}$ (\textit{i.e.}\ the fidelity between $\rho$ and $\psi$), and rejects with probability at least (approximately) $\frac{1-\bra{\psi}\rho\ket{\psi}}{C}$ for some constant $C>1$ (\textit{i.e.}\ the infidelity up to a constant factor).
One can readily convert this one-shot protocol to one that distinguishes, with any desired success probability $1-\delta$, a lab state $\rho$ that has $1-\frac{\epsilon}{C'}$ fidelity with $\ket{\psi}$ from a $\rho$ that has at most $1-\epsilon$ fidelity with $\ket{\psi}$ (for some constant $C'>1$ closely related to $C$). The latter protocol can be obtained via a standard sequential repetition using $O\left( \epsilon^{-2} \log (1/\delta)\right)$ copies of $\rho$, which is information-theoretically optimal even among protocols that are allowed to make arbitrary entangled measurements.

The only source of adaptivity in the resulting protocol comes from Step 4 of the one-shot protocol, where the measurement basis depends on the previous outcome $b$. This adaptivity can be removed via classical shadow tomography, as it is done in \cite{Huang2025}. In short, the adaptive measurement on the holdout qubit $j$ in Step~4 can be replaced by a non-adaptive measurement in a uniformly random choice of the $X$, $Y$, or $Z$ basis. This allows one to compute (by classical post-processing) an unbiased estimator for the outcome of any adaptively-chosen measurement on the holdout qubit. This process is repeated $T = O\left( \epsilon^{-2} \log (1/\delta)\right)$ times (using that many copies of $\rho$), and the protocol accepts if the average of the $T$ estimators is above the appropriate threshold. Note that the holdout qubit, as well as the post-measurement state, in general changes in each of the $T$ repetitions, but this is fine: each of the $T$ runs is an unbiased estimator of the one-shot acceptance probability (and since the runs are independent, we can invoke Chernoff). We refer to \cite{Huang2025} for more details about this application of classical shadows.
The nonadaptive protocol is described below.
The constant $C>2$ in \Cref{prot:main} is universal and independent of $n$.
In our proofs, the constant we establish is only implicit.
However according to numerical experiments, the tight constant appears to be $C < 5.5$.

\begin{protocol}[Non-adaptive robust state certification]
\label{prot:main}
\textbf{Parameters:} $\epsilon, \delta>0$.
\medskip

\textbf{Input:} $T = \mc O\left( \epsilon^{-2} \log (1/\delta)\right)$ copies of an $n$-qubit lab state $\rho$; and access to the $n$-qubit target state $\ket{\psi}$ (via the access model described in \Cref{sec:access-model}).

\medskip
\textbf{1. Sample a measurement.} Choose:
    \begin{itemize}[itemsep=0pt, topsep=6pt]
        \item A basis $\mathsf B\in\{\mathsf X,\mathsf Z\}$ with equal probability,
        \item A holdout qubit $j\in[n]$ with uniform probability,
        \item A holdout basis $\mathsf{C}\in\{\mathsf{X},\mathsf{Y},\mathsf{Z}\}$
        with uniform probability.
    \end{itemize}
\medskip

\noindent \textbf{2. Measure the lab state.}
Measure the lab state $\rho$ in the following single-qubit product basis: all qubits $k\neq j$ in basis $\mathsf B$ and the holdout qubit $j$ in basis $\mathsf C$.
    This yields an outcome bit  $c\in\{0,1\}$ for the holdout qubit and an outcome string $b\in\{0,1\}^{n-1}$ for the rest of the qubits.
\medskip

\noindent \textbf{3. Classical postprocessing.} Let $\rho_b^{(j,\mathsf B)}$ and $\ket{\psi_{b}^{(j,\mathsf B)}}$ be as in Protocol 1. Compute the classical shadow tomography estimate $\omega$ for 
$$ \bra{\psi_{b}^{(j,\mathsf B)}} \rho_b^{(j,\mathsf B)} \ket{\psi_{b}^{(j,\mathsf B)}}$$
corresponding to basis $\mathsf C$ and outcome $c$.

\medskip
\noindent \textbf{4. Repeat.} Repeat steps 1--3 $T$ times to obtain estimates $\omega_1, \ldots, \omega_T$. Let $\hat{\omega} = \frac{1}{T} \sum_{i=1}^T \omega_i$. Accept if $\hat{\omega} \geq 1- \frac{3\epsilon}{2C}$.
\end{protocol}
The protocol above accepts with probability at least $1-\delta$ if $\bra{\psi} \rho \ket{\psi} > 1-\frac{\epsilon}{C}$, and rejects with probability at least $1-\delta$ if $\bra{\psi} \rho \ket{\psi} < 1-\epsilon$. The latter follows in a straightforward way from the guarantees of Protocol 1 and the fact that classical shadow tomography gives an unbiased estimator for the true overlap.

\subsection{Access model}
\label{sec:access-model}
As in earlier protocols \cite{Huang2025,guptaHeODonnell2025} (henceforth HPS and GHO) we assume classical access to an oracle-based description of the target state $\ket{\psi}$. The access model that we work with is, in terms of strength,  between that of HPS and GHO.
In particular, we require knowledge of amplitudes of the target state in \emph{two} bases: standard and Hadamard.
For comparison, HPS only requires amplitudes in the standard basis, whereas GHO require amplitudes in all product bases (but is capable of certifying \emph{all} target states, rather than ``almost all'').

We acknowledge that our access model creates a bit of a tension: Haar-random states are highly entangled and exponentially complex, whereas states for which this kind of amplitude access can be simulated efficiently are generally more structured and tractable. For example, tensor network states with polynomial bond dimension allow efficient amplitude queries but are low-entangled, and neural quantum states typically support efficient queries in only one basis. This means that our test will generally require inefficient classical post-processing. 
We note that, while GHO require an even stronger access model, their certification works for all states (which thus includes classes of states for which the access model is efficiently reproducible). That said, we find the existence of a robust test that uses few-qubit measurements surprising even from a purely information-theoretic standpoint.
\begin{definition}
\label{def:oracle}
We consider an oracle model which accepts queries in the format $(b,\mathsf{B})$ for $b\in\{0,1\}^n$ and $\mathsf{B}\in\{\mathsf{std},\mathsf{had}\}$\,, and returns the complex number
\[\mathsf{O}_\psi(b,\mathsf B)=\begin{cases}
\braket{b|\psi} & \text{if } \mathsf{B}=\mathsf{std}\\[0.5em]
\braket{b|H^{\otimes n}|\psi} & \text{if } \mathsf{B}=\mathsf{had}\,.
\end{cases}\]
\end{definition}

This access model allows us to explicitly and efficiently learn the leftover state on the last qubit, conditioned on the outcome of any measurement on $\ket{\psi}$ used in our protocols.

\subsection{Uncertainty principles for (weighted) total influence}
\label{sec:ups}
The correctness of our protocol relies centrally on a new uncertainty principle for a certain weighted generalization of the total influence of Boolean functions.
This uncertainty principle may be of independent interest, so we describe it here in self-contained form.
The connection to our protocols is explained in the next subsection, \Cref{sec:proof-ideas-reduction}.

Before stating the full theorem, let us record a simple unweighted version of the phenomenon.
For any function $h:\{0,1\}^n\to \mathbb{C}$, define its total influence via
    \[\Inf[h]=\sum_{i,x}\left|\frac{h(x)-h(x+i)}{2}\right|^2\,,\]
    where here and throughout, for an index $i \in [n]$, we use the shorthand notation $x +i$ to denote the string where we flip the $i$-th bit of $x$. 
    Then we have the following.

\begin{prop}[Uncertainty principle for unweighted total influence]
\label{thm:baby-uncertainty-intro}
    For any $h:\{0,1\}^n\to\C$ with $\E_x|h(x)|^2=1$,
    \[\Inf[h]+\Inf[\wh h] \gtrsim n\,.\]
\end{prop}
\noindent Here and throughout, for quantities $A$ and $B$, the notation $A \gtrsim B$ means that there exists a universal constant $c>0$ such that $A \geq c B$.
Also, the Fourier transform $\,\wh{\cdot}\,$ will throughout bear the unitary normalization  $2^{n/2}$.
That is, for $s\in\{0,1\}^n$ we take the convention
\[\wh f(s) =2^{-n/2}\tsum_{x\in\{0,1\}^n}f(x)(-1)^{\sum_ix_is_i}\,.\]

The total influence, $\Inf$, is a central quantity in the analysis of Boolean functions and plays an important role throughout theoretical computer science~\cite{ODonnell2014}.
Despite this, an uncertainty principle for $\Inf$ does not appear to be well-known within the TCS community.
The proof of \Cref{thm:baby-uncertainty-intro} is simple however, and can be obtained for example by specializing a graph-theoretic uncertainty principle due to Benedetto and Koprowski \cite{7148912} to the hypercube graph.
For completeness, we give a direct proof in \Cref{sec:toy-case}.
As elaborated in \cite{7148912}, this influence uncertainty principle is the natural analogue to the additive Heisenberg uncertainty principle.

The correctness of our test reduces to the following \emph{weighted} generalization of the influence-uncertainty phenomenon.
For any probability measure $\mu$ on $\{0,1\}^n$ and any $h:\{0,1\}^n\to \C$, define the $\mu$-weighted total influence to be
\[\Inf_\mu[h]=\sum_{i,x}\frac{2\mu(x)\mu(x+i)}{\mu(x)+\mu(x+i)}\left|\frac{h(x) - h(x+i)}{2}\right|^2\,.\]
With this definition, we have the following uncertainty principle for $\Inf_\mu$.

\begin{theorem}[Uncertainty principle for weighted total influence]
\label{thm:influences-informal}
With probability $1-2^{-\Omega(n)}$ over a Haar-random $f:\{0,1\}^n\to \C$ with $\|f\|_2=1$, for all $g:\{0,1\}^n\to \C$ with $\|g\|_2=1$ and $\langle f,g\rangle=0$, 
\[\Inf_{\mu_f}\left[\frac{g}{f}\right]+\Inf_{\mu_{\wh f}}\left[\frac{\wh g}{\wh f}\right] \gtrsim n\,.\]
Here $\mu_f$ (resp. $\mu_{\wh f}$) is defined by $\mu_f(x)=|f(x)|^2$ (resp. $\mu_{\wh f}(s)=|\wh f(s)|^2$).
\end{theorem}

It is worth mentioning that $\Inf_\mu$ is precisely the Dirichlet energy associated with continuous-time Glauber dynamics for the stationary distribution $\mu$.
See, for example, \cite[\S 3.3.2 and \S 13.2.1]{LevinPeres2017}.
In comparison to the unweighted version, the uncertainty principle of \Cref{thm:influences-informal} seems to demand a dramatically more involved proof, spanning Sections \ref{sec:concentration}--\ref{sec:finishing}.

\subsection{Proof ideas I: From state certification to uncertainty principles}
\label{sec:proof-ideas-reduction}
As mentioned earlier, the crux is to analyze the \Cref{prot:adaptive}, \textit{i.e.}\ the one-shot protocol. This will be our focus for the rest of the technical overview.





\medskip
Notice that the acceptance probability in \Cref{prot:adaptive} can be written as follows:

\[\Pr[\mathsf{accept}]=\frac12\left(\mathop{\E}_{\substack{j\sim[n]\\x\sim \mathrm{Meas}(\rho, \mathsf X, [n]\backslash j)}}\braket{\psi_x^{(j,\mathsf X)}|\rho_x^{(j,\mathsf X)}|\psi_x^{(j,\mathsf X)}}
    \;\;+ \mathop{\E}_{\substack{j\sim[n]\\z\sim \mathrm{Meas}(\rho, \mathsf Z,[n]\backslash j)}}\braket{\psi_z^{(j,\mathsf Z)}|\rho_z^{(j,\mathsf Z)}|\psi_z^{(j,\mathsf Z)}}\right),\]
where $\ket{\psi_b^{(j,\mathsf B)}}$ and $\rho_b^{(j,\mathsf B)}$ for $B \in \{\mathsf X, \mathsf Z\}$ and $b \in \{0,1\}^{n-1}$ are defined as in \Cref{prot:adaptive}: they are the post-measurement state of the $j$-th qubit when measuring the rest of the qubits of, respectively, $\ket{\psi}$ and $\rho$ in basis $\mathsf B$ and obtaining outcome $b$. And $\mathrm{Meas}(\rho, \mathsf B, [n]\backslash j)$ is the corresponding distribution of over the outcome $b$ (when measuring $\rho$).
    
It is straightforward to see that $\Pr[\mathsf{accept}]\geq \bra{\psi}\rho\ket{\psi}$:  we can equivalently write 
$\Pr[\mathsf{accept}] = \text{Tr}[A \rho]$
for an appropriate positive-semidefinite operator $A$, which has $\ket{\psi}$ as a $+1$ eigenvector (see \Cref{sec:4} for the definition of $A$). Then, it follows that $\text{Tr}[A\rho] \geq \bra{\psi}\rho\ket{\psi}$. The bulk of the work is to prove an \emph{upper bound} on $\Pr[\mathsf{accept}]$ in terms of fidelity.
Ultimately, we will show the following.
\begin{theorem}
\label{prop:ortho-pure-UP}
    With all but $2^{-\Omega(n)}$ probability over a Haar-random target state $\ket{\psi}$, for all lab states $\rho$,
    \[\mathop{\E}_{\substack{j\sim[n]\\x\sim \mathrm{Meas}(\rho, \mathsf X, [n]\backslash j)}}\hspace{-1em}\braket{\psi_x^{(j,\mathsf X)}|\rho_x^{(j,\mathsf X)}|\psi_x^{(j,\mathsf X)}}
    \;\;+ \hspace{-1em}\mathop{\E}_{\substack{j\sim[n]\\z\sim \mathrm{Meas}(\rho, \mathsf Z,[n]\backslash j)}}\hspace{-1em}\braket{\psi_z^{(j,\mathsf Z)}|\rho_z^{(j,\mathsf Z)}|\psi_z^{(j,\mathsf Z)}}\;\;\leq\;\; 2-c \big(1-\braket{\psi|\rho|\psi}\!\big)\]
    Here $c>0$ is a universal constant.
\end{theorem}
Accounting for the factor of $\frac12$ in the test, we arrive at
\[\braket{\psi|\phi|\psi}\leq \Pr[\mathsf{accept}]\leq 1-\frac{c}{2}\big(1-\braket{\psi|\rho|\psi}\!\big)\,.\]


\paragraph{Reducing to a formulation on the hypercube.}
By a straightforward argument, it suffices to argue the case where $\rho=\ketbra{\phi}{\phi}$ is pure and orthogonal to $\ket{\psi}$ (the case of general $\rho$ can then be understood by writing $\rho$ in terms of the eigenbasis of the ``acceptance'' operator $A$ mentioned earlier, which includes $\ket{\psi}$ as a top eigenvector).
Our main technical result is the following.

\begin{prop}[Orthogonal pure state version]
\label{prop:ortho-pure-UP}
    With probability $1-2^{-\Omega(n)}$ over a Haar-random target state $\ket{\psi}$, for all pure states $\ket{\phi}$ orthogonal to $\ket{\psi}$,
    \[\mathop{\E}_{\substack{j\sim[n]\\x\sim \mathrm{Meas}(\ket{\phi}, \mathsf X, [n]\backslash j)}}|\braket{\psi_x^{(j,\mathsf X)}|\phi_x^{(j,\mathsf X)}}|^2
    + \mathop{\E}_{\substack{j\sim[n]\\z\sim \mathrm{Meas}(\ket{\phi}, \mathsf Z,[n]\backslash j)}}|\braket{\psi_z^{(j,\mathsf Z)}|\phi_z^{(j,\mathsf Z)}}|^2\leq 2-c,\]
    for a universal constant $c>0$ independent from $n$.
\end{prop}

Let us define the functions $f,g:\{0,1\}^n\to \C$ via
\[\ket{\psi}=\sum_{x\in\{0,1\}^n}f(x)\ket{x}\qquad\text{and}\qquad\ket{\phi}=\sum_{x\in\{0,1\}^n}g(x)\ket{x}\,.\]
For convenience, we will also switch to the natural notation $\ket{f} = \ket{\psi}$, and $\ket{g} = \ket{\phi}$.

To formulate \Cref{prop:ortho-pure-UP} in terms of $f$ and $g$ (and their Fourier transforms), we can define the ``edgewise averages''
\begin{align*} \Avg_f[g]&:=\frac{1}{2n}\sum_i\sum_x\frac{\Big|\overline{f(x)}g(x)+\overline{f(x+i)}g(x+i)\Big|^2}{|f(x)|^2+|f(x+i)|^2}\,\in[0,1],\\
\Avg_{\wh f}[\wh g]&:=\frac{1}{2n}\sum_i\sum_s\frac{\left|\overline{\wh f(s)}\wh g(s)+\overline{\wh f(s+i)}\wh g(s+i)\right|^2}{|\wh f(s)|^2+|\wh f(s+i)|^2}\;\;\in[0,1]\,,
\end{align*}
noting that 
\[\Avg_f[g]=\hspace{-1em}\mathop{\E}_{\substack{j\sim[n]\\x\sim \mathrm{Meas}(\ket{\phi}, \mathsf X, [n]\backslash j)}}\hspace{-1em}|\braket{\psi_x^{(j,\mathsf X)}|\phi_x^{(j,\mathsf X)}}|^2
\qquad\text{and}\qquad\Avg_{\wh f}[\wh g]=\hspace{-1em}\mathop{\E}_{\substack{j\sim[n]\\z\sim \mathrm{Meas}(\ket{\phi}, \mathsf Z,[n]\backslash j)}}\hspace{-1em}|\braket{\psi_z^{(j,\mathsf Z)}|\phi_z^{(j,\mathsf Z)}}|^2.\]
(The factor of $\frac12$ in $\Avg_f[g]$ and $\Avg_{\wh f}[\wh g]$ is because the sum double counts each inner product.)
With these definitions, \Cref{prop:ortho-pure-UP} is equivalent to the following ``uncertainty principle'', which, as we will show along the way, is equivalent to the uncertainty principle for weighted influences of \Cref{thm:influences-informal}.
\begin{theorem}
    \label{thm:full-uncertainty}
    With probability $1-2^{-\Omega(n)}$ over a Haar-random $f:\{0,1\}^n\to \C$, $\|f\|_2=1$,
    for all $g:\{0,1\}^n\to \C$, $\|g\|_2=1$ with $\langle f,g\rangle=0$,
    \[\Avg_f[g]+\Avg_{\wh f}[\wh g]\leq 2-c,\]
    where $c>0$ is a universal constant independent of $n$.
\end{theorem}

Before explaining our approach and relating this inequality to other types of uncertainty, we pause to note one final formulation that will be used several times.
The averages can be expressed as quadratic forms in $g$, namely as
\[\Avg_f[g]=:\langle g, M_f g\rangle\]
for an appropriate matrix parameterized by $f$ (this is the matrix corresponding to the ``acceptance'' operator $A$ mentioned earlier).
Similarly,
\[\Avg_{\wh f}[\wh g]=\langle \wh g, M_{\wh f} \wh g\rangle= \braket{g|H^{\otimes n}M_{\wh f}H^{\otimes n}|g}=:\langle g, \wh{M_{\wh f}} \,g\rangle\,.\]
The uncertainty principle thus has the following spectral interpretation: individually, the two operators $M_f$ and $\wh{M_{\wh f}}$ have an inverse-polynomial spectral gap:
\[\mathrm{Spec}(M_f)=\left\{1,1-\Theta\left(\tfrac1n\right),\ldots \right\}\qquad \text{and}\qquad \mathrm{Spec}(\wh{M_{\wh f}})=\left\{1,1-\Theta\left(\tfrac1n\right),\ldots \right\},\]
where the top eigenspace is in both cases one-dimensional with eigenvector $\ket{f}$. The latter follows directly from the analysis of the ``one-basis'' test of \cite{Huang2025} (and the fact that this analysis is tight can be seen, for instance, via an example described in~\cite[Appendix A]{CLSW26}). Yet, what we show is that \[\mathrm{Spec}(M_f+\wh{M_{\wh f}})=\{2,2-c,\ldots\},\]
for a constant $c>0$ independent of $n$.

As we present the proof ideas in the rest of the technical overview, we will arrive along the way at a reformulation of this uncertainty principle in terms of the total influences (as stated earlier in \Cref{thm:influences-informal}).

\subsection{Proof ideas II: Proving the uncertainty principles}
We explain our proof approach in three stages: an unweighted toy case, the phase state case, and the general case.
\subsubsection{Unweighted toy case: an uncertainty principle for total influence}
\label{sec:toy-case}
Consider the uniform unweighted case, where we replace each occurrence of $f(x)$ or $\wh f(s)$ with $2^{-n/2}$ (note that this is a toy case because such a function $f$ does not exist!). Then, instantiating the previous definition of $\Avg_f[g]$ with the constant function $f = 2^{-n/2}$ we have 
\[\Avg[g] =\frac1{4n}\sum_i\sum_x|g(x)+g(x+i)|^2\,.\]
A simple, but revealing, observation is that, since $\norm{g}_2 = 1$, we can rewrite $\Avg[g]$ as
\begin{equation}
\label{eq:avg-as-influence}
\Avg[g]=1-\frac1n\sum_{i,x}\left(\frac{g(x)-g(x+i)}{2}\right)^2=1-\frac1n\Inf_{\mathrm{count}}[g] \,,
\end{equation}
where $\Inf_{\mathrm{count}} =\Inf[2^{n/2}\cdot g]$ is the \textit{total influence} of $g$ under the counting measure (which is the correct scaling for $g$).
The identity \eqref{eq:avg-as-influence} can be straightforwardly checked; for the proof of a much more general identity, we refer the reader to \Cref{app:identity}.

In this toy picture, the desired uncertainty principle for $\Avg$ is then, for a universal $c>0$,
\[\Avg[g]+\Avg[\wh g]\leq 2-c \iff \Inf_{\mathrm{count}}[g]+\Inf_{\mathrm{count}}[\wh g]\geq cn.\]
It turns out this uncertainty principle is true and for example follows from a graph-theoretic uncertainty principle of \cite{7148912} specialized to the hypercube.
We give a short, self-contained proof next.

\begin{prop}[Uncertainty principle for total influence]
\label{thm:influence-up}
    Let $g:\{0,1\}^n\to\C$ with $\|g\|_2=1$.
    Then
    \begin{equation}
    \label{eq:influence-up}
        \Inf_{\mathrm{count}}[g] + \Inf_{\mathrm{count}}[\wh g] \geq \left(1-\tfrac{1}{\sqrt{2}}\right)n
    \end{equation}
    Moreover, this inequality is sharp.
\end{prop}
\noindent
Note that it is a standard fact that 
\begin{equation}
\label{eq:influence}
    \Inf_{\mathrm{count}}[g] =  \sum_s |s| \cdot |\wh g(s)|^2 \; ,
\end{equation}
and thus \Cref{thm:influence-up} is equivalent to
\begin{equation}
\label{eq:influence-up-diag}
    \sum_s |s| \cdot |\wh g(s)|^2 +\sum_x|x|\cdot |g(x)|^2 \geq \left(1-\tfrac{1}{\sqrt{2}}\right)n\,.
\end{equation}
\begin{proof}[Proof of~\Cref{thm:influence-up}]
    Let $D=\ketbra{1}{1}=\left(\begin{smallmatrix}
        0 & 0 \\ 0 & 1
    \end{smallmatrix}\right)$.
    Inspecting \eqref{eq:influence-up-diag}, one observes that
    \begin{align*}
    \Inf_{\mathrm{count}}[g] &= \bra{g}H^{\otimes n}\Big(\sum_{j=1}^nI^{\otimes j-1}\otimes D\otimes I^{n-j}\Big) H^{\otimes n}\ket{g}\\
    \text{and}\qquad\Inf_{\mathrm{count}}[\wh g] &= \bra{g}\sum_{j=1}^nI^{\otimes j-1}\otimes D\otimes I^{n-j}\ket{g}\,,
    \end{align*}
    where $I=\left(\begin{smallmatrix}
        1 & 0 \\ 0 & 1
    \end{smallmatrix}\right)$ and  $H=\left(\begin{smallmatrix}
        1 & 1\\1& -1
    \end{smallmatrix}\right)/\sqrt{2}$ is the $2\times 2$ Hadamard matrix.
    The best constant in \eqref{eq:influence-up} is therefore given by the minimum eigenvalue of the operator
    \[A:=\tsum_{j=1}^nI^{\otimes j-1}\otimes(\underbrace{HDH+D\vphantom{\Big|}}_{:=M})\otimes I^{\otimes n-j}=\sum_{j=1}^nI^{j-1}\otimes M\otimes I^{n-j}\,.\]
    Let $\{\ket{u},\ket{v}\}$ be an eigenbasis for $M$ and note that the $(I^{j-1}\otimes M\otimes I^{n-j})$'s are simultaneously diagonalized by the eigenbasis $\{\ket{u},\ket{v}\}^{\otimes n}$.
    As a result, the spectrum of $A$ comprises $n$-term sums of elements of $\mathrm{Spec}(M)$, and in particular
    \[\lambda_{\min}(A) = n\lambda_{\min}(M)\,.\]
    The last quantity $\lambda_{\min}(M)$ is determined by direct calculation.
\end{proof}
We remark that this influence uncertainty principle, especially in its formulation in \eqref{eq:influence-up-diag}, is in some sense the right hypercube analogue of the Heisenberg uncertainty principle.
This is elaborated in \cite{7148912}.

\subsubsection{The phase state case and some key ideas}
Let us now return to the $f$-weighted quantities $\Avg_f[\,\cdot\,]$ and $\Avg_{\wh f}[\,\cdot\,]$, where recall that the goal is to show
 \[\Avg_f[g]+\Avg_{\wh f}[\wh g]\leq 2-c\,,\]
for some universal constant $c>0$.
In this subsection, we will focus on the crucial special case of $f$ a \textit{phase state}, \emph{i.e.,} $|f(x)|^2= 2^{-n}$ for all $x$.
Note that, while in the toy case the uncertainty principle holds for all $g$ with $\|g\|_2=1$, once we move to the $f$-weighted case, we must restrict our attention to $g$ that are orthogonal to $f$, as in the hypothesis of \Cref{thm:full-uncertainty} (and this restriction is necessary since, for example, $\Avg_f[f]+\Avg_{\wh f}[\wh f] = 2$).
Since we seek a theorem holding w.h.p. over a random $f$, it is natural to attempt to prove concentration of the operator norm of $M_f+\wh{M_{\wh f}}$ when restricted to the subspace of allowed $g$'s.
Unfortunately, the set of valid $g$'s is now a ``moving target'' (moving with $f$), and this seems to frustrate naive spectral approaches, such as those used in the toy unweighted case.

The first key idea in this work is a simple transformation that both decouples $f$ from $g$ and recovers a connection to total influence.
Specifically, we factor out a copy of $f$ from $g$, which is possible with probability $1$ over Haar-random $f$ and always possible for $f$'s coming from phase states (the latter being our current focus).
Formally, define $q:\{0,1\}^n\to\C$ via $g = fq$,
and note that
\[\|g\|_2^2=1\iff \sum_x|f(x)|^2|q(x)|^2=\|q\|_{L^2(\unif)}^2=1 \,,\]
and
\[\langle f,g\rangle=0\iff \sum_x|f(x)|^2q(x)=\E_{x\sim\unif}q=0,\]
where we have used that $|f(x)|^2= 2^{-n}$ for all $x$.
Recall that in the formulation of the uncertainty principle from \Cref{thm:full-uncertainty}, we care precisely about functions $g$ satisfying the two conditions on the LHS above.
Thus the class of $g$ ``compatible'' with $f$ is exactly those $g=fq$ where $q$ is nothing but a balanced Boolean function, with standard $L^2$ normalization with respect to the uniform probability measure on $\{0,1\}^n$.

Moreover, with this definition of $q$,
\begin{align}
\Avg_f[g] &= \frac{1}{2n}\sum_i\sum_x\frac{\left||f(x)|^2q(x)+|f(x+i)|^2q(x+i)\right|^2}{|f(x)|^2+|f(x+i)|^2}\nonumber\\
&=\frac{1}{4n}\sum_i\mathop{\E}_{x\sim\unif}\left|q(x)+q(x+i)\right|^2\nonumber\\
&= 1 - \frac{1}{n} \sum_i \mathop{\E}_{x \sim \unif} \left( \frac{|q(x) - q(x + i)|}{2} \right)^2 \nonumber \\
&= 1-\tfrac1n\Inf[q],
\label{eq:inf-appearance-phase}
\end{align}
where $\Inf$ is the standard total influence.
So we have conveniently recovered, for the standard-basis side quantity $\Avg_f[g]$, an analogous expression to the earlier \Cref{eq:avg-as-influence}, up to the correct normalization for the influence, and the passage from $g$ to $q$.

One may then hope the rest of the story matches, and we recover $\Inf[\wh q]$ on the Fourier side.
Unfortunately this is not what happens; making this same substitution into $\Avg_{\wh f}[\wh g]$ does \textit{not} yield $\Inf[\wh q]$ but the following expression:
\begin{equation}
\label{eq:ungainly-expr}
\Avg_{\wh f}[\wh{fq}]=\frac{1}{2n}\sum_i\sum_s\frac{\left|\overline{\wh f(s)}\widehat{fq}(s)+\overline{\wh f(s+i)}\widehat{fq}(s+i)\right|^2}{|\wh f(s)|^2 + |\wh f(s+i)|^2}\,.
\end{equation}
In dual fashion to the standard basis side, this can be simplified to $1-\frac{1}{n}\Inf\left[\wh g/\wh f\right]$, but $\wh g/\wh f$ has no simple algebraic relationship to $\wh q$, and we thus cannot pass to $1-\frac{1}{n}\Inf\left[\wh q\right]$.

Despite this, it turns out that one can (with significant effort) show the following concentration inequality: for essentially\footnote{We argue this concentration for all $q$ satisfying certain technical regularity conditions.
This suffices for the argument that follows.} any fixed $q$, we have
\begin{equation}
    \label{eq:fixed-q-conc-intro}
    \Pr_{f}\left[\Avg_{\wh f}[\wh{fq}]\geq \frac12 + t\right]\leq \exp(-ct^62^{n/3})\,.
\end{equation}
Note this observation is not immediately useful because the quantifiers are in the wrong order: recall that we want a statement asserting that with high probability over $f$, it holds that for all $q$, $\Avg_{\wh f}[\wh{fq}]$ is well-behaved, whereas in \eqref{eq:fixed-q-conc-intro} we are fixing a $q$ before selecting a random $f$.
But we \textit{could} hope to correct the quantifier order via a union bound over an $\eps$-net for the $q$'s, provided this net is not too large.

What set of $q$'s should we use?
The covering number of the set of \textit{all} $q$'s is unfortunately too large: $\exp(\Omega(2^n))$ for $\epsilon = \mc O(1)$.
But recall that the standard-basis quantity $\Avg_f[g]$ is directly related to the total influence of $q$---and it is well-known that the set of \emph{low-influence} $q$'s has small covering number.
Connecting influence to the concentration phenomenon in \eqref{eq:fixed-q-conc-intro} is the second key idea of this work, and we do so via a win-win argument.
In the first case, $\Inf[q]$ is large, which means $\Avg_f[g]$ is small and we are already done.
In the second case, $\Inf[q]$ is small and the $\eps$-net argument succeds in switching quantifiers in \eqref{eq:fixed-q-conc-intro}.
In more detail, we argue as follows. Fix a sufficiently small $\kappa >0$.
\begin{enumerate}
    \item Let $f$ be a random phase function.
    Let $g$ be compatible with $f$ (\textit{i.e.}\ $\|g\|=1$ and $\langle f,g\rangle=0$), and let $q=g/f$.
    \item If $\Inf[q]> \kappa n$ for some $\kappa >0$, then, by \eqref{eq:inf-appearance-phase}, $\Avg_f[g] < 1-\kappa$ and we are done: $\Avg_f[g]+\Avg_{\wh f}[\wh g]< 2-\kappa$.
    \item Otherwise, $\Inf[q]\leq \kappa n$.
    \begin{enumerate}
        \item Since $\Inf[q] = \sum_s |s| \cdot |\wh{q}(s)|^2$, small influence forces $q$ to be approximately a \emph{low-degree} polynomial.
    By a straightforward argument, there is a $0.01$-net (with respect to $L^2$ distance) $Q_\kappa^{(0.01)}$ of such $q$ with cardinality $|Q_\kappa^{(0.01)}|\leq \exp(c\cdot2^{\mc O(\sqrt{\kappa}n)})$.
    \item For $\kappa$ sufficiently small, the concentration in \eqref{eq:fixed-q-conc-intro} beats the $0.01$-net cardinality $|Q_\kappa^{(0.01)}|$.
    Thus by a union bound, w.h.p. over $f$, all $q^*\in Q_\kappa^{(0.01)}$ have $\Avg_{\wh f}[\wh{fq^*}]\leq 0.51$.
    Moreover, $\Avg_{\wh f}[\wh{f\,\cdot\,}]$ can easily be seen to be 1-Lipschitz, so we conclude that w.h.p. over $f$,
    \[\Avg_{\wh f}[\wh{fq}]\leq 0.52\text{ for all } q\text{ with }\Inf[q]\leq \kappa\,.\]
    Thus in this case, w.h.p. over $f$, $\Avg_f[g]+\Avg_{\wh f}[\wh g]\leq 1.52$.
    \end{enumerate}
    \item Putting things together, we conclude that w.h.p. over $f$, $\Avg_f[g]+\Avg_{\wh f}[\wh g]\leq \max\{2-\kappa, 1.52\}$.
\end{enumerate}
\noindent Although this argument is not strictly necessary for our general $f$ result, we include a formal version of it as an interlude between \Cref{sec:concentration,sec:AB-split}. For the case of general $f$ to be outlined shortly in \Cref{sec:full-monty}, several more obstacles will be lurking, but the general shape of the argument is shared.


\paragraph{Concentration bounds for $\Avg_{\wh f}[\wh{fq}]$: Proof of~\Cref{eq:fixed-q-conc-intro}.}
We now briefly describe the ideas behind the proof of the concentration bound in \eqref{eq:fixed-q-conc-intro}:
\begin{equation*}
    \Pr_{f}\left[\Avg_{\wh f}[\wh{fq}]\geq \frac12 + t\right]\leq \exp(-ct^62^{n/3})\,.
\end{equation*}
The main technical challenge is that (for a fixed $q$) while the numerator of $\Avg_{\wh f}[\wh{fq}]$, as in Equation~\eqref{eq:ungainly-expr}, is a quartic polynomial in the entries of $f$---and thus amenable to concentration of polynomial chaos arguments---the presence of the denominators seems to frustrate all off-the-shelf concentration approaches.
Moreover, the exponentially many terms in the sum are all mutually dependent (albeit weakly), and naive bounded difference (Efron--Stein) and Herbst-type arguments also do not work because $\Avg_{\wh f}[\wh{f q}]$ is highly non-differentiable in the $f(x)$'s.

Ultimately we overcome this barrier as follows:
we smooth each summand of ~\Cref{eq:fixed-q-conc-intro} by adding a carefully chosen term $\tau>0$ to the denominator and separately control both the smoothed version $\Avg^{(\tau)}_{\wh f}[\wh{fq}]$ and the error term $E=\big|\Avg^{(\tau)}_{\wh f}[\wh{fq}]-\Avg_{\wh f}[\wh{fq}]\big|$.
This fixes the differentiability of the average, but seems to simply move the problem into the $E$ term.
However, we show that $E$ is upper-bounded by another quantity that \textit{is} differentiable and indeed concentrates very strongly around 0.
This general perspective of smoothing and then differentiably bounding the smoothing error will also be crucially utilized in the Haar random setting that we discuss shortly.

\medskip
Overall, we have completed the description of the main ideas that go into to proving \Cref{thm:full-uncertainty} for the restricted class of $f$ corresponding to phase states. Restated as an uncertainty principle for total influences, we have that, w.h.p.\ over phase functions $f$, $\|f\|_2=1$,
    for all $g:\{0,1\}^n\to \C$, $\|g\|_2=1$ with $\langle f,g\rangle=0$, 
\begin{equation}
\label{eq:influence-uncertainty-principle}
\Inf\left[\frac{g}{f}\right] + \Inf\left[\frac{\wh g}{\wh f}\right] \geq C \cdot n \,, 
\end{equation}
for a universal constant $C>0$.

\medskip
Looking ahead to the full Haar-random $f$ case, we will only be able to use the above win-win argument on a ``good part'' of the hypercube where $|f(x)|^2$ is sufficiently close to its expectation (notice that in the phase state case we do not have to worry about this!), and we will have to handle the rest of the functions by other means.
We explain this in the next subsection.

\subsubsection{The full monty:  Haar-random states and the good-bad partition of $\{0,1\}^n$}
\label{sec:full-monty}
We now turn our attention to the case of a Haar random $f$. Let us see what happens when we factor an $f$ out of $g$ for general $f$ (not necessarily with $|f|^2\equiv 2^{-n}$).
Similarly to the above, we have
\[\|g\|_2^2=1\iff \sum_x|f(x)|^2|q(x)|^2=1\]
and
\[\langle f,g\rangle=0\iff \sum_x|f(x)|^2q(x)=0\,.\]
This suggests that for a given $f$ we should define the probability measure $\mu_f(x)=|f(x)|^2$, whereupon the compatibility requirements for $q$ can be written as
\[\|q\|_{L^2(\mu_f)}^2=1\qquad\text{and}\qquad\E_{x\sim\mu_f}q(x)=0.\]
Crucially, we see that in the $q$-world, compatibility is governed only by the moduli of the $f(x)$'s, not their arguments.
Thus from now on we will say that a $q$ is \emph{$\mu_f$-compatible} if it satisfies the two equations above.

This phenomenon carries through to $\Avg_f$. 
When viewed as a function of $q$, it only depends on $\mu_f$, and a straightforward calculation gives:
\begin{align*}\Avg_f[g] &= 1 - \frac{1}{2n}\sum_i\sum_x\frac{|f(x)|^2|f(x+i)|^2}{|f(x)|^2+|f(x+i)|^2}\left|\frac{g(x)}{f(x)}-\frac{g(x+i)}{f(x+i)}\right|^2\\[1em]
&= 1-\frac{1}{n}\sum_{i,x}\frac{2\mu_f(x)\mu_f(x+i)}{\mu_f(x)+\mu_f(x+i)}\left|\frac{q(x)-q(x+i)}{2}\right|^2\\
&=: 1-\frac1n\Inf_{\mu_f}[q]\,,
\end{align*}
where we have defined $\Inf_{\mu_f}[q]$ to match precisely the latter expression. Thus, $\Avg_f[g]$ is linearly related to a weighted version of the influence of $q$, where each edge is weighted by the harmonic mean of $\mu_f$ at each incident vertex.
What's more, since $f$ is Haar-random, by standard concentration inequalities, $\mu_f$ is approximately flat: $\mu_f(x) \approx 2^{-n}$ for most $x$
(if this approximation held uniformly for all $x$, then we would more or less be back in the phase state setting which we already understand).

Given the apparent dependence on the measure induced by $f$, rather than $f$ itself, we are naturally led to the following approach.
First, stratify the $f$'s by their measure; \textit{i.e.}, define
\[F_\mu=\{f:\mu_f=\mu\}\,.\]
Then, for each ``typical'' $\mu$, repeat the win-win argument from the phase-state case, using the randomness in $f\sim_\unif F_\mu$.
More concretely, for any $\mu$ we have two cases: first, if $\Inf_\mu[q] > \kappa n$ for a suitable constant $\kappa$, then we are done; second, if $\Inf_\mu[q]\leq\kappa n$, it is reasonable to try repeating the $\eps$-net argument from before.
And indeed, the fixed-$q$ and random-$f$ concentration bound \eqref{eq:fixed-q-conc-intro} holds in this adjusted setting where $f$ is drawn uniformly from $F_\mu$.

However, a problem appears when we try to take the union bound in the $\eps$-net argument.
It turns out that, unlike in the phase state case, a bound on $\Inf_\mu[q]$ does not, for typical $\mu$, meaningfully control the covering number of $\mu$-compatible $q$'s.
Roughly, the reason is that $\Inf_\mu[q]$ is not sensitive to changes in $q$'s value at the points $x$ with both: (\textit{i}) abnormally large mass $\mu(x)$ and (\textit{ii}) neighbors that are ``normal-sized.''
There are many such points in a typical $\mu$, and this allows one to ``hide'' a large $\eps$-net by varying the values of $q$ at such $x$.
To be clear, this is not a weakness in our analysis of the covering number; one really can show the existence of a $\Omega_\kappa(N)$-dimensional subspace of functions $q$, all of which have $\Inf_\mu[q]\leq\kappa n$---see \Cref{app:badspace}.

Our strategy for circumventing this barrier is to partition the domain into two sets: $A\sqcup B = \{0,1\}^n$, where $A$ represents the points where $\mu$ ``looks uniform'',
$A:=\{x:\mu(x)\in[1/(TN),T/N]\}$ for an appropriate $T>1$.
First we assume $\Inf_\mu[q]\leq \kappa n$ (otherwise we are done immediately), and then we then argue two facts:
\begin{itemize}
    \item \emph{$A$-part covering numbers.}
    The covering number of the $A$-part of $q$, \textit{i.e.,} the set of $q_A=q\cdot \mathbf{1}_A$ for $\mu$-compatible $q$, is suitably controlled by $\Inf_\mu[q]$. After some effort, this allows the $\eps$-net argument to go through for $q_A$, and we have: w.h.p over $f \in F_{\mu}$,  $\Avg_{\wh f}[\wh{fq_A}] \leq (1/2 + \text{small})\|q_A\|_{L^2(\mu)}$ for all $q_A$ corresponding to low-influence $q$.
    \item \emph{$B$-part operator bounds.}
    Second, for the ``bad'' part of $q$, $q_B=q\cdot\mathbf 1_B$ with $B=A^c$, we can afford to make a coarse bound because $|B|/N$ concentrates below a small constant of our choosing (depending on $T$).
In particular, defining $\wh{M_{\wh f}}$ by the relation $\Avg_{\wh f}[\wh{g}]=:\langle g,\wh{M_{\wh f}} \,g\rangle$, we show the operator norm of $\Pi_B\wh{M_{\wh f}}\Pi_B$ is small (close to 1/2) with high probability. Equivalently, $\langle g_B,\wh{M_{\wh f}} \,g_B \rangle$ is small with high probability.
The cross terms in \\$\Avg_{\wh f}[\reallywidehat{(g_A+g_B)}]=\langle g_A+g_B,\wh{M_{\wh f}} \,(g_A+g_B)\rangle$ are dispatched similarly, and are close to 1/4 with high probability.
\end{itemize}
Together, these bounds yield the following schematic argument in the low $\Inf_\mu[q]$ case:
\begin{align*}
    \Avg_{\wh f}[\wh g]&=\langle g, \wh{M_{\wh f}} \,g\rangle\\
    &= \langle g_A + g_B, \wh{M_{\wh f}} \,(g_A + g_B) \rangle \\
    &\leq \Avg_{\wh f}[\wh{fq_A}]+2\|\Pi_A\wh{M_{\wh f}}\Pi_B\|_\op \cdot \norm{g_A}_2 \norm{g_B}_2 + \|\Pi_BM_f\Pi_B\|_\op\|g_B\|_2^2\\
    &\leq (1/2+\text{small})\|q_A\|_{L^2(\mu)}^2 + (1/2+\text{small})\|g_A\|_2\|g_ B\|_2 + (1/2+\text{small})\|g_B\|_{2}^2\\
    &=(1/2+\text{small})\|g_A\|^2_2 + (1/2+\text{small})\|g_A\|_2\|g_B\|_2 + (1/2+\text{small})\|g_B\|_{2}^2 \\
    &= (\|g_A\|_2 \,,\, \|g_B\|_2) \, \begin{pmatrix}
    1/2 & 1/4\\
    1/4 & 1/2
\end{pmatrix} \,\begin{pmatrix}
    \| g_A\|_2\\
    \|g_B \|_2
\end{pmatrix}  + \text{small} \\
    &\leq 3/4+\text{small} \,.
\end{align*}
Here, the fourth line follows from the spectral bounds on $\Pi_AM_f\Pi_B$ and $\Pi_BM_f\Pi_B$ (plus the $A$-part analysis); the fifth line follows from the fact that $\|q_A\|_{L^2(\mu)} = \| g_A\|$; and the last line follows from the observations \[\left\|\left(\begin{smallmatrix}
    1/2 & 1/4\\1/4 & 1/2
\end{smallmatrix}\right)\right\|_\op= \frac34\qquad\text{and}\qquad\Big\|\big(\|g_A\|_2,\,\|g_B\|_2\big)\Big\|_2 = \| g\|_2 = 1\,.\]
Combining this bound with the high-influence side of the win-win argument, we conclude $\Avg_f[g]+\Avg_{\wh f}[\wh g]\leq \max\{2-\kappa,1.75 + \text{small}\}\leq 2-\Omega(1)$, as desired.

As anticipated, this uncertainty principle has a satisfying equivalent reformulation in terms of influences, as in \Cref{thm:influences-informal}. 

\medskip
We conclude the technical overview with some technical details about the $A$- and $B$-part analyses.

\paragraph{$A$-part analysis: from influence on $A$ to influence everywhere via linear extensions.}
There is a significant---but subtle---issue with completing the covering number argument for low-influence $q_A$'s.
Recall that in the win-win argument we are only allowed to assume that $\Inf_\mu$ is small for $q$ itself, but here we would like to control the covering number of the $q_A$'s, and it is \textit{not} always the case that $\Inf_\mu(q_A)\leq \Inf_\mu(q)$.
This lack of monotonicity may be blamed on the edges crossing the $A$-$B$ boundary, which can contribute very different values to $\Inf_\mu$ for $q$ vs. $q_A$.
The most that can be said immediately is that $\Inf_\mu^{\Int(A)}(q_A)\leq \Inf_\mu(q)$, where the \textit{interior influence}
\[\Inf_\mu^{\Int(A)}(q_A):=\sum_{\substack{i, x\in A,\\x+i\in A}}\frac{\mu(x)\mu(x+i)}{2\mu(x)+\mu(x+i)}\left(\frac{q(x)-q(x+i)}{2}\right)^2\]
is the subsum of $\Inf_\mu[q]$ using only edges fully inside $A$.
And unfortunately, a bound on $\Inf_\mu^{\Int(A)}(q_A)$ alone does not always constrain the covering number of the $q_A$'s---some counterexamples are explained in
\Cref{sec:AB-split}.

The crucial insight is this: as long as $\mu$ is derived from a Haar-random $f$, the geometry necessary for these counterexamples to arise is not present except with very low probability.
More formally, to finish the proof, we show that with high probability over the choice of $\mu$, a bound on $\Inf_\mu^{\Int(A)}[q_A]$ actually \textit{does} upper-bound the influence of a certain extension $\Ext_Aq_A$ of $q_A$ to the whole cube, and this suffices to control the metric entropy of the $q_A$'s.
We make this comparison via a careful combinatorial argument, drawing some inspiration from the canonical paths analysis of \cite{Huang2025} and their ``local escape property.''

\paragraph{$B$-part analysis: decoupled Fourier coefficients and random pavings.}
We explain the argument for bounding $\|\Pi_B\wh{M_{\wh f}}\Pi_B\|_\op$; the cases of $\Pi_A\wh{M_{\wh f}}\Pi_B$ and $\Pi_B\wh{M_{\wh f}}\Pi_A$ are similar.
First observe that, as $\mu$ varies, $B$ acts as a random subset of $\{\pm 1\}^n$ with permutation-invariant law---and actually membership in $B$ is very nearly i.i.d. Bernoulli-distributed.
As such, $\Pi_B$ can be considered a random coordinate projection, and when applied to a fixed matrix as in $\Pi_B Q\Pi_B$, constitutes a \textit{random paving} in the sense of a work of Tropp \cite{MR2379999} concerning a partial resolution to the Kadison--Singer conjecture.
Provided certain technical conditions hold, \cite{MR2379999} roughly states that sufficiently sparse random coordinate projections have small operator norm.

Of course, our situation is not $\Pi_B Q\Pi_B$ for fixed $Q$; $\wh{M_{\wh f}}$ is certainly not independent of $B$.
However, the entries of $\wh{M_{\wh f}}$ are defined by the Fourier coefficients of $\wh f$, and it turns out that for Haar-random $f$, $\mc O(n)$-sized marginals of $\big(\wh f(s)\big)_s$ can actually be considered as independent of $\mu$.
Thus for tracial moments of $\Pi_B\wh{M_{\wh f}}\Pi_B$ up to order $\mc O(n)$, we can ``freeze'' an approximation to $\wh{M_{\wh f}}$ and apply the random paving result of Tropp.
Unfreezing $\wh{M_{\wh f}}$ gives the result.
This argument is quite involved, and requires several Lindeberg-type replacement and Gaussian interpolation steps. Moreover, throughout, we must work with the $\tau$-smoothed quantities introduced earlier.

\subsection{Outlook}
\label{sec:outlook}
Our work leaves several questions open.
\begin{enumerate}
	\item \textit{Hay from the haystack.} Our test can handle all but a $2^{-\Omega(n)}$ Haar-fraction of target states.
    What familiar, structured families are in the ``good'' set?
   For example, the concurrent work by Du et al.~\cite{du2025certifying}, which also certifies almost all target states with constant robustness, provably applies to random graph states, a subclass of stabilizer states, and, numerically, seems to apply to random brickwork states (which we expect to be the case for our test too).
	\item \textit{Improving the robustness constant.}
	We proved that our test is $C$-robust, \textit{i.e.}\ it can distinguish states that are $\eps$-close from states that are $C\eps$-far, for some \emph{very large} constant $C$ closely related to the second eigenvalue of the $M_f+H^{\otimes n}M_{\wh f}H^{\otimes n}$ operator. However, experiments suggest that $C$ should be at most $4$---see \Cref{fig:numerics}. The constant is important in practice: establishing that $C = 4$ means that our test can distinguish, \textit{e.g.}, states that have $99.9$ fidelity with the target from states that only have $99.6$ fidelity). For comparison, in our prior work~\cite{CLSW26}, although our test relied on one $\mc O(\log n)$-qubit measurement, we were able to formally show that $C=2$.
    \item \textit{A ``fully tolerant'' test.}
	One might also wish for an algorithm that is $(\eps_1,\eps_2)$-tolerant for independently-settable $\eps_1,\eps_2>0$, meaning that it can distinguish between states that are $\epsilon_1$-close and states that are $\epsilon_2$-far. In order to achieve this, it is necessary to depart from the framework of a one-shot (independent of the tolerance parameters) test that is subsequently amplified sequentially.    
    
	\item \textit{Robustness for all states.} As mentioned, our test can only handle all but a $2^{-\Omega(n)}$ Haar-fraction of target states. This limitation is inherent for non-adaptive tests: Gupta, He, and O'Donnell~\cite{guptaHeODonnell2025} showed that some \emph{adaptivity} is required for any test (robust or not) that handles \emph{all} target states. Finding such a test, or showing that it does not exist, is one of the main open questions in this area. 
\end{enumerate}


\begin{figure}
    \centering
    \includegraphics[width=0.8\linewidth]{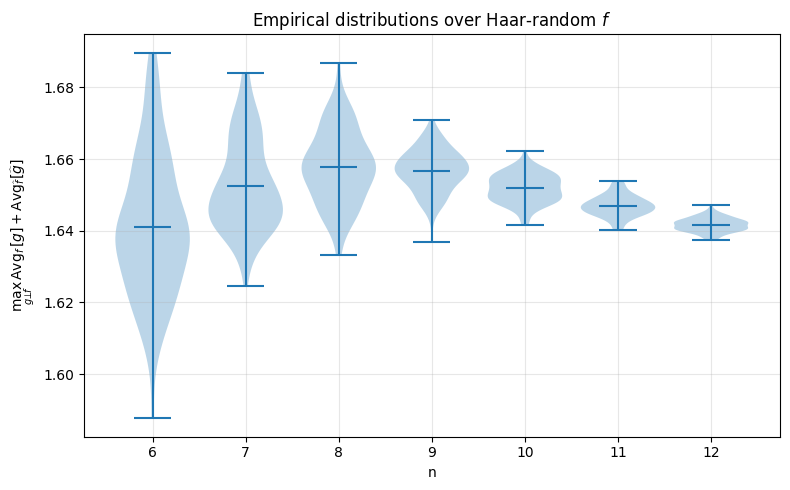}
    \caption{Numerical values for the uncertainty principle constant $2-c$.
    Each distribution represents the empirical distribution of the second eigenvalue of $A_f+H^{\otimes n}A_{\wh f}H^{\otimes n}$ for 100 Haar-random samples.}
    \label{fig:numerics}
\end{figure}

\section{Preliminaries and conventions}

\subsection{Notation}
We briefly fix some notation used throughout.
Further, special symbols defined throughout the proof are recorded in \Cref{sec:gloss}.
\begin{itemize}
    \item We denote $N=2^n$ always.
    \item We use the unitary normalization of the Fourier transform over $\mathbb Z_2^n$:
    \[\wh f (s)= \frac{1}{\sqrt{N}}\sum_x f(x)(-1)^{\sum_js_jx_j}\]
    \item We write $s+t$ for $s,t\in\{0,1\}^n$ to mean the bitwise XOR.
    If $i\in[n]$, we interpret $s+i$ as $s$ with the $i$th index negated.
    The latter will also be written as $s+e_i$ when extra clarity is beneficial.
    \item The undecorated $L^p$ norm $\|\cdot\|_p$ is always with respect to the \textit{counting} measure.
    \item We denote $\langle f, g\rangle = \sum_x \overline{f(x)}g(x)$.
    We use $\langle f, Mg\rangle = \bra{f}M\ket{g}$ interchangeably.
    For a probability measure $\mu$ on $\{0,1\}^n$, we use $\langle f, g\rangle_\mu=\sum_x\mu(x)\overline{f(x)}g(x)$.
    \item The shorthand ``$\unif$'' denotes the uniform probability distribution on $\{0,1\}^n$.
    \item A \textit{Steinhaus} random variable is a uniform (Haar) draw from the set $\{z\in \C:|z|=1\}$. 
    \item Exponents attach first.
    For example, for some random matrix $A$,
    \[\E\tr A^p := \E\left[\tr[A^p]\right]\]
    \item For a probability measure $\mu$ on the hypercube, define $\langle f, g\rangle_\mu= \sum_x \mu(x)\overline{f(x)}g(x)$.
\end{itemize}

\subsection{Basic facts about the Porter--Thomas distribution}

For any probability measure $\mu$ on $\{0,1\}^n$ let the \emph{(Fourier) bias} of $\mu$ be denoted by
\[
        \Bias(\mu)=\max_{s\neq 0^n}\sqrt{N}\,|\wh \mu(s)|
        = \max_{s\neq 0^n}\big|\tsum_x\mu(x)(-1)^{s\cdot x}\big|\,.
        \]
\begin{lem}[Bias of Porter-Thomas measures]
    \label{lem:muhat-bound}
    Let $\mu$ be Dirichlet$(\vec 1)$ on \(N=2^n\) coordinates. Then with
    probability at least \(1-2^{-\Omega(n)}\),
    \[\Bias(\mu)\leq
        \mc O\left(\sqrt{\frac nN}\right).
    \]
\end{lem}

\begin{proof}
Fix a nonzero \(s\in\{0,1\}^n\). Write
\[
    E_s:=\{x:\chi_s(x)=1\},
    \qquad
    O_s:=\{x:\chi_s(x)=-1\}.
\]
Since \(s\neq 0\), the character \(\chi_s\) is balanced, and hence
\[
    |E_s|=|O_s|=N/2.
\]
Therefore, using unitary Fourier normalization,
\[
    \sqrt N\,\widehat\mu(s)
    =
    \sum_x \mu(x)\chi_s(x)
    =
    \mu(E_s)-\mu(O_s)
    =
    2\mu(E_s)-1.
\]
By the aggregation property of the Dirichlet distribution,
\[
    \mu(E_s)\sim \mathrm{Beta}(N/2,N/2).
\]
It is standard that \(\mathrm{Beta}(N/2,N/2)\) is subgaussian with scale
\(N^{-1/2}\) around \(1/2\). Thus there is a universal constant \(c>0\) such
that for all \(t\ge0\),
\[
    \Pr\left[
        \sqrt N\,|\widehat\mu(s)|\ge t
    \right]
    =
    \Pr\left[
        |\mu(E_s)-1/2|\ge t/2
    \right]
    \le
    2\exp(-cNt^2).
\]
Taking a union bound over all \(N-1\) nonzero \(s\), we get
\[
    \Pr\left[
        \max_{s\neq0}\sqrt N\,|\widehat\mu(s)|\ge t
    \right]
    \le
    2N\exp(-cNt^2).
\]
Now choose
\[
    t=C\sqrt{\frac nN}
\]
with \(C\) sufficiently large. Since \(N=2^n\),
\[
    2N\exp(-cNt^2)
    =
    2\exp\bigl(n\log 2-cC^2n\bigr)
    \le
    2^{-\Omega(n)}
\]
for \(C\) large enough. This proves the claim.
\end{proof}

By similar, routine arguments, the same bound holds for the $3/2$-norm:

\begin{fact}
    With probability at least $1-\exp(-\Omega(n))$,
    \[\sum_x\mu(x)^{3/2}\leq \mc O\left(\sqrt{\frac{n}{N}}\right)\,.\]
\end{fact}

\section{Reducing to an uncertainty principle}
\label{sec:4}

We begin the analysis of our one-shot test, \Cref{prot:adaptive}. From here on, the paper will focus on proving that \Cref{prot:adaptive} satisfies \Cref{thm:informal-1}, restated here for convenience.
\theoremstyle{plain}
\newtheorem*{repeatedtheorem}{Theorem \ref{thm:informal-1}}

\begin{repeatedtheorem}[Constant robustness with single-qubit measurements]
There exists an efficient protocol using single-qubit measurements that, for all but a $2^{-\Omega(n)}$ fraction of pure $n$-qubit target states $\ket{\psi}$, satisfies the following. Given oracle access to a classical description of the target state $\ket{\psi}$ (via the model in \Cref{def:oracle}), and a single copy of an $n$-qubit mixed state $\rho$,
\begin{itemize}
    \item (Completeness) Outputs $\mathsf{accept}$ with probability at least $\bra{\psi}\rho\ket{\psi}$.
    \item (Soundness) Outputs $\mathsf{reject}$ with probability at least
    $\frac{1-\braket{\psi|\rho|\psi}}{C} - o(1)$, for some constant $C>1$ (\textit{i.e.}\ approximately the ``infidelity'' up to a constant factor).
\end{itemize}
Moreover, the single-qubit measurements are all non-adaptive Pauli measurements, except for one final (non-Pauli) adaptive measurement.
\end{repeatedtheorem}

The first observation is that, for a target state $\ket{\psi}$ and a lab state $\rho$, we can rewrite the acceptance probability in \Cref{prot:adaptive} as $\textnormal{Tr}\left[A \rho\right]$ where $A$ is the following PSD operator:
$$ A = \frac{1}{2n} \sum_i \sum_{z \in \{0,1\}^{n-1}} \ket{z}\!\!\bra{z}_{[n]\setminus i}\otimes  \ket{\psi^{(i)}_z}\!\!\bra{\psi^{(i)}_z}_i + \frac{1}{2n}  \sum_i \sum_{z \in \{0,1\}^{n-1}}  (H^{\otimes n-1}\ket{z}\!\!\bra{z} H^{\otimes n-1})_{[n]\setminus i}\otimes\ket{\wh {\psi}^{(i)}_z}\!\!\bra{\wh {\psi}^{(i)}_z}_i \,,$$
where $\ket{\psi^{(i)}_z}$ is the post-measurement state of the $i$-th qubit upon measuring all qubits of $\ket{\psi}$, except the $i$-th, in the $Z$ basis and obtaining outcome $z$; and $\ket{\wh {\psi}^{(i)}_z}$ is defined analogously, except the measurements are in the $X$ basis. 

We first argue \emph{completeness}. Consider first a pure lab state $\ket{\phi}$, and write it in an eigenbasis of $A$. Notice then that $\textnormal{Tr}\left[A \ket{\phi}\!\!\bra{\phi}\right]$ is a convex combination of the contributions from each eigenvector of $A$ (since the eigenvectors are orthogonal to each other and remain orthogonal under the action of $A$). Since $\ket{\psi}$ itself is an eigenvector of $A$ with eigenvalue $+1$, we have $\textnormal{Tr}\left[A \ket{\phi}\!\!\bra{\phi}\right] \geq |\braket{\psi| \phi}|^2$. This extends to mixed lab states $\rho$, to give $\textnormal{Tr}[A \rho] \geq \bra{\psi} \rho \ket{\psi}$.  

Thus, from here on, we will focus on \emph{soundness}. For convenience, let us switch to a notation that will be more natural going forward: denote the target state by $\ket{f} = \sum_x f(x) \ket{x}$. Notice that it suffices to argue soundness for pure lab states $\ket{g} = \sum_x g(x) \ket{x}$, as soundness for mixed lab states $\rho$ will follow by convexity. We can further reduce soundness to a statement about lab states $\ket{g}$ that are orthogonal to $\ket{f}$:
 \begin{prop}
 \label{prop:2}
With all but $2^{-\Omega(n)}$ probability over a Haar-random target state $\ket{f}$, for all pure states $\ket{g}$ orthogonal to $\ket{f}$, we have
$$\textnormal{Tr}[A_f \ket{g}\!\!\bra{g}] \leq 1- c +o(1)\,,$$
for some constant $c>0$ where $A_f$ is the acceptance operator corresponding to the target state $\ket{f}$.
\end{prop}
\begin{proof}[Soundness of \Cref{thm:informal-1} from \Cref{prop:2}.] We know $\braket{f|A_f|f}=1$, and, by \Cref{prop:2}, with probability $1 -2^{-\Omega(n)}$ over the target state $\ket{f}$, the following holds for all $\ket{g}$ orthogonal to $\ket{f}$: $\braket{f|A_f|f}\leq 1 - c+o(1)$, for some $c>0$.
    Then, $A_f\preceq \ket{f}\!\!\bra{f}+(1-c+o(1))P$ for $P=I-\ket{f}\!\!\bra{f}$.
    That is,
    \[A_f\preceq (1-c +o(1)) I + \left(c-o(1) \right)\ket{f}\!\!\bra{f}\,.\]
    Thus, any lab state $\rho$ satisfies
    \[\tr[A_f \rho] \leq 1-c + o(1) + (c - o(1)) \braket{f|\rho|f} \,.\]
    Rearranging gives 
    $$\Pr[\mathsf{reject}] = 1- \tr[A_f \rho]  \geq c  (1-\braket{f|\rho|f}) - o(1)\,.$$
This yields exactly the desired soundness of \Cref{thm:informal-1} by setting $C = \frac{1}{c}$.
\end{proof}

So, our goal for the rest of the paper will be to prove \Cref{prop:2}. As discussed in the technical overview, we will view \Cref{prop:2} as an \emph{uncertainty principle}. Let $\ket{\wh f} = H^{\otimes n} \ket{f}$ and $\ket{\wh g} = H^{\otimes n} \ket{g}$. Now, $A_f$ is the sum of a ``$Z$-basis test'' term and an ``$X$-basis test'' term. It is straightforward to verify that we can rewrite $\textnormal{Tr}[A_f \ket{g}\!\!\bra{g}]$ as
$$ \frac12 \Avg_{f}[g] + \frac12 \Avg_{\wh f}[\wh{g}] \,,$$
where
\begin{align*} \Avg_f[g]&:=\frac{1}{2n}\sum_i\sum_x\frac{\Big|\overline{f(x)}g(x)+\overline{f(x+i)}g(x+i)\Big|^2}{|f(x)|^2+|f(x+i)|^2}\,,\\
\Avg_{\wh f}[\wh g]&:=\frac{1}{2n}\sum_i\sum_s\frac{\left|\overline{\wh f(s)}\wh g(s)+\overline{\wh f(s+i)}\wh g(s+i)\right|^2}{|\wh f(s)|^2+|\wh f(s+i)|^2}\,.
\end{align*}
Here, for an index $i \in [n]$, we used the shorthand notation $x +i$ to denote the string where we flip the $i$-th bit of $x$. And the additional factor of $\frac12$ in $\Avg_f[g]$ and $\Avg_{\wh f}[\wh g]$ comes from the fact that the sum over $x$ double counts each inner product. Thus, we can equivalently rewrite our goal as the following theorem (restated from the technical overview):
\theoremstyle{plain}
\newtheorem*{repeatedtheorem2}{\Cref{thm:full-uncertainty}}
\begin{repeatedtheorem2}[The uncertainty principle]
    With probability at least $1-2^{-\Omega(n)}$ over a Haar random $f:\{0,1\}^n\to \C$, $\|f\|_2=1$,
    for all $g:\{0,1\}^n\to \C$, $\|g\|_2=1$ with $\langle f,g\rangle=0$,
    \[\Avg_f[g]+\Avg_{\wh f}[\wh g]\leq 2-c,\]
    where $c>0$ is a universal constant independent of $n$.
\end{repeatedtheorem2}

As mentioned in the technical overview, the latter uncertainty principle has an elegant reformulation as an uncertainty principle for (generalized) total influences.


\begin{cor}[Uncertainty principle for generalized influences]
\label{cor:influences-formal}
For a probability measure $\mu$ on $\{0,1\}^n$ and a function $h:\{0,1\}^n\to \C$, define the ``Harmonic mean weighted influence''
\[\Inf_\mu[h]=\sum_{i,x}\frac{2\mu(x)\mu(x+i)}{\mu(x)+\mu(x+i)}\left(\frac{h(x) - h(x+i)}{2}\right)^2\]
Then, with high probability over a Haar random $f:\{0,1\}^n\to \C$, $\|f\|_2=1$,
    for all $g:\{0,1\}^n\to \C$, $\|g\|_2=1$ with $\langle f,g\rangle=0$,
\[\Inf_{\mu_f}\left[\frac{g}{f}\right]+\Inf_{\mu_{\wh f}}\left[\frac{\wh g}{\wh f}\right] \geq c n\,,\]
where $c>0$ is the same constant as in \Cref{thm:full-uncertainty}.
\end{cor}
\begin{proof}
This follows immediately from the identity 
$$\Avg_f[g] = 1-\frac1n\Inf_{\mu_f}\left[\frac{g}{f}\right] \,,$$
which holds for all $f$ and $g$ such that $\|f\|_2=1$, $\|g\|_2=1$, and hence also for $\wh f$ and $\wh g$
(as long as $f$ and $\wh f$ are non-zero everywhere, which happens with probability $1$). We leave the proof of this identity to \Cref{app:identity}.
\end{proof}

\section{Concentration of $\Avg_{\wh f}[\wh{fq}]$ for fixed $q$}
\label{sec:concentration}

The first key technical step in our argument (which is outlined at a high level in the technical overview) is to prove the following concentration bound for $\Avg_{\wh f}[\wh{fq}]$. This section is devoted to the proof of this bound.

\begin{theorem}[{Concentration of $\Avg_{\wh f}[\wh{fq}]$ for fixed $q$}]
\label{thm:avg-hat-fixed-q}
There exists a universal constant $c>0$ such that for all $t>0$, with probability $1-\exp(-\Omega(n))$ over $\mu\sim\bs\mu$, for any $\mu$-compatible $q$ satisfying
\[
    \max_x\mu(x)|q(x)|^2
    \le
    \|\mu\|_\infty^{1/3},
\]
we have the bound
\[
\Pr_{f\sim F_\mu}\left[
    \Avg_{\wh f}[\wh{fq}]
    \ge
    \frac12+t+o(1)
\right]
\le \exp\left(
    -ct^6N^{1/3-o(1)}
\right).
\]
\end{theorem}

Fixing $f$ and $g$ for now, let us establish some notation.
Write
\[A^{(i)}=\tfrac12\tsum_{s}\left|\langle v_s,w_s\rangle\right|^2\]
for
\[v_s=\frac{u_s}{\|u_s\|_2}\quad\text{ with }\quad u_s=\big(\wh f(s),\wh f(s+i)\big) \qquad\text{and}\qquad w_s=\big(\wh g(s), \wh g(s+i)\big)\,.\]
Then $\Avg_{\wh f}[\wh g]=\frac1n\sum_iA^{(i)}$.

Each $A^{(i)}$s is a sum of terms each of which depend on disjoint pairs of Fourier coefficients of the underlying random vector $\big(f(x)\big)_x$.
If the underlying randomness were rotationally-invariant (namely, Gaussian), then unitarity of the Fourier transform would also imply $A^{(i)}$ is the sum of independent random variables, each functions of disjoint independent Fourier random variables.
Unfortunately, the Steinhaus ensemble $\big(f(x)\big)_x$ is not rotationally invariant.
That would not be a problem if $v_s$'s were not normalized, for then $A^{(i)}$ would at least be a low-degree polynomial chaos.
But normalized it is, and so typical approaches---even bounded differences or Herbst-type arguments---do not immediately apply.

To get around this, our strategy is to introduce for each $i\in[n]$ the following smoothed quantity, which has well-behaved derivatives.
For $\delta>0$, let
\[A_\delta^{(i)} := \frac12\sum_s\frac{|\langle u_s,w_s\rangle|^2}{\|u_s\|_2^2+\delta}.\]
Let $D^{(i)}_\delta = A^{(i)}-A^{(i)}_\delta$ be the difference.
Now $D^{(i)}_\delta$ also has bad derivatives, so it would seem we have done nothing to help the situation.
But in fact, $D_\delta^{(i)}$ is pointwise
dominated by something nice:
\begin{align*}
D^{(i)}_\delta = \frac12\sum_s|\langle u_s, w_s\rangle|^2\left(\frac1{\|u_s\|_2^2}-\frac1{\|u_s\|_2^2+\delta}\right) &= \frac12\sum_s |\langle u_s, w_s\rangle|^2 \frac{\delta}{\|u_s\|_2^2\cdot(\|u_s\|_2^2+\delta)} \\
&
\leq\frac12\sum_s\frac{\delta\|w_s\|_2^2}{\|u_s\|_2^2+\delta}\tag{Cauchy--Schwarz}\\
&=:B^{(i)}_\delta\,.
\end{align*}
The gradient of $B^{(i)}_\delta$ is sufficiently well-behaved for a Herbst-type argument succeed.

\Cref{thm:avg-hat-fixed-q} is proved in two parts.
First, we show the expectation is $1/2+o(1)$ by passing to a Gaussian model (each $f(x)$ distributed as a complex Gaussian with variance $\mu(x)$).
Then we show $\Avg_{\wh f}[\wh{fq}]$ concentrates about this expectation via a Herbst argument and log-Sobolev inequalities for Steinhaus ensembles.
The $\delta$-smoothing idea above will be crucial in both parts.

\subsection{Technical preparations}

We will need hypercontractivity for Steinhaus random variables:
\begin{fact}[$L^p$-$L^2$ Steinhaus Hypercontractivity \cite{WEISSLER1980218}]
\label{fact:steinhaus-hc}
	Let $\mathbb T^n = \{(z_1, \ldots, z_n): |z_i| = 1\}$. Let $f:\mathbb T^n\to \C$ be a polynomial in $z_i,\bar z_i$, $i=1,\ldots, n$ of total degree at most $d$.
	Then for all $p\geq 2$,
	\[\left(\E |f|^p\right)^{1/p}\leq(p-1)^{d/2}\left(\E |f|^2\right)^{1/2}.\]
\end{fact}

And let us also state the straightforward Wick-type pairing rule about Steinhaus mixed moments.
\begin{fact}[Steinhaus pairing rule]
Let $J \in \mathbb{N}$.
	\label{fact:wick}
	For $\omega_j$ i.i.d.\ Steinhaus random variables, and $a_j, b_j$ integers, for $j=1,\ldots, J$, one has
	\[\E\textstyle\prod_j\omega_j^{a_j}\bar\omega_j^{b_j}=\begin{cases}
		1 & \text{if }a_j=b_j,\,j=1,\ldots, J\\
		0 &\text{otherwise}
	\end{cases}\]
\end{fact}

We will also use moments of $w_s$ sums in several places.

\begin{lem}
\label{lem:w4}
Let $W=\sum_s\|w_s\|_2^4$.
There exist universal constants $c,C>0$ such that
\[\E W\leq \frac8N\qquad\text{and}\qquad\Pr[|W-\E W|>t]\leq \exp\big(-c\sqrt{Nt}\big) \quad\text{for all} \quad  t\geq \frac{C}{N}.\]
\end{lem}
\noindent In particular, we will use the instantiation
\begin{equation}
	\label{ineq:Wbound}
	\Pr\left[W > \frac{K}{N^{1/3}}\right]\leq \exp(-c\sqrt{K}N^{1/3}).
\end{equation}
\begin{proof}
	We claim that for all $s$,
	\begin{equation}
		\label{ineq:Ew_s}
		\E\|w_s\|_2^4\leq \frac8{N^2}\,.
	\end{equation}
	Indeed, 
	\[\E\|w_s\|_2^4 = \E|\wh g(s)|^4 + 2\E|\wh g(s)|^2|\wh g(s+i)|^2+\E|\wh g(s+i)|^4\]
	and by \Cref{fact:wick}, for any $s$,
	\begin{align*}
		\E |\wh g(s)|^4 &= \frac{1}{N^2}\E|\tsum_x \sqrt{\mu(x)}\omega_x q(x)(-1)^{\langle x,s\rangle}|^4\\
		&= \frac1{N^2}\left(2\sum_{x, y}\mu(x)|q(x)|^2\mu(y)|q(y)|^2-\sum_x\mu(x)^2|q(x)|^4\right)\\
		&=\frac{1}{N^2}\left(2 - \sum_x \mu(x)^2|q(x)|^4\right)\leq \frac{2}{N^2}.
	\end{align*}
	Similarly,
	\begin{align*}
		\E|\wh g(s)|^2|\wh g(s+i)|^2 &= \frac{1}{N^2}\left(1+\left|\sum_x\mu(x)|q(x)|^2(-1)^{x_i}\right|^2-\sum_x\mu(x)^2|q(x)|^4\right)\\
		&\leq \frac{1}{N^2}\left(1+\left|\sum_x\mu(x)|q(x)|^2\right|^2\right)\tag{$\|g\|_2^2=1$}\\
		&= \frac{2}{N^2}\,.
	\end{align*}
	This proves \eqref{ineq:Ew_s}; summing over $s$ then gives $\E W\leq 8/N$.

	Towards concentration, applying Minkowski ($L^2$ triangle inequality) we get
	\[\Var[W] \leq \E W^2 = \E\left(\sum_s\|w_s\|_2^4\right)^2\leq\left(\sum_s \left(\E \|w_s\|_2^8\right)^{1/2}\right)^2\,.\]
	Then by hypercontractivity for Steinhaus (\Cref{fact:steinhaus-hc}) and Eq. \eqref{ineq:Ew_s},
	\begin{align*}
		\E\|w_s\|_2^8 &\lesssim (\E\|w_s\|_2^4)^2\lesssim \frac{1}{N^4}.
	\end{align*}
	So
	\[\Var[W]\lesssim \frac{1}{N^2}.\]
	Markov's inequality finishes the argument.
	With $p$ to be chosen later,
	\begin{align*}
		\Pr[|W-\E W| > t] &= \Pr[|W-\E W|^p > t^p]\\
		&\leq \frac{\E[|W-\E W|^p]}{t^p}\\
		&\leq t^{-p}(p-1)^{2p}\Var(W)^{p/2}\tag{\Cref{fact:steinhaus-hc}}\\
		&\leq \left(\frac{p^2C}{tN}\right)^p
	\end{align*}
	and choosing $p=\sqrt{tN/(eC)}$,
	\[\Pr[|W-\E W| > t]\leq \exp\left(-c\sqrt{tN}\right)\,.\qedhere \]
\end{proof}

\subsection{Bounding the expectation}
\label{sec:fourier-avg-small}

In this section we prove:
\begin{lem}
\label{lem:fixed-q-expect}
For all $t>0$, for all $n$ large enough, with exponential probability over $\mu\sim \bs \mu$, for all fixed $\mu$-compatible $q$,
\[\mathop{\E}_{f\sim F_\mu}\Avg_{\wh f}[\wh{fq}]=\mathop{\E}_{f\sim F_\mu}\left[\frac1n\sum_iA^{(i)}\right] \leq \frac12 + t\,.\]
\end{lem}

To evaluate the expectation, the general idea is to use that in the $n\to\infty$ limit, small marginals of the joint distribution $\big(\wh f(s)\big)_{s\in\{0,1\}^n}$ are jointly Gaussian---and each $A^{(i)}$ can be exactly integrated in the Gaussian model.

More formally, let $\mu\sim\bs \mu$ be fixed.
In the Gaussian model we replace $f$ with independent scaled Steinhaus entries with $f$ sampled with independent entries $f(x)\sim\sqrt{\mu(x)}\cdot \mc{CN}(0,1)$.
Denote this distribution of $f$'s by $\mathrm{Gauss}(\mu)$.

Unfortunately, while the $A^{(i)}$ \textit{are} sums of terms involving only a small marginal of the joint distribution $\big(\wh f(s)\big)_s$, these terms are not themselves Lipschitz in the underlying randomness, so standard Berry--Esseen-type comparison theorems do not apply.
To get the desired convergence to the Gaussian model, we must therefore work with the smoothed quantities $A^{(i)}_\delta$ and the error bounds $B^{(i)}_\delta$ introduced earlier.

In particular, because
\[\E A^{(i)}\leq \E A^{(i)}_\delta + \E B^{(i)}_\delta,\]
\Cref{lem:fixed-q-expect} would be true if we could show the RHS terms are at most $1/2+o(1)$ and $o(1)$ respectively.
Schematically, we get the bound via
\begin{align*}
    \E A^{(i)} &\leq \E A^{(i)}_\delta + \E B^{(i)}_\delta\\
    &= \E_{\text{Gauss}(\mu)} A^{(i)}_\delta + \left(\E A^{(i)}_\delta - \E_{\text{Gauss}(\mu)}A^{(i)}_\delta\right) + \E B^{(i)}_\delta\\
    &\leq \E_{\text{Gauss}(\mu)} A^{(i)} + \left(\E A^{(i)}_\delta - \E_{\text{Gauss}(\mu)} A^{(i)}_\delta\right) + \E B^{(i)}_\delta
\end{align*}
where the last line holds because $A^{(i)}_\delta \leq A^{(i)}$ pointwise.
Each of the terms in the last line are bounded in a separate subsection: the first term is proved by direct integration, where we make heavy use of Gaussianity; the second term is by a Lindeberg replacement argument; and the third term is handled via negative moment bounds for Steinhaus sums.

\subsubsection{Bounding $\E A^{(i)}$ in the Gaussian model}

The main claim of this section is as follows:
\begin{lem}
\label{lem:expect-A-Gauss}
With probability $1-2^{-\Omega(n)}$ over $\mu\sim \bs \mu$,
    \[\mathop{\E}_{\mathrm{Gauss}(\mu)}\left[\frac1n\sum_iA^{(i)}\right]\leq \frac12 + o(1)\]
\end{lem}

Let us record some facts about $u_s$ and $w_s$ in the Gaussian model.
Begin by noting that for all $s$,
\[u_s = \mathcal{CN}(0,\Sigma_i) \quad\text{with}\qquad \Sigma_i = \frac1N\begin{pmatrix}
    1 & \sqrt{N}\wh\mu(i)\\
    \sqrt{N}\wh\mu(i) & 1
\end{pmatrix}\,.\]
Justification for the off-diagonal values of $\Sigma_i$ is as follows: 
\begin{align*}
\E\left[\overline{\wh f(s)}\wh f(s+i)\right] &=\tfrac{1}{N}\E\left[\tsum_{x,y}\overline{f(x)}f(y)\chi_{s}(x)\chi_{s+i}(y)\right]\\
&= \tfrac{1}{N}\E\left[\tsum_{x}|f(x)|^2\chi_{i}(x)\right]\\
&=\tfrac{1}{N}\tsum_x\mu(x)\chi_i(x)\\
&= \frac{\wh\mu(i)}{\sqrt{N}}\,.
\end{align*}
And similarly,
\[w_s = \mc{CN}(0,\wt\Sigma_i)\qquad\text{with}\qquad \wt\Sigma_i=\E[\ketbra{w_s}{w_s}]=\frac1N\begin{pmatrix}
    1 & \sqrt{N}\wh{\mu |q|^2}(i)\\
    \sqrt{N}\wh{\mu |q|^2}(i) & 1
\end{pmatrix}\]

Moreover, in this model we may relate $w_s$ to $u_s$ via Gaussian regression.
That is, we may decompose
    \begin{equation}
    \label{eq:gaussian-regression}
        w_s = Au_s + \gamma,\quad A = R \Sigma^{-1},\quad \gamma \perp u,\quad R = \E[w_su_s^\dagger]=\frac{1}{\sqrt{N}}\begin{pmatrix}
        0 & \wh{q\mu}(i)\\
        \wh{q\mu}(i) & 0
    \end{pmatrix}\,.
    \end{equation}
    Note that the diagonal is zero because it is $\wh{q\mu}(0)=\E_\mu q=0$.

\medskip

Before continuing, let us record a similar result for the expected outer product of $v_s$, \emph{i.e.}, $\E[v_sv_s^\dagger]$:

\begin{prop}
\label{lem:soft-vs}
Suppose $\sqrt N\,|\wh\mu(i)|\leq 1/2$.
Then
\[
\E[\ket{v_s}\!\bra{v_s}]
=
\frac12
\begin{pmatrix}
1 & \Phi\big(\sqrt N\,\wh\mu(i)\big)\\
\Phi\big(\sqrt N\,\wh\mu(i)\big) & 1
\end{pmatrix},
\]
where
\[
\Phi(x)
:=
\int_0^1 \frac{2t-1+x}{1+x(2t-1)}\,dt
=
\frac{2x-(1-x^2)\log\!\frac{1+x}{1-x}}{2x^2}.
\]
Consequently,
\[
\E[\ket{v_s}\!\bra{v_s}]
=
\frac12
\begin{pmatrix}
1 & \mc O\big(\sqrt N\,\wh\mu(i)\big)\\
\mc O\big(\sqrt N\,\wh\mu(i)\big) & 1
\end{pmatrix}.
\]
\end{prop}
The proof is by explicit integration.
\begin{proof}
Let $\rho_i=\sqrt{N}\wh{\mu}(i)$.
The Hadamard matrix $H=\frac1{\sqrt2}\left(\begin{smallmatrix}
1&1\\
1&-1
\end{smallmatrix}\right)$ diagonalizes $\Sigma_i$; \textit{i.e.,}
\[
H\Sigma_i H
=
\frac1N
\begin{pmatrix}
1+\rho_i&0\\
0&1-\rho_i
\end{pmatrix}.
\]
So if $\xi_1,\xi_2\overset{\mathsf{iid}}{\sim}\mc{CN}(0,1)$ and
\[
y:=\begin{pmatrix}
\sqrt{1+\rho_i}\,\xi_1\\
\sqrt{1-\rho_i}\,\xi_2
\end{pmatrix},
\]
then $u_s$ has the same distribution as $\frac1{\sqrt N}Hy$.
Since normalization cancels the factor $1/\sqrt N$,
\[
\ket{v_s}\!\bra{v_s}\overset{(\text{dist})}=H\,\frac{yy^*}{\|y\|_2^2}\,H\qquad\text{and}\qquad 
\E[\ket{v_s}\!\bra{v_s}]
=
H\,\E\left[\frac{yy^*}{\|y\|_2^2}\right]H.
\]

Now write
\[
X:=|\xi_1|^2,\qquad Y:=|\xi_2|^2.
\]
Then $X,Y$ are i.i.d. $\mathrm{Exp}(1)$. By phase symmetry,
\[
\E\left[\frac{\xi_1\bar\xi_2}{(1+\rho_i)X+(1-\rho_i)Y}\right]=0,
\]
so
\[
\E\left[\frac{yy^*}{\|y\|_2^2}\right]
=
\begin{pmatrix}
\phi(\rho_i)&0\\
0&1-\phi(\rho_i)
\end{pmatrix},
\]
where
\[
\phi(\rho):=\E\left[\frac{(1+\rho)X}{(1+\rho)X+(1-\rho)Y}\right].
\]
Conjugating by $H$ gives
\[
\E[\ket{v_s}\!\bra{v_s}]
=
\frac12
\begin{pmatrix}
1 & 2\phi(\rho_i)-1\\
2\phi(\rho_i)-1 & 1
\end{pmatrix}.
\]
Thus it remains to identify
\[
\Phi(\rho):=2\phi(\rho)-1
=
\E\left[\frac{(1+\rho)X-(1-\rho)Y}{(1+\rho)X+(1-\rho)Y}\right].
\]

Use the change of variables
\[
r=X+Y,\qquad t=\frac{X}{X+Y}.
\]
For i.i.d. exponentials $X,Y$, $r\sim \Gamma(2,1)$, $t\sim \mathrm{Unif}[0,1]$, and $r,t$ are independent. Since $X=rt$ and $Y=r(1-t)$, the $r$ cancels and we get
\[
\Phi(\rho)
=
\int_0^1 \frac{(1+\rho)t-(1-\rho)(1-t)}{(1+\rho)t+(1-\rho)(1-t)}\,dt
=
\int_0^1 \frac{2t-1+\rho}{1+\rho(2t-1)}\,dt.\qedhere
\]
\end{proof}

\medskip

\Cref{lem:expect-A-Gauss} is proved in two parts.
First an upper bound on each $\E_{\mathrm{Gauss}(\mu)}A^{(i)}$ is obtained in terms of quantities like $\wh{\mu(i)}$, and then the average of these quantities over $i\in[n]$ is shown to be small.

\begin{prop}For all $\mu$ with $\sqrt{N}|\wh \mu(i)|\leq 1/2$ and all $\mu$-compatible $q$, in the Gaussian model just described,
    \[0\leq \mathop{\E}_{f\sim \mathrm{Gauss}(\mu)}[A^{(i)}]\leq \frac12 + \mc O\left(\sqrt{N}|\wh\mu(i)|+\frac{N|\wh{\mu q}(i)|^2}{(1-\sqrt{N}|\wh\mu(i)|)^2}\right)\,.\]
\end{prop}
\begin{proof}
    According to the Gaussian regression identites \eqref{eq:gaussian-regression}, the $s$\textsuperscript{th} term of $A^{(i)}$ decomposes as
    \[
        |\langle v_s,w_s\rangle|^2 = |\braket{v_s|A|u_s}|^2+|\braket{v_s|\gamma}|^2 + 2\Re\overline{\braket{v_s|A|u_s}}\braket{v_s|\gamma}\]
    And by independence,
    \[\E[ |\langle v_s,w_s\rangle|^2]=\underbrace{\E[|\braket{v_s|A|u_s}|^2]}_{(*)}+\underbrace{\E[|\braket{v_s|\gamma}|^2]}_{(**)}.\]
    \medskip

    \noindent \textit{Bound the $(**)$ term.}
Again by independence we have
\[
    (**)
    =
    \tr\left(
        \E[\ketbra{\gamma}{\gamma}]
        \E[\ketbra{v_s}{v_s}]
    \right).
\]
Because the covariance matrix $\wt\Sigma_i$ of $w_s$ dominates that of
$\gamma$, we have $\E[\ketbra{\gamma}{\gamma}]
    \preceq
    \wt\Sigma_i$.
Therefore, using \Cref{lem:soft-vs},
\begin{align*}
(**)
&\leq
\tr\left(
    \wt\Sigma_i
    \E[\ketbra{v_s}{v_s}]
\right)
\\
&=
\tr\left[
    \frac1N
    \begin{pmatrix}
        1
        &
        \sqrt N\,\wh{\mu|q|^2}(i)
        \\
        \sqrt N\,\wh{\mu|q|^2}(i)
        &
        1
    \end{pmatrix}
    \cdot
    \frac12
    \begin{pmatrix}
        1
        &
        \Phi\bigl(\sqrt N\,\wh\mu(i)\bigr)
        \\
        \Phi\bigl(\sqrt N\,\wh\mu(i)\bigr)
        &
        1
    \end{pmatrix}
\right]
\\
&=
\frac1N
\left[
    1
    +
    \sqrt N\,\wh{\mu|q|^2}(i)
    \Phi\bigl(\sqrt N\,\wh\mu(i)\bigr)
\right]
\\
&\leq
\frac1N
\left(
    1
    +
    C\sqrt N\,|\wh\mu(i)|
\right).
\end{align*}
Here we used
\[
\left|
    \sqrt N\,\wh{\mu|q|^2}(i)
\right|
=
\left|
    \sum_x
    \mu(x)|q(x)|^2(-1)^{x_i}
\right|
\leq
\sum_x\mu(x)|q(x)|^2
=
1,
\]
together with
\[
    |\Phi(\rho)|
    \leq
    C|\rho|
    \qquad
    \text{for }
    |\rho|\leq\frac12.
    \tag*{$\lozenge$}
\]

    \medskip
    \noindent \textit{Bounding the $(*)$ term.}
    \begin{align*}
        (*)&= \E\left[\left|\frac{\braket{u_s|A|u_s}}{\|u_s\|_2}\right|^2\right]\leq \|A\|_\text{op}^2\E\|u_s\|^2= \|A\|_\text{op}^2\tr\Sigma_i=\frac{2}{N}\|A\|_\text{op}^2
    \end{align*}
    Now to bound $\|A\|_\op$, use that $\|R\|_\op=|\wh{\mu q}(i)|/\sqrt{N}$ and $\|\Sigma^{-1}\|_\text{op}= N/(1-\sqrt{N}|\wh\mu(i)|)$.
    Thus
    \[(*)\leq \frac{2|\wh{\mu q}(i)|^2}{\big(1-\sqrt{N}|\wh\mu(i)|\big)^2} \tag*{$\lozenge$}\]
    \medskip
    Note that in both parts dependence on $s$ has disappeared.
    Substituting back into $A^{(i)}$ we find
    \begin{align*}\E[A^{(i)}]&\leq \frac{N}{2}\left[\frac1N\left(1+C\sqrt{N}|\wh\mu(i)|\right)+\frac{2|\wh{\mu q}(i)|^2}{(1-|\wh\mu(i)|\sqrt{N})^2}\right]\\
    &=\frac12 + \mc O\left(\sqrt{N}|\wh\mu(i)|+\frac{N|\wh{\mu q}(i)|^2}{(1-\sqrt{N}|\wh\mu(i)|)^2}\right).\qedhere
    \end{align*}
\end{proof}

\medskip

It remains to show that both $\wh\mu(i)$ and $\wh{\mu q}(i)^2$ are small with high probability:

\begin{prop}
    With probability $1-2^{-\Omega(n)}$ over $\mu\sim \bs \mu$, for all $\mu$-compatible $q$,
    \[\frac1n\sum_i\left(\sqrt{N}|\wh\mu(i)|+\frac{N|\wh{\mu q}(i)|^2}{(1-\sqrt{N}|\wh\mu(i)|)^2}\right)=\mc O(1/n)\]
\end{prop}
\begin{proof}
    By \Cref{lem:muhat-bound}, with probability $1-2^{-\Omega(n)}$, $\sqrt{N}\wh{\mu}(i)\leq \mc O(\sqrt{n/N})$.
    Thus it remains to understand the numerator of the second term in each summand.
    Summing over $i$, we find
    \[\sum_i N|\wh{\mu q}(i)|^2=\sum_i|\langle q,\chi_i\rangle_\mu|^2\leq  \|G\|_\text{op}\|q\|^2_{L^2(\mu)}=\|G\|_\text{op}\]
    where $G$ is Gram matrix of level-1 Fourier characters wrt $\mu$ inner product:
    \[G_{ij}=\langle\chi_i,\chi_j\rangle_\mu.\]
    Now write $G=\mathbf{1}+M$ and note that
    \[\|M\|_\text{op}\leq \max_i\sum_{j\neq i}|M_{ij}|\leq n\max_{|s|=2}\sqrt{N}|\wh \mu(s)|,\]
    where in the last inequality we have used that for $i\neq j$,
    \[|\langle \chi_i,\chi_j\rangle| = |\tsum_x\mu(x)(-1)^{x_i+x_j}|=\sqrt{N}|\wh\mu(\{i,j\})|\,.\]
    By \Cref{lem:muhat-bound}, with high probability,
    $\max_{|s|=2}\sqrt{N}|\wh \mu(s)|=\mc O(\sqrt{n/N})$.
    Thus we conclude
    \[\frac1n\sum_iN|\wh{\mu q}(i)|^2\leq \frac1n\|G\|_\text{op}\leq \frac1n\big(1+n\mc O(\sqrt{n/N})\big)=\mc O(1/n).\qedhere\]
\end{proof}

\medskip
We are ready to finish this subsection.
\begin{proof}[Proof of \Cref{lem:expect-A-Gauss}]
Combining the propositions, we have
\[\frac1n\sum_i\mathop{\E}_{\mathrm{Gauss}(\mu)}[A^{(i)}]\leq \frac12 + O\left(\frac1n\sum_i\left(\sqrt{N}|\wh\mu(i)|+\frac{N|\wh{\mu q}(i)|^2}{(1-\sqrt{N}|\wh\mu(i)|)^2}\right)\right)=\frac12 + \mc O(1/n)\,.\]
\end{proof}

\subsubsection{Comparing $\E A^{(i)}_\delta$ and $\E_\text{Gauss} A^{(i)}_\delta$}

In this section it will be convenient to use the notation of Fr\'echet derivatives to differentiate random variables as functions with respect to a coordinate of underlying randomness in $\C\cong\mathbb R^2$.
For any $h:\C\cong\mathbb R^2\to \mathbb R$, the notation $\frechet^Jh(z)[v_1,\ldots, v_J]$ denotes the $J$\textsuperscript{th} Fr\'echet derivative at $z$ evaluated along directions $v_1,\ldots, v_J \in \C\cong \mathbb R^{2}$ and extends to mixed derivatives in the natural way.
In the context of Fr\'echet derivatives we shall also use the undecorated norm $\|\cdot\|$ to denote the operator norm of the relevant space.
Concretely,
    \[\|\frechet^Jh(z)\|:=\sup_{\substack{v_j\in \mathbb R^2,\|v_j\|_2=1\\j=1,2,\ldots, J}}\big|\frechet^Jh(z)[v_1,\ldots, v_J]\big|\,.\]

For this section we will also use $\Phi_\delta$ to describe the terms in our smoothed average:
    \[\Phi_\delta(u,w) := \frac{|\langle u,w\rangle|^2}{\|u\|_2^2+\delta}\]


We will need some bounds on the mixed derivatives of $\Phi_\delta(u,w)$; for this it is most convenient to argue in two steps.
\begin{lem}[Master derivative lemma]
\label{lem:master-frechet}
Define
\[
P_\delta(u):=\frac{uu^*}{\|u\|_2^2+\delta},
\qquad u\in\C^2.
\]
View \(P_\delta:\C^2\cong\R^4\to M_2(\C)\) as a smooth real Fr\'echet map.
Then for each \(k=0,1,2,3\),
\[
\|\frechet^kP_\delta(u)\|\lesssim_k \delta^{-k/2}
\qquad\text{for all }u\in\C^2.
\]
\end{lem}


\begin{proof}
Write
\[
s(u):=\|u\|_2^2+\delta,
\qquad
T(u):=uu^*,
\qquad
P_\delta(u)=\frac{T(u)}{s(u)}.
\]
We regard all derivatives as real Fr\'echet derivatives on \(\C^2\cong\mathbb R^4\).

For \(h,h_1,h_2,h_3\in\C^2\), define
\[
\beta_u(h):=2\Re\langle u,h\rangle,
\qquad
\gamma(h_1,h_2):=2\Re\langle h_1,h_2\rangle.
\]
Then
\[
DT(u)[h]=uh^*+hu^*,
\qquad
D^2T(u)[h_1,h_2]=h_1h_2^*+h_2h_1^*,
\qquad
D^3T(u)\equiv 0,
\]
and
\[
Ds(u)[h]=\beta_u(h),
\qquad
D^2s(u)[h_1,h_2]=\gamma(h_1,h_2),
\qquad
D^3s(u)\equiv 0.
\]

Now let \(r(u):=s(u)^{-1}\). Then
\[
Dr(u)[h]=-\frac{\beta_u(h)}{s(u)^2},
\]
\[
D^2r(u)[h_1,h_2]
=
\frac{2\beta_u(h_1)\beta_u(h_2)}{s(u)^3}
-
\frac{\gamma(h_1,h_2)}{s(u)^2},
\]
and
\begin{align*}
D^3r(u)[h_1,h_2,h_3]
&=
-\frac{6\beta_u(h_1)\beta_u(h_2)\beta_u(h_3)}{s(u)^4}\\
&\quad
+\frac{2\bigl(\gamma(h_1,h_2)\beta_u(h_3)+\gamma(h_1,h_3)\beta_u(h_2)+\gamma(h_2,h_3)\beta_u(h_1)\bigr)}{s(u)^3}.
\end{align*}

Since \(P_\delta=T\,r\), the product rule gives
\[
DP_\delta(u)[h]
=
\frac{uh^*+hu^*}{s(u)}
-
\frac{\beta_u(h)}{s(u)^2}uu^*.
\]

Differentiating once more,
\begin{align*}
D^2P_\delta(u)[h_1,h_2]
&=
\frac{h_1h_2^*+h_2h_1^*}{s(u)}\\
&\quad
-\frac{\beta_u(h_2)(uh_1^*+h_1u^*)+\beta_u(h_1)(uh_2^*+h_2u^*)}{s(u)^2}\\
&\quad
-\frac{\gamma(h_1,h_2)}{s(u)^2}uu^*
+
\frac{2\beta_u(h_1)\beta_u(h_2)}{s(u)^3}uu^*.
\end{align*}

Differentiating a third time yields
\begin{align*}
D^3P_\delta(u)[h_1,h_2,h_3]
&=
-\frac{\beta_u(h_3)(h_1h_2^*+h_2h_1^*)+\beta_u(h_2)(h_1h_3^*+h_3h_1^*)+\beta_u(h_1)(h_2h_3^*+h_3h_2^*)}{s(u)^2}\\
&\quad
-\frac{\gamma(h_2,h_3)(uh_1^*+h_1u^*)+\gamma(h_1,h_3)(uh_2^*+h_2u^*)+\gamma(h_1,h_2)(uh_3^*+h_3u^*)}{s(u)^2}\\
&\quad
+\frac{2\beta_u(h_2)\beta_u(h_3)(uh_1^*+h_1u^*)+2\beta_u(h_1)\beta_u(h_3)(uh_2^*+h_2u^*)+2\beta_u(h_1)\beta_u(h_2)(uh_3^*+h_3u^*)}{s(u)^3}\\
&\quad
+\frac{2\bigl(\gamma(h_1,h_2)\beta_u(h_3)+\gamma(h_1,h_3)\beta_u(h_2)+\gamma(h_2,h_3)\beta_u(h_1)\bigr)}{s(u)^3}uu^*\\
&\quad
-\frac{6\beta_u(h_1)\beta_u(h_2)\beta_u(h_3)}{s(u)^4}uu^*.
\end{align*}

We now estimate these in operator norm. Assume \(\|h\|_2,\|h_j\|_2\le 1\). Then
\[
\|uh^*+hu^*\|\le 2\|u\|_2,
\qquad
\|h_1h_2^*+h_2h_1^*\|\le 2,
\qquad
\|uu^*\|=\|u\|_2^2,
\]
and
\[
|\beta_u(h)|\le 2\|u\|_2,
\qquad
|\gamma(h_1,h_2)|\le 2.
\]
Therefore
\[
\|DP_\delta(u)\|
\lesssim
\frac{\|u\|_2}{s(u)}+\frac{\|u\|_2^3}{s(u)^2},
\]
\[
\|D^2P_\delta(u)\|
\lesssim
\frac1{s(u)}+\frac{\|u\|_2^2}{s(u)^2}+\frac{\|u\|_2^4}{s(u)^3},
\]
and
\[
\|D^3P_\delta(u)\|
\lesssim
\frac{\|u\|_2}{s(u)^2}
+\frac{\|u\|_2^3}{s(u)^3}
+\frac{\|u\|_2^5}{s(u)^4}.
\]

Using the elementary bound
\[
\frac{\|u\|_2^m}{(\|u\|_2^2+\delta)^\ell}\lesssim_{m,\ell}\delta^{m/2-\ell}
\qquad (m\le 2\ell),
\]
we obtain
\[
\|DP_\delta(u)\|\lesssim \delta^{-1/2},
\qquad
\|D^2P_\delta(u)\|\lesssim \delta^{-1},
\qquad
\|D^3P_\delta(u)\|\lesssim \delta^{-3/2}.
\]
Also,
\[
\|P_\delta(u)\|=\frac{\|u\|_2^2}{\|u\|_2^2+\delta}\le 1.
\]
This proves the lemma.
\end{proof}

\begin{prop}
    \label{prop:frechet-third}    
    Let
    \[
    \Phi_\delta(u,w):=\frac{|\langle u,w\rangle|^2}{\|u\|_2^2+\delta}.
    \]
    Then
    \[
    \|D_u^3\Phi_\delta(u,w)\|\lesssim \delta^{-3/2}\|w\|_2^2,\qquad
    \|D_u^2D_w\Phi_\delta(u,w)\|\lesssim\delta^{-1}\|w\|_2,\qquad
    \|D_uD_w^2\Phi_\delta(u,w)\|\lesssim\delta^{-1/2},
    \]
    and \(D_w^3\Phi_\delta(u,w)\equiv0\).
\end{prop}

\begin{proof}
Write
\[
\Phi_\delta(u,w)=w^*P_\delta(u)w.
\]
Since \(\Phi_\delta\) is quadratic in \(w\), we immediately have
\[
D_w^3\Phi_\delta(u,w)\equiv0.
\]

Also,
\[
D_w\Phi_\delta(u,w)[k]=2\Re\langle P_\delta(u)w,k\rangle,
\qquad
D_w^2\Phi_\delta(u,w)[k_1,k_2]=2\Re\langle P_\delta(u)k_1,k_2\rangle.
\]
Hence
\[
D_uD_w^2\Phi_\delta(u,w)[h,k_1,k_2]
=
2\Re\langle DP_\delta(u)[h]k_1,k_2\rangle,
\]
so
\[
\|D_uD_w^2\Phi_\delta(u,w)\|
\lesssim
\|DP_\delta(u)\|
\lesssim
\delta^{-1/2}
\]
by \Cref{lem:master-frechet}.

Similarly,
\[
D_u^2D_w\Phi_\delta(u,w)[h_1,h_2,k]
=
2\Re\langle D^2P_\delta(u)[h_1,h_2]w,k\rangle,
\]
and therefore
\[
\|D_u^2D_w\Phi_\delta(u,w)\|
\lesssim
\|D^2P_\delta(u)\|\,\|w\|_2
\lesssim
\delta^{-1}\|w\|_2.
\]

Finally,
\[
D_u^3\Phi_\delta(u,w)[h_1,h_2,h_3]
=
w^*D^3P_\delta(u)[h_1,h_2,h_3]w,
\]
whence
\[
\|D_u^3\Phi_\delta(u,w)\|
\lesssim
\|D^3P_\delta(u)\|\,\|w\|_2^2
\lesssim
\delta^{-3/2}\|w\|_2^2.
\]
This proves the proposition.
\end{proof}

With these technical facts established, we can move on to the main comparison of this subsection.

\begin{prop}[Smooth replacement]
\label{prop:smooth-replacement}
For every fixed $i\in[n]$ and every $\delta=\eta/N$,
\[
\Big|\mathop{\E}_{f\sim F_\mu} A_\delta^{(i)}-\mathop{\E}_{\textnormal{Gauss}(\mu)} A_{\delta}^{(i)}\Big|
\lesssim_\eta \|\mu\|_\infty^{1/2}.
\]
\end{prop}
\begin{proof}
    We make a Lindeberg replacement argument.
    Let $\xi_j\overset{\mathsf{iid}}{\sim}\mc{CN}(0,1)$, $j\in[N]$ and $\omega_j\overset{\mathsf{iid}}{\sim}\mathrm{Stein}$, $j\in[N]$.
    We will identify these indices $j$ with points $x\in \{0,1\}^n$; any bijection $\pi:[N]\to \{0,1\}^n$ will do.
    Define the ``hybrid'' joint distributions
    \[Z^{(m)}=(\xi_{1},\ldots, \xi_{m},\omega_{m+1},\ldots, \omega_{N}).\]

    Let us consider $A^{(i)}_\delta$ as an explicit function of the underlying randomness in any of these hybrids: write $A^{(i)}_\delta(Z^{(m)})$ to denote $A^{(i)}_\delta$ according to $f(\pi(k))=\sqrt{\mu(\pi(k))}\xi_{k}$ for all $k\leq m$ and $f(\pi(k))=\sqrt{\mu(\pi(k))}\omega_{k}$ for $k>m$.
    Then
    \begin{align*}
    \mathop{\E}_{f\sim F_\mu}A^{(i)}_\delta - \E_{\mathrm{Gauss}(\mu)}A_\delta^{(i)} =\E A^{(i)}_\delta(Z^{(0)})-\E A^{(i)}_\delta(Z^{(N)})
    = \sum_{m=1}^N\E\left[A^{(i)}_\delta(Z^{(m-1)})-A^{(i)}_\delta(Z^{(m)})\right],
    \end{align*}
    so it suffices to control the error of the increments.

    To that end, let us fix $m$, freeze variables except those for the $m$\textsuperscript{th} coordinate and define
    \begin{align*}
    h:\C\cong \mathbb R^2 &\to \mathbb R\\
    z\hspace{1.6em}&\mapsto A^{(i)}_\delta(\xi_1,\ldots, \xi_{m-1},z,\omega_{m+1},\ldots, \omega_{N})\,.
    \end{align*}
    Taking the second-order Taylor expansion of $h$ at $0$, we get
    \[h(z)=h(0)+Dh(0)[z]+\frac12 D^2h(0)[z,z]+R(z).\]
    The remainder $R$ satisfies
    \[R(z)=\int_0^1\frac{(1-t)^2}{2}D^3h(tz)[z,z,z]\,\mathrm{d}t,\qquad\text{so}\qquad |R(z)|\leq \frac12|z|^3\sup_{0\leq t\leq 1}\|D^3 h(tz)\|\,.\]

    Viewed as random vectors in $\mathbb R^2$, the $\xi$'s and $\omega$'s have matching means and covariance matrices,
    so the constant, linear, and quadratic terms have the same expectation.
    Thus:
    \begin{align*}
    \big|\E[h(\omega_{m})-h(\xi_{m})]\big|&\;=\;|\E R(\omega_m)-\E R(\xi_m)|\\
    &\;\leq\;\E|R(\omega_m)|+\E|R(\xi_m)|\\
    &\;\lesssim\;\E\left[\sup_{0\leq t\leq 1}\|D^3h(t\omega_m)\|\right]+\E\left[|\xi_{m}|^3\sup_{0\leq t\leq 1}\|D^3h(t\xi_m)\|\right]\,,
    \end{align*}
    so it remains to bound the third derivative of $h$.

    Recall that
    \[
    h(z)=\frac12\sum_s \Phi_\delta(u_s(z),w_s(z)),
    \]
    where $u_s(z) = \big((\wh{f}(s))(z),(\wh{f}(s+i))(z)\big)$ and $w_s(z) = \big((\wh{fq}(s))(z),(\wh{fq}(s+i))(z)\big)$.
    For each $s$, the dependence of $u_s$ and $w_s$ on the $m$th coordinate is affine, so we may make the definitions
    \[
    u_s(z)=:u_s^{[m]}+z\,a_{m,s},
    \qquad
    w_s(z)=:w_s^{[m]}+z\,b_{m,s},
    \]
    where
    \[
    a_{m,s}=\frac{\sqrt{\mu(\pi(m))}}{\sqrt N}\big(\chi_s(\pi(m)),\chi_{s+i}(\pi(m))\big),
    \qquad
    b_{m,s}=q(\pi(m))\,a_{m,s}.
    \]
    In particular,
    \[
    \|a_{m,s}\|_2\asymp \sqrt{\mu(\pi(m))/N},
    \qquad
    \|b_{m,s}\|_2\asymp \sqrt{\mu(\pi(m))/N}\,|q(\pi(m))|.
    \]
    
    With this notation, the chain rule and the bounds on the third Fr\'echet derivatives of $\Phi_\delta$ in \Cref{prop:frechet-third} give
    \begin{align*}
    \|D^3 h(z)\|
    &\lesssim
    \sum_s\Big(
    \|D_u^3\Phi_\delta(u_s(z),w_s(z))\|\,\|a_{m,s}\|_2^3\\
    &\hspace{4em}
    +\|D_u^2D_w\Phi_\delta(u_s(z),w_s(z))\|\,\|a_{m,s}\|_2^2\|b_{m,s}\|_2\\
    &\hspace{4em}
    +\|D_uD_w^2\Phi_\delta(u_s(z),w_s(z))\|\,\|a_{m,s}\|_2\|b_{m,s}\|_2^2
    \Big)\\
    &\lesssim
    \mu(\pi(m))^{3/2}\Bigg[
    \delta^{-3/2}N^{-3/2}\sum_s\|w_s(z)\|_2^2
    +\delta^{-1}|q(\pi(m))|\,N^{-3/2}\sum_s\|w_s(z)\|_2\\
    &\hspace{11em}
    +\delta^{-1/2}|q(\pi(m))|^2N^{-1/2}
    \Bigg].
    \end{align*}

    Write
\[
Y^{(m,\zeta)}:=(\xi_1,\dots,\xi_{m-1},\zeta,\omega_{m+1},\dots,\omega_N),
\]
so that \(h(\zeta)=A_\delta^{(i)}(Y^{(m,\zeta)})\). Then Parseval gives
\[
\sum_s \|w_s(\zeta)\|_2^2
=
2\sum_{k=1}^N \mu(\pi(k))|q(\pi(k))|^2\,|Y_k^{(m,\zeta)}|^2.
\]

For the Steinhaus interpolation segment, since $|\omega_m|=1$,
\[
\begin{aligned}
\sup_{0\leq t\leq1}
\sum_s\|w_s(t\omega_m)\|_2^2
&=
2\sum_{k<m}
\mu(\pi(k))|q(\pi(k))|^2|\xi_k|^2+
2\sum_{k\geq m}
\mu(\pi(k))|q(\pi(k))|^2.
\end{aligned}
\]
Taking expectations and using $\E|\xi_k|^2=1$ gives
\[
\begin{aligned}
\E\left[
    \sup_{0\leq t\leq1}
    \sum_s\|w_s(t\omega_m)\|_2^2
\right]
&=
2\sum_k
\mu(\pi(k))|q(\pi(k))|^2
=
2.
\end{aligned}
\]

For the Gaussian interpolation segment,
\[
\begin{aligned}
\sup_{0\leq t\leq1}
\sum_s\|w_s(t\xi_m)\|_2^2
&=
2\sum_{k<m}
\mu(\pi(k))|q(\pi(k))|^2|\xi_k|^2+
2\mu(\pi(m))|q(\pi(m))|^2|\xi_m|^2\\
&\qquad+
2\sum_{k>m}
\mu(\pi(k))|q(\pi(k))|^2.
\end{aligned}
\]
Therefore, by independence and the finiteness of the Gaussian moments,
\begin{align*}
\E\left[
    |\xi_m|^3
    \sup_{0\leq t\leq1}
    \sum_s\|w_s(t\xi_m)\|_2^2
\right]
&=
2\E|\xi_m|^3
\sum_{k<m}
\mu(\pi(k))|q(\pi(k))|^2\E|\xi_k|^2\\
&\hspace{3em}+
2\mu(\pi(m))|q(\pi(m))|^2\E|\xi_m|^5
\\
&\hspace{3em}+
2\E|\xi_m|^3
\sum_{k>m}
\mu(\pi(k))|q(\pi(k))|^2
\\
&\lesssim
\sum_k
\mu(\pi(k))|q(\pi(k))|^2
\\
&=
1.
\end{align*}

Also, by Cauchy--Schwarz,
\[
    \sum_s\|w_s(z)\|_2
    \leq
    \sqrt N
    \left(
        \sum_s\|w_s(z)\|_2^2
    \right)^{1/2}.
\]
Consequently,
\begin{align*}
\E\left[
    \sup_{0\leq t\leq1}
    \sum_s\|w_s(t\omega_m)\|_2
\right]
&\leq
\sqrt N
\left(
    \E\left[
        \sup_{0\leq t\leq1}
        \sum_s\|w_s(t\omega_m)\|_2^2
    \right]
\right)^{1/2}\lesssim
\sqrt N,
\end{align*}
and
\begin{align*}
\E\left[
    |\xi_m|^3
    \sup_{0\leq t\leq1}
    \sum_s\|w_s(t\xi_m)\|_2
\right]&\leq
\sqrt N
\left(
    \E|\xi_m|^3
\right)^{1/2}
\left(
    \E\left[
        |\xi_m|^3
        \sup_{0\leq t\leq1}
        \sum_s\|w_s(t\xi_m)\|_2^2
    \right]
\right)^{1/2}
\\
&\lesssim
\sqrt N.
\end{align*}

    Substituting $\delta=\eta/N$, we conclude that
    \[
    \big|\E[h(\omega_m)-h(\xi_m)]\big|
    \lesssim_\eta \mu(\pi(m))^{3/2}\bigl(1+|q(\pi(m))|^2\bigr).
    \]
    Summing over $m$ in the Lindeberg telescoping identity yields
    \[
    \Big|\mathop{\E}_{f\sim F_\mu} A_\delta^{(i)}-\mathop{\E}_{\mathrm{Gauss}(\mu)} A_\delta^{(i)}\Big|
    \lesssim_\eta
    \sum_x \mu(x)^{3/2}(1+|q(x)|^2).
    \]
    Finally,
    \[
    \sum_x \mu(x)^{3/2}\le \|\mu\|_\infty^{1/2},
    \qquad
    \sum_x \mu(x)^{3/2}|q(x)|^2
    \le
    \|\mu\|_\infty^{1/2}\sum_x \mu(x)|q(x)|^2
    =
    \|\mu\|_\infty^{1/2},
    \]
    so
    \[
    \Big|\mathop{\E}_{f\sim F_\mu} A_\delta^{(i)}-\mathop{\E}_{\mathrm{Gauss}(\mu)} A_\delta^{(i)}\Big|
    \lesssim_\eta \|\mu\|_\infty^{1/2}.
    \]
    This proves the proposition.
\end{proof}

\subsubsection{Bounding $\E B_\delta^{(i)}$}
This subsection is devoted to the following:
\begin{prop}
	\label{lem:EB-bound}
	For any $p\in (0,1)$ and with $\delta=\eta/N$,
	\[\mathop{\E}_{f\sim F_\mu} B_\delta \lesssim \eta^{p/4}.\]
\end{prop}

\begin{proof}
To begin, separate the $w$'s and the $u$'s with Cauchy--Schwarz:
\begin{align*}
    B_\delta = \sum_s \|w_s\|_2^2\cdot\frac{\delta}{\|u_s\|_2^2+\delta} \leq \left(\sum_s\|w_s\|_2^4\right)^{1/2}\left(\sum_s\left(\frac{\delta}{ \|u_s\|_2^2+\delta}\right)^2\right)^{1/2}
\end{align*}
so that (after another Cauchy--Schwarz),
\begin{align*}\E B_\delta &\leq \left(\E\sum_s\|w_s\|_2^4 \right)^{1/2}\left(\E\sum_s\left(\frac{\delta}{\|u_s\|_2^2+\delta}\right)^2 \right)^{1/2}\\
&\leq \left(\E\sum_s\|w_s\|_2^4 \right)^{1/2}\left(\E\sum_s\frac{\delta}{\|u_s\|_2^2+\delta}\right)^{1/2}
\end{align*}
The first term is controlled again by \Cref{lem:w4}: $\E\sum_s\|w_s\|_2^4\lesssim1/N $.
For the second term, note that for all $y>0$ and $0<p<2$,
    \[\frac{1}{1+y}\leq y^{-p/2}.\]
    Thus,
	\begin{align*}
		\E\frac{\delta}{\|u_s\|_2^2+\delta} &=\E\frac{\eta}{N\|u_s\|_2^2+\eta}\\
        &\leq \eta^{p/2}\E\left[(N\|u_s\|_2^2)^{-p/2}\right]\\
        &\leq \eta^{p/2}\E\left[(N|\wh f(s)|^2)^{-p/2}\right]\\
        &= \left(\frac{\eta}{N}\right)^{p/2}\E\left[|\wh f(s)|^{-p}\right],
	\end{align*}
    where we have used $\|u_s\|_2^2=|\wh f(s)|^2+|\wh f(s+i)|^2\geq |\wh f(s)|^2.$
    Now by the negative-moment Khintchine inequality for Steinhaus sums \cite{rapaport2025negativemomentssteinhaussums}, we have for exponent $q=-p\in(-1,0)$ that,
	\[\left(\E |\wh f(s)|^q\right)^{1/q} \geq C_q \left(\E |\wh f(s)|^2\right)^{1/2} = C_q\left(\frac{1}{N} \E\tsum_{x}\mu(x)\right)^{1/2}=\frac{C_q}{\sqrt{N}}\]
    where $C_q>0$ is an absolute constant depending on $q$ only.
    Therefore $\E|\wh f(s)|^{-p}\leq C_{-p}^{-p}N^{p/2}$ and so
    \[\E\frac{\delta}{\|u_s\|_2^2+\delta}\lesssim_p \eta^{p/2}\]
    Combining these bounds, we conclude:
    \[\E B^{(i)}_\delta\lesssim_p \sqrt{\frac{1}{N}}\sqrt{N\eta^{p/2}}=\eta^{p/4}\,.\qedhere\]
\end{proof}

\subsubsection{Putting it together}

\begin{proof}[Proof of \Cref{lem:fixed-q-expect}]
Fix $t>0$.
Taking $p=1/2$ in \Cref{lem:EB-bound}, there is a universal constant $C>0$
such that, whenever $\delta=\eta/N$,
\[
    \E_{f\sim F_\mu}B_\delta^{(i)}
    \leq
    C\eta^{1/8}.
\]
Choose $\eta\in(0,1]$ sufficiently small that
\[
    C\eta^{1/8}
    \leq
    \frac{t}{3},
\]
and set
\[
    \delta
    :=
    \frac{\eta}{N}.
\]

On the event where \Cref{lem:expect-A-Gauss} holds and
\[
    \|\mu\|_\infty
    =
    N^{-1+o(1)},
\]
the following estimates hold simultaneously for every $\mu$-compatible $q$:
\begin{align*}
\mathop{\E}_{f\sim F_\mu}
\Avg_{\wh f}[\wh{fq}]
&=
\mathop{\E}_{f\sim F_\mu}
\left[
    \frac1n\sum_i A^{(i)}
\right]
\\
&\leq
\mathop{\E}_{f\sim F_\mu}
\left[
    \frac1n\sum_i
    \left(
        A_\delta^{(i)}+B_\delta^{(i)}
    \right)
\right]
\\
&\leq
\mathop{\E}_{f\sim\mathrm{Gauss}(\mu)}
\left[
    \frac1n\sum_i A_\delta^{(i)}
\right]+
\frac1n\sum_i
\left|
    \mathop{\E}_{f\sim F_\mu}A_\delta^{(i)}
    -
    \mathop{\E}_{f\sim\mathrm{Gauss}(\mu)}A_\delta^{(i)}
\right|
+
\frac1n\sum_i
\mathop{\E}_{f\sim F_\mu}B_\delta^{(i)}
\\
&\leq
\mathop{\E}_{f\sim\mathrm{Gauss}(\mu)}
\left[
    \frac1n\sum_i A^{(i)}
\right]
+
C_\eta\|\mu\|_\infty^{1/2}
+
C\eta^{1/8}
\\
&\leq
\frac12
+
\mc O\left(\frac1n\right)
+
C_\eta\|\mu\|_\infty^{1/2}
+
C\eta^{1/8}.
\end{align*}
Here the first inequality uses that $A^{(i)}
    \leq A_\delta^{(i)}+B_\delta^{(i)}$
and the third uses that
$A_\delta^{(i)}
    \leq
    A^{(i)}$
pointwise in the Gaussian model.

The parameter $\eta$ is now fixed.
Therefore, for all sufficiently large $n$,
\[
    \mc O\left(\frac1n\right)
    +
    C_\eta\|\mu\|_\infty^{1/2}
    \leq
    \frac{2t}{3}.
\]
Together with the choice of $\eta$, this gives
\[
    \mathop{\E}_{f\sim F_\mu}
    \Avg_{\wh f}[\wh{fq}]
    \leq
    \frac12+t.
\]
The event used above has probability $1-\exp(-\Omega(n))$ over
$\mu\sim\bs\mu$, which proves the lemma.
\end{proof}

\subsection{ Concentrating $\Avg_{\wh f}[\wh{fq}]$}
\label{sec:fixed-q-conc}

Now that we have shown $\E_{f\sim F_\mu}\Avg_{\wh f}[\wh{fq}]\leq 1/2 + o(1)$, we turn to concentration.
We will show doubly-exponential concentration (in $n$) for each $A^{(i)}$ individually, and then a union bound over coordinates to finish the argument.

Recall our goal is to get an $\exp(-cN^\alpha)$ tail bound on the deviation of the following quantity from its mean
\[A^{(i)}=\sum_s \frac{|\langle u_s,w_s\rangle|^2}{\|u_s\|_2^2}.\]
As mentioned above, $A^{(i)}$ seems to frustrate naive approaches because it does not have a nice gradient.

\begin{theorem}
	There exists a universal $c>0$ such that
	\[\Pr[|A^{(i)}-\E A^{(i)}|>t]\leq \exp(-ct^6N^{1/3})\]
\end{theorem}

As before, we work with the following smoothed quantity, which has well-behaved derivatives:
\[A_\delta^{(i)} := \frac12\sum_s\frac{|\langle u_s,w_s\rangle|^2}{\|u_s\|_2^2+\delta}\]
Put $D_\delta = A-A_\delta$.
Now $D_\delta$ is not amenable to Herbst, but we make the observation that pointwise,
\[D_\delta = \frac12\sum_s|\langle u_s, w_s\rangle|^2\left(\frac1{\|u_s\|_2^2}-\frac1{\|u_s\|_2^2+\delta}\right)\overset{\text{(Cauchy-Schwarz)}}{\leq}\frac12\sum_s\frac{\delta\|w_s\|_2^2}{\|u_s\|_2^2+\delta}=:B_\delta,\]
and the gradient of $B_\delta$ is much friendlier.

Finally, note that from our definitions so far,
\[|A-\E A|\leq |A_\delta - \E A_\delta| + |B_\delta - \E B_\delta| + 2\E B_\delta.\]
So we will show each of these terms are concentrated in turn.

\subsubsection{Concentrating $A^{(i)}_\delta$}
A Herbst argument.
We begin with the relevant carr\'e du champ, where the gradient is taken with
respect to the independent phase coordinates $(\theta_x)_x$.

\begin{prop}[Carr\'e du champ bound]
\label{prop:A-delta-carre-du-champ}
For every $i\in[n]$,
\begin{equation}
\label{ineq:A-delta-carre-du-champ}
\begin{aligned}
    \Gamma\left(A_\delta^{(i)}\right)
    &:=
    \left\|\nabla A_\delta^{(i)}\right\|_2^2
    =
    \sum_x
    \left|
        \partial_{\theta_x}A_\delta^{(i)}
    \right|^2\lesssim
    \frac{\|\mu\|_\infty}{\delta}
    \sum_s\|w_s\|_2^4
    +
    \|\mu|q|^2\|_\infty.
\end{aligned}
\end{equation}
\end{prop}

\begin{proof}
Define
\[
    \Phi(u_s,w_s)
    :=
    \frac{|\langle u_s,w_s\rangle|^2}
    {\|u_s\|_2^2+\delta},
\]
so that
\[
    A_\delta^{(i)}
    =
    \frac12\sum_s\Phi(u_s,w_s).
\]
By the multivariate chain rule,
\[
\partial_{\theta_x}\Phi(u_s,w_s)
=
\Re\left\langle
    \nabla_u^\C\Phi(u_s,w_s),
    \partial_{\theta_x}u_s
\right\rangle
+
\Re\left\langle
    \nabla_w^\C\Phi(u_s,w_s),
    \partial_{\theta_x}w_s
\right\rangle.
\]
Here $\nabla_u^\C$ and $\nabla_w^\C$ denote the usual complex gradients,
equivalently $2\partial_{\bar u}$ and $2\partial_{\bar w}$.

By direct calculation,
\begin{align*}
\nabla_u^\C\Phi(u_s,w_s)
&=
\frac{
    2\overline{\langle u_s,w_s\rangle}
}{
    \|u_s\|_2^2+\delta
}
w_s
-
\frac{
    2|\langle u_s,w_s\rangle|^2
}{
    \left(\|u_s\|_2^2+\delta\right)^2
}
u_s,
\\
\nabla_w^\C\Phi(u_s,w_s)
&=
\frac{
    2\langle u_s,w_s\rangle
}{
    \|u_s\|_2^2+\delta
}
u_s.
\end{align*}
Consequently,
\begin{equation}
\label{ineq:local-complex-gradient-bounds}
    \left\|
        \nabla_u^\C\Phi(u_s,w_s)
    \right\|_2
    \lesssim
    \frac{\|w_s\|_2^2}{\sqrt\delta},
    \qquad
    \left\|
        \nabla_w^\C\Phi(u_s,w_s)
    \right\|_2
    \lesssim
    \|w_s\|_2.
\end{equation}

Write
\[
    \chi_s(x)
    :=
    (-1)^{s\cdot x}.
\]
Since
\[
    \wh f(s)
    =
    \frac1{\sqrt N}
    \sum_y
    \sqrt{\mu(y)}e^{\mathrm{i}\theta_y}\chi_s(y)
\]
and $g=fq$, we have
\[
    \partial_{\theta_x}\wh f(s)
    =
    \mathrm{i}
    \sqrt{\frac{\mu(x)}{N}}
    e^{\mathrm{i}\theta_x}
    \chi_s(x)
\]
and
\[
    \partial_{\theta_x}\wh g(s)
    =
    q(x)\partial_{\theta_x}\wh f(s).
\]
It follows that
\[
    \partial_{\theta_x}u_s
    =
    \mathrm{i}
    \sqrt{\frac{\mu(x)}{N}}
    e^{\mathrm{i}\theta_x}
    \chi_s(x)
    \bigl(1,(-1)^{x_i}\bigr)
\]
and
\[
    \partial_{\theta_x}w_s
    =
    q(x)\partial_{\theta_x}u_s.
\]
Substituting these identities into the chain rule and factoring out the common
scalar of modulus $\sqrt{\mu(x)/N}$ gives
\begin{align}
\label{ineq:exact-phase-gradient-bound}
\left|
    \partial_{\theta_x}A_\delta^{(i)}
\right|
&\le
\frac12
\sqrt{\frac{\mu(x)}{N}}
\left|
    \sum_s
    \chi_s(x)
    \left\langle
        \nabla_u^\C\Phi(u_s,w_s),
        \bigl(1,(-1)^{x_i}\bigr)
    \right\rangle
\right|
\nonumber\\
&\quad+
\frac12
\sqrt{\frac{\mu(x)}{N}}
|q(x)|
\left|
    \sum_s
    \chi_s(x)
    \left\langle
        \nabla_w^\C\Phi(u_s,w_s),
        \bigl(1,(-1)^{x_i}\bigr)
    \right\rangle
\right|.
\end{align}

For any scalar family $(a_s)_s$ and every $b\in\{0,1\}$,
Parseval on the slice $\{x:x_i=b\}$ gives
\begin{equation}
\label{eq:restricted-walsh-parseval}
\begin{aligned}
\sum_{x:x_i=b}
\left|
    \sum_s\chi_s(x)a_s
\right|^2
&=
\frac N2
\sum_{s:s_i=0}
\left|
    a_s+(-1)^ba_{s+i}
\right|^2\le
N\sum_s|a_s|^2.
\end{aligned}
\end{equation}
Applying \eqref{eq:restricted-walsh-parseval} on the two slices $x_i=0$ and
$x_i=1$, and using $
\left\|
    \bigl(1,(-1)^{x_i}\bigr)
\right\|_2^2
=
2$,
we obtain
\begin{equation}
\label{ineq:u-gradient-walsh-sum}
\begin{aligned}
&
\sum_x
\left|
    \sum_s
    \chi_s(x)
    \left\langle
        \nabla_u^\C\Phi(u_s,w_s),
        \bigl(1,(-1)^{x_i}\bigr)
    \right\rangle
\right|^2\le
4N
\sum_s
\left\|
    \nabla_u^\C\Phi(u_s,w_s)
\right\|_2^2
\end{aligned}
\end{equation}
and likewise
\begin{equation}
\label{ineq:w-gradient-walsh-sum}
\begin{aligned}
&
\sum_x
\left|
    \sum_s
    \chi_s(x)
    \left\langle
        \nabla_w^\C\Phi(u_s,w_s),
        \bigl(1,(-1)^{x_i}\bigr)
    \right\rangle
\right|^2\le
4N
\sum_s
\left\|
    \nabla_w^\C\Phi(u_s,w_s)
\right\|_2^2\,.
\end{aligned}
\end{equation}

Squaring \eqref{ineq:exact-phase-gradient-bound}, summing over $x$, and using
\eqref{ineq:u-gradient-walsh-sum} and
\eqref{ineq:w-gradient-walsh-sum}, we conclude that
\begin{align*}
\left\|\nabla A_\delta^{(i)}\right\|_2^2
&\lesssim
\|\mu\|_\infty
\sum_s
\left\|
    \nabla_u^\C\Phi(u_s,w_s)
\right\|_2^2+
\|\mu|q|^2\|_\infty
\sum_s
\left\|
    \nabla_w^\C\Phi(u_s,w_s)
\right\|_2^2.
\end{align*}
Using \eqref{ineq:local-complex-gradient-bounds}, this becomes
\[
\left\|\nabla A_\delta^{(i)}\right\|_2^2
\lesssim
\frac{\|\mu\|_\infty}{\delta}
\sum_s\|w_s\|_2^4
+
\|\mu|q|^2\|_\infty
\sum_s\|w_s\|_2^2.
\]
Finally, Parseval and $\mu$-compatibility give
\begin{align*}
\sum_s\|w_s\|_2^2
&=
\sum_s
\left(
    |\wh g(s)|^2
    +
    |\wh g(s+i)|^2
\right)
\\
&=
2\|g\|_2^2
\\
&=
2\sum_x\mu(x)|q(x)|^2
\\
&=
2.
\end{align*}
Absorbing this factor into the implicit constant proves
\eqref{ineq:A-delta-carre-du-champ}.
\end{proof}

Now $\delta$ will be chosen like $\delta=\eta/N$, so we need $W:=\sum_s\|w_s\|_2^4\lesssim\frac1N$ with high probability.
Fortunately, $W$ is nothing but a Steinhaus chaos of fixed degree, and its concentration is given by \Cref{lem:w4}.
We shall combine \Cref{prop:A-delta-carre-du-champ} and \Cref{lem:w4} using the
following localized form of the Herbst argument.

\begin{fact}[Localized Herbst bound for Steinhaus variables]
\label{fact:herbst}
There exists a universal constant $c>0$ such that the following holds.
Let $F$ be a centered smooth function of independent Steinhaus variables, and
let $G$ be an event such that
\[
    |F|\leq1,\qquad 
    \|\nabla F\|_2^2\leq b
    \quad\text{on }G,\qquad \text{and}\qquad 
    \|\nabla F\|_2^2\leq B
    \quad\text{everywhere}.
\]
If
\[
    B\Pr(G^c)<b,
\]
then, for every
\[
    0<t
    \leq
    cb\log\left(
        \frac{b}{B\Pr(G^c)}
    \right),
\]
one has
\[
    \Pr[|F|>t]
    \leq
    2\exp\left(
        -\frac{ct^2}{b}
    \right).
\]
\end{fact}

\begin{proof}
For $\lambda\geq0$, write
\[
    M(\lambda)
    :=
    \E e^{\lambda F},
    \qquad
    \psi(\lambda)
    :=
    \log M(\lambda).
\]
The product Steinhaus log-Sobolev inequality \cite{WEISSLER1980218} and the standard Herbst
calculation (see, \textit{e.g.,} \cite{ledoux}) give
\[
    \left(
        \frac{\psi(\lambda)}{\lambda}
    \right)'
    \leq
    C
    \frac{
        \E\left[
            e^{\lambda F}\|\nabla F\|_2^2
        \right]
    }{
        \E e^{\lambda F}
    }.
\]
Since $F$ is centered, Jensen's inequality gives $\E e^{\lambda F}\geq1$.
Since $|F|\leq1$,
\[
    \frac{
        \E\left[
            e^{\lambda F}\mathbf 1_{G^c}
        \right]
    }{
        \E e^{\lambda F}
    }
    \leq
    e^\lambda\Pr(G^c).
\]
It follows that
\[
    \left(
        \frac{\psi(\lambda)}{\lambda}
    \right)'
    \leq
    C\left(
        b+B\Pr(G^c)e^\lambda
    \right).
\]
Integrating from $0$ to $\lambda$, and using
\[
    e^\lambda-1\leq\lambda e^\lambda,
\]
gives
\[
    \psi(\lambda)
    \leq
    C\lambda^2
    \left(
        b+B\Pr(G^c)e^\lambda
    \right).
\]
Therefore, whenever
\[
    0\leq\lambda
    \leq
    \log\left(
        \frac{b}{B\Pr(G^c)}
    \right),
\]
we have
\[
    \psi(\lambda)
    \leq
    Cb\lambda^2.
\]
Chernoff's bound, with $\lambda=ct/b$ for a sufficiently small universal
constant $c>0$, now gives
\[
    \Pr[F>t]
    \leq
    \exp\left(
        -\frac{ct^2}{b}
    \right)
\]
throughout the asserted range of $t$.
Applying the same argument to $-F$ proves the result.
\end{proof}

\begin{cor}
\label{cor:A-delta-concentration}
There exist universal constants $c,C>0$ such that the following holds.
Let $q$ be $\mu$-compatible, let $\eta\in(0,1]$, and set
\[
    \delta
    =
    \frac{\eta}{N}.
\]
Suppose that
\[
    \left(
        \frac{\eta}{\|\mu\|_\infty}
    \right)^{1/3}
    \geq
    C\log N.
\]
Then, for every $i\in[n]$ and every $0<t\leq1$,
\begin{equation}
\label{ineq:A-delta-concentration}
\Pr_{f\sim F_\mu}\left[
    \left|
        A_\delta^{(i)}
        -
        \E_{f\sim F_\mu}A_\delta^{(i)}
    \right|
    >
    t
\right]
\leq
2\exp\left(
    -
    \frac{ct^2}{
        \left(
            \|\mu\|_\infty/\eta
        \right)^{1/3}
        +
        \|\mu|q|^2\|_\infty
    }
\right).
\end{equation}
\end{cor}

\begin{proof}
Put
\[
    F
    :=
    A_\delta^{(i)}
    -
    \E_{f\sim F_\mu}A_\delta^{(i)}.
\]
Since
\[
    0
    \leq
    A_\delta^{(i)}
    \leq
    \frac12\sum_s\|w_s\|_2^2
    =
    1,
\]
we have
\[
    |F|\leq1.
\]

Let
\[
    W
    :=
    \sum_s\|w_s\|_2^4.
\]
Let $K\geq1$ be a universal constant to be fixed below, and define
\begin{equation}
\label{eq:A-delta-good-event}
    G
    :=
    \left\{
        W
        \leq
        \frac8N
        +
        \frac{K}{N}
        \left(
            \frac{\eta}{\|\mu\|_\infty}
        \right)^{2/3}
    \right\}.
\end{equation}
Since $\E W\leq8/N$, the hypothesis of the corollary ensures, after
increasing $C$ if necessary, that the deviation in
\eqref{eq:A-delta-good-event} lies in the range of \Cref{lem:w4}.
Consequently,
\begin{equation}
\label{ineq:A-delta-good-event-probability}
    \Pr_{f\sim F_\mu}[G^c]
    \leq
    \exp\left(
        -c\sqrt K
        \left(
            \frac{\eta}{\|\mu\|_\infty}
        \right)^{1/3}
    \right).
\end{equation}

On $G$, \Cref{prop:A-delta-carre-du-champ} gives
\begin{align}
\label{ineq:A-delta-gradient-good-event}
\left\|\nabla A_\delta^{(i)}\right\|_2^2
&\lesssim
\frac{N\|\mu\|_\infty}{\eta}
\left[
    \frac8N
    +
    \frac{K}{N}
    \left(
        \frac{\eta}{\|\mu\|_\infty}
    \right)^{2/3}
\right]
+
\|\mu|q|^2\|_\infty
\nonumber\\
&=
8\frac{\|\mu\|_\infty}{\eta}
+
K
\left(
    \frac{\|\mu\|_\infty}{\eta}
\right)^{1/3}
+
\|\mu|q|^2\|_\infty
\nonumber\\
&\lesssim
(K+8)
\left(
    \frac{\|\mu\|_\infty}{\eta}
\right)^{1/3}
+
\|\mu|q|^2\|_\infty.
\end{align}
In the last line we used
\[
    \frac{\|\mu\|_\infty}{\eta}
    \leq
    1,
\]
which follows from the hypothesis of the corollary.

We also need the everywhere bound.
By Parseval, $\sum_s\|w_s\|_2^2
    =
    2$,
and hence
\[
    W
    \leq
    \left(
        \sum_s\|w_s\|_2^2
    \right)^2
    =
    4.
\]
Therefore, everywhere,
\begin{equation}
\label{ineq:A-delta-gradient-global}
    \left\|\nabla A_\delta^{(i)}\right\|_2^2
    \lesssim
    \frac{4N\|\mu\|_\infty}{\eta}
    +
    \|\mu|q|^2\|_\infty.
\end{equation}

By \eqref{ineq:A-delta-gradient-good-event} and
\eqref{ineq:A-delta-gradient-global}, there is a universal constant $C_0>0$
such that
\[
    \|\nabla F\|_2^2
    \leq
    b
    \qquad\text{on }G
\]
and
\[
    \|\nabla F\|_2^2
    \leq
    B
    \qquad\text{everywhere},
\]
where
\[
\begin{aligned}
    b
    &:=
    C_0
    \left[
        (K+8)
        \left(
            \frac{\|\mu\|_\infty}{\eta}
        \right)^{1/3}
        +
        \|\mu|q|^2\|_\infty
    \right],
    \\
    B
    &:=
    C_0
    \left[
        \frac{4N\|\mu\|_\infty}{\eta}
        +
        \|\mu|q|^2\|_\infty
    \right].
\end{aligned}
\]
Because $q$ is $\mu$-compatible,
\[
    \|\mu|q|^2\|_\infty
    \leq
    \sum_x\mu(x)|q(x)|^2
    =
    1.
\]
Using also $\|\mu\|_\infty/\eta\leq1$, we obtain
\[
    \frac{B}{b}
    \leq
    C\left[
        N
        \left(
            \frac{\|\mu\|_\infty}{\eta}
        \right)^{2/3}
        +
        1
    \right]
    \leq
    CN.
\]
Combining this with
\eqref{ineq:A-delta-good-event-probability} gives
\begin{align*}
\log\left(
    \frac{b}{
        B\Pr_{f\sim F_\mu}(G^c)
    }
\right)
&\geq
c\sqrt K
\left(
    \frac{\eta}{\|\mu\|_\infty}
\right)^{1/3}
-
C\log N
\\
&\geq
\frac{c\sqrt K}{2}
\left(
    \frac{\eta}{\|\mu\|_\infty}
\right)^{1/3},
\end{align*}
after increasing the constant in the hypothesis of the corollary.
In particular,
\[
    B\Pr_{f\sim F_\mu}(G^c)<b.
\]

Moreover,
\[
    b
    \geq
    C_0K
    \left(
        \frac{\|\mu\|_\infty}{\eta}
    \right)^{1/3},
\]
and therefore
\[
    cb
    \log\left(
        \frac{b}{
            B\Pr_{f\sim F_\mu}(G^c)
        }
    \right)
    \geq
    cK^{3/2}.
\]
We now fix the universal constant $K$ large enough that the right-hand side
is at least $1$.
Thus the admissible range in \Cref{fact:herbst} contains every
$t\in(0,1]$.

Applying \Cref{fact:herbst}, we obtain
\[
    \Pr_{f\sim F_\mu}[|F|>t]
    \leq
    2\exp\left(
        -\frac{ct^2}{b}
    \right).
\]
Since $K$ is now a fixed universal constant,
\[
    b
    \lesssim
    \left(
        \frac{\|\mu\|_\infty}{\eta}
    \right)^{1/3}
    +
    \|\mu|q|^2\|_\infty,
\]
which proves \eqref{ineq:A-delta-concentration}.
\end{proof}

\subsubsection{Concentrating $B^{(i)}_\delta$}

The same carr\'e-du-champ estimate holds for the error dominator.

\begin{cor}
\label{cor:B-delta-carre-du-champ}
For every $i\in[n]$,
\[
    \left\|\nabla B_\delta^{(i)}\right\|_2^2
    \lesssim
    \frac{\|\mu\|_\infty}{\delta}
    \sum_s\|w_s\|_2^4
    +
    \|\mu|q|^2\|_\infty.
\]
\end{cor}

\begin{proof}
Define
\[
    \Phi(u_s,w_s)
    :=
    \frac{\delta\|w_s\|_2^2}
    {\|u_s\|_2^2+\delta},
\]
so that
\[
    B_\delta^{(i)}
    =
    \frac12\sum_s\Phi(u_s,w_s).
\]
By direct calculation,
\[
    \nabla_u^\C\Phi(u_s,w_s)
    =
    -
    \frac{
        2\delta\|w_s\|_2^2
    }{
        \left(\|u_s\|_2^2+\delta\right)^2
    }
    u_s
\]
and
\[
    \nabla_w^\C\Phi(u_s,w_s)
    =
    \frac{
        2\delta
    }{
        \|u_s\|_2^2+\delta
    }
    w_s.
\]
Hence
\[
    \left\|
        \nabla_u^\C\Phi(u_s,w_s)
    \right\|_2
    \lesssim
    \frac{\|w_s\|_2^2}{\sqrt\delta},
    \qquad
    \left\|
        \nabla_w^\C\Phi(u_s,w_s)
    \right\|_2
    \lesssim
    \|w_s\|_2.
\]
The phase derivatives of $u_s$ and $w_s$ are unchanged, so the
Parseval argument from
\Cref{prop:A-delta-carre-du-champ} applies without further change.
\end{proof}

We therefore obtain the same concentration estimate as for
$A_\delta^{(i)}$.

\begin{cor}[Concentration of the error dominator]
\label{cor:B-delta-concentration}
There exist universal constants $c,C>0$ such that the following holds.
Let $q$ be $\mu$-compatible, let $\eta\in(0,1]$, and set
\[
    \delta
    =
    \frac{\eta}{N}.
\]
Suppose that
\[
    \left(
        \frac{\eta}{\|\mu\|_\infty}
    \right)^{1/3}
    \geq
    C\log N.
\]
Then, for every $i\in[n]$ and every $0<t\leq1$,
\begin{equation}
\label{ineq:B-delta-concentration}
\Pr_{f\sim F_\mu}\left[
    \left|
        B_\delta^{(i)}
        -
        \E_{f\sim F_\mu}B_\delta^{(i)}
    \right|
    >
    t
\right]
\leq
2\exp\left(
    -
    \frac{ct^2}{
        \left(
            \|\mu\|_\infty/\eta
        \right)^{1/3}
        +
        \|\mu|q|^2\|_\infty
    }
\right).
\end{equation}
\end{cor}

\begin{proof}
Put
\[
    F
    :=
    B_\delta^{(i)}
    -
    \E_{f\sim F_\mu}B_\delta^{(i)}.
\]
Since
\[
    0
    \leq
    B_\delta^{(i)}
    \leq
    \frac12\sum_s\|w_s\|_2^2
    =
    1,
\]
we have
\[
    |F|\leq1.
\]

Let
\[
    W
    :=
    \sum_s\|w_s\|_2^4.
\]
Choose the same sufficiently large universal constant $K$ as in the proof of
\Cref{cor:A-delta-concentration}, and define
\[
    G
    :=
    \left\{
        W
        \leq
        \frac8N
        +
        \frac{K}{N}
        \left(
            \frac{\eta}{\|\mu\|_\infty}
        \right)^{2/3}
    \right\}.
\]
As in that proof, \Cref{lem:w4} gives
\[
    \Pr_{f\sim F_\mu}[G^c]
    \leq
    \exp\left(
        -c\sqrt K
        \left(
            \frac{\eta}{\|\mu\|_\infty}
        \right)^{1/3}
    \right).
\]

On $G$, \Cref{cor:B-delta-carre-du-champ} gives
\[
    \left\|\nabla B_\delta^{(i)}\right\|_2^2
    \lesssim
    (K+8)
    \left(
        \frac{\|\mu\|_\infty}{\eta}
    \right)^{1/3}
    +
    \|\mu|q|^2\|_\infty.
\]
Moreover, since
\[
    W
    \leq
    \left(
        \sum_s\|w_s\|_2^2
    \right)^2
    =
    4,
\]
we have the global bound
\[
    \left\|\nabla B_\delta^{(i)}\right\|_2^2
    \lesssim
    \frac{4N\|\mu\|_\infty}{\eta}
    +
    \|\mu|q|^2\|_\infty.
\]

These are exactly the good-event gradient bound, global gradient bound, and
exceptional-event estimate used in the proof of
\Cref{cor:A-delta-concentration}.
Together with $|F|\leq1$, they verify the hypotheses of
\Cref{fact:herbst} with the same choices of $b$, $B$, and $K$.
Applying that fact proves \eqref{ineq:B-delta-concentration}.
\end{proof}

\subsubsection{Combining the bounds}

We now remove the smoothing and deduce concentration of $A^{(i)}$.

\begin{prop}[Per-coordinate concentration for flat $q$]
\label{prop:A-i-concentration}
There exist universal constants $c,C>0$ such that the following holds.
Let $t\in(0,1]$, let $q$ be $\mu$-compatible, and suppose that
\[
    \|\mu|q|^2\|_\infty
    \le
    \|\mu\|_\infty^{1/3}
\]
and
\[
    \left(
        \frac{t^8}{\|\mu\|_\infty}
    \right)^{1/3}
    \geq
    C\log N.
\]
Then, for every $i\in[n]$,
\[
\Pr_{f\sim F_\mu}\left[
    \left|
        A^{(i)}
        -
        \E_{f\sim F_\mu}A^{(i)}
    \right|
    >
    t
\right]
\leq
4\exp\left(
    -ct^6\|\mu\|_\infty^{-1/3}
\right).
\]
\end{prop}

\begin{proof}
Take $p=1/2$ in \Cref{lem:EB-bound}, so that for $\delta=\eta/N$,
\[
    \E_{f\sim F_\mu}B_\delta^{(i)}
    \leq
    C\eta^{1/8}
\]
for a universal constant $C>0$.
Choose a sufficiently small universal constant $c_0>0$ and set
\[
    \eta
    :=
    c_0t^8,
    \qquad
    \delta
    :=
    \frac{\eta}{N},
\]
so that
\[
    2\E_{f\sim F_\mu}B_\delta^{(i)}
    \leq
    \frac t2.
\]
After increasing the constant in the hypothesis of the proposition, we also
have
\[
    \left(
        \frac{\eta}{\|\mu\|_\infty}
    \right)^{1/3}
    \geq
    C\log N,
\]
so \Cref{cor:A-delta-concentration} and
\Cref{cor:B-delta-concentration} apply.

Recall that
\[
\begin{aligned}
\left|
    A^{(i)}-\E A^{(i)}
\right|
\leq\;&
\left|
    A_\delta^{(i)}-\E A_\delta^{(i)}
\right|+
\left|
    B_\delta^{(i)}-\E B_\delta^{(i)}
\right|
+
2\E B_\delta^{(i)}.
\end{aligned}
\]
It follows that
\begin{align*}
&
\Pr_{f\sim F_\mu}\left[
    \left|
        A^{(i)}-\E A^{(i)}
    \right|>t
\right]\leq
\Pr_{f\sim F_\mu}\left[
    \left|
        A_\delta^{(i)}-\E A_\delta^{(i)}
    \right|>\frac t4
\right]
+
\Pr_{f\sim F_\mu}\left[
    \left|
        B_\delta^{(i)}-\E B_\delta^{(i)}
    \right|>\frac t4
\right]
\\
&\qquad\leq
4\exp\left(
    -
    \frac{ct^2}{
        \left(
            \|\mu\|_\infty/\eta
        \right)^{1/3}
        +
        \|\mu|q|^2\|_\infty
    }
\right).
\end{align*}
By the flatness hypothesis and the fact that $\eta\leq1$,
\[
\begin{aligned}
\left(
    \frac{\|\mu\|_\infty}{\eta}
\right)^{1/3}
+
\|\mu|q|^2\|_\infty
&\leq
\eta^{-1/3}\|\mu\|_\infty^{1/3}
+
\|\mu\|_\infty^{1/3}
\\
&\leq
2\eta^{-1/3}\|\mu\|_\infty^{1/3}.
\end{aligned}
\]
Therefore,
\[
\Pr_{f\sim F_\mu}\left[
    \left|
        A^{(i)}-\E A^{(i)}
    \right|>t
\right]
\leq
4\exp\left(
    -ct^2\eta^{1/3}\|\mu\|_\infty^{-1/3}
\right).
\]
Since $\eta=c_0t^8$, this gives
\[
\Pr_{f\sim F_\mu}\left[
    \left|
        A^{(i)}-\E A^{(i)}
    \right|>t
\right]
\leq
4\exp\left(
    -ct^{14/3}\|\mu\|_\infty^{-1/3}
\right).
\]
Finally, because $0<t\leq1$,
\[
    t^{14/3}
    \geq
    t^6,
\]
which proves the proposition.
\end{proof}

\subsection{Finishing up}

\begin{proof}[Proof of \Cref{thm:avg-hat-fixed-q}]
Fix $t>0$.
It suffices to consider $t\in(0,1]$, since
\[
    0
    \leq
    \Avg_{\wh f}[\wh{fq}]
    \leq
    1.
\]

With probability $1-\exp(-\Omega(n))$ over $\mu\sim\bs\mu$, simultaneously
for every $\mu$-compatible $q$,
\[
    \frac1n\sum_i
    \E_{f\sim F_\mu}A^{(i)}
    \leq
    \frac12+o(1)
\]
by \Cref{lem:fixed-q-expect}, and
\[
    \|\mu\|_\infty
    =
    N^{-1+o(1)}.
\]
Fix such a $\mu$ and a $\mu$-compatible $q$ satisfying
\[
    \|\mu|q|^2\|_\infty
    \leq
    \|\mu\|_\infty^{1/3}.
\]
For fixed $t>0$ and all sufficiently large $n$,
\[
    \left(
        \frac{t^8}{\|\mu\|_\infty}
    \right)^{1/3}
    \geq
    C\log N.
\]
Thus \Cref{prop:A-i-concentration} gives, for every $i\in[n]$,
\[
\Pr_{f\sim F_\mu}\left[
    A^{(i)}
    -
    \E_{f\sim F_\mu}A^{(i)}
    >
    t
\right]
\leq
4\exp\left(
    -ct^6N^{1/3-o(1)}
\right).
\]

Recalling that $
    \Avg_{\wh f}[\wh{fq}]
    =
    \frac1n\sum_iA^{(i)}$,
we obtain
\begin{align*}
&
\Pr_{f\sim F_\mu}\left[
    \Avg_{\wh f}[\wh{fq}]
    \geq
    \frac12+t+o(1)
\right]
\\
&\qquad\leq
\Pr_{f\sim F_\mu}\left[
    \frac1n\sum_i
    \left(
        A^{(i)}
        -
        \E_{f\sim F_\mu}A^{(i)}
    \right)
    \geq
    t
\right]
\\
&\qquad\leq
\sum_i
\Pr_{f\sim F_\mu}\left[
    A^{(i)}
    -
    \E_{f\sim F_\mu}A^{(i)}
    \geq
    t
\right]
\\
&\qquad\leq
4n\exp\left(
    -ct^6N^{1/3-o(1)}
\right)
\\
&\qquad\leq
\exp\left(
    -ct^6N^{1/3-o(1)}
\right),
\end{align*}
where the last line follows, after decreasing $c$, for all sufficiently large
$n$.
This proves the theorem.
\end{proof}

\section*{Interlude: Proof of the uncertainty principle for phase states}
\label{sec:phase-state-UP}
\addcontentsline{toc}{section}{\nameref{sec:phase-state-UP}}

As a warmup to the main theorem, we take a moment to prove the uncertainty principle in the special case of random phase states.
The contents of this section are not required for the proof of the main theorem and the reader may freely skip it.

\begin{theorem}
    \label{thm:UP-phase-state}
    For $n$ large enough, for $f$ a random phase function, for all $g$ compatible with $f$,
    \[\Avg_f[g]+\Avg_{\wh f}[\wh g]\leq 1.999,\]
    except with probability $2^{-\Omega(n)}$ over the choice of $f$.
\end{theorem}

\begin{proof}
Fix $\kappa=0.001$ and $d=\lfloor n/20\rfloor$.
Since $f$ is a phase function, writing $g=fq$ gives
\[
\E_{\unif}q=0,
\qquad
\|q\|_{L^2(\unif)}=1.
\]
By \eqref{eq:inf-appearance-phase}, if $\Inf[q]\geq\kappa n$, then
\[
\Avg_f[g]+\Avg_{\wh f}[\wh g]
\leq 1-\kappa+1
=1.999.
\]
It remains to handle the case $\Inf[q]\leq\kappa n$.

Let
\[
\mc V_d:=\mathrm{Span}\{\chi_S:1\leq |S|\leq d\},
\qquad
D:=\dim_{\C}\mc V_d
=\sum_{j=1}^{d}\binom nj
\leq N^{H_2(1/20)}.
\]
Choose a maximal $0.01$-separated subset $\mc N$ of the unit sphere of $\mc V_d$ in $L^2(\unif)$.
Then $\mc N$ is a $0.01$-net of this sphere, and the standard volumetric bound gives
\[
|\mc N|\leq 201^{2D}.
\]
Note this net is chosen independently of $f$.
Moreover, by Cauchy--Schwarz applied to its Fourier expansion,
\[
\|v\|_\infty^2\leq D\|v\|_{L^2(\unif)}^2=D.
\]
Thus, with $\mu$ the uniform probability distribution,
\[
\max_x\mu(x)|v(x)|^2
\leq \frac DN
\leq N^{H_2(1/20)-1}
\leq N^{-1/3}
=\|\mu\|_\infty^{1/3}.
\]

We next justify the fixed-function estimates at the deterministic measure $\mu=\unif$.
Here $\|\mu\|_\infty=1/N$ and $\wh\mu(s)=0$ for $s\neq0$, while Parseval gives
\[
\sum_{i=1}^n N|\wh{\mu v}(i)|^2
\leq N\sum_s|\wh{\mu v}(s)|^2
=1.
\]
Consequently, the Gaussian calculation in the proof of \Cref{lem:expect-A-Gauss}, together with \Cref{prop:smooth-replacement,lem:EB-bound}, gives, for every fixed $\eta\in(0,1]$,
\[
\E_{f\sim F_\mu}\Avg_{\wh f}[\wh{fv}]
\leq \frac12+\mc O(1/n)+C_\eta N^{-1/2}+C\eta^{1/8},
\]
uniformly over $\unif$-compatible $v$.
First letting $n\to\infty$ with $\eta$ fixed, and then letting $\eta\to0$, we obtain
\[
\E_{f\sim F_\mu}\Avg_{\wh f}[\wh{fv}]
\leq \frac12+o(1),
\]
where the $o(1)$ is uniform in $v$.

Take $t=0.01$.
For all sufficiently large $n$, the flatness bound above and $(t^8N)^{1/3}\geq C\log N$ allow us to apply \Cref{prop:A-i-concentration} to every $v\in\mc N$.
Since $\Avg_{\wh f}[\wh{fv}]=\frac1n\sum_iA^{(i)}$, a union bound over coordinates gives
\[
\Pr_f\left[
\Avg_{\wh f}[\wh{fv}]>\frac12+0.01+o(1)
\right]
\leq 4n\exp(-cN^{1/3}),
\]
where $c>0$ absorbs the fixed value of $t$.
A further union bound over $\mc N$ shows that
\[
\Avg_{\wh f}[\wh{fv}]
\leq \frac12+0.01+o(1)
\qquad\text{for every }v\in\mc N,
\]
except with probability at most
\[
4n\,201^{2D}\exp(-cN^{1/3})
\leq \exp(-c'N^{1/3})
\]
for all sufficiently large $n$.
Here we used $H_2(1/20)\approx0.286397<1/3$.

Condition on this event, and consider any $\unif$-compatible $q$ with $\Inf[q]\leq\kappa n$.
Let $p$ be its orthogonal projection onto $\mc V_d$.
Since $q$ has mean zero, the Fourier spectral decomposition gives
\[
\|q-p\|_{L^2(\unif)}^2
\leq \frac{\Inf[q]}{d+1}
\leq \frac{\kappa n}{d+1}
\leq 0.02.
\]
In particular, $\|p\|_{L^2(\unif)}^2\geq0.98$, so we may define $u=p/\|p\|_{L^2(\unif)}$.
Orthogonality of the projection gives
\[
\begin{aligned}
\|q-u\|_{L^2(\unif)}^2
&=2\bigl(1-\|p\|_{L^2(\unif)}\bigr)\\
&\leq 2\bigl(1-\|p\|_{L^2(\unif)}^2\bigr)\\
&=2\|q-p\|_{L^2(\unif)}^2
\leq0.04.
\end{aligned}
\]
Choose $v\in\mc N$ with $\|u-v\|_{L^2(\unif)}\leq0.01$.
Then
\[
\|q-v\|_{L^2(\unif)}\leq0.21.
\]

Multiplication by $f$ is an isometry from $L^2(\unif)$ to $\ell^2$, and the Fourier transform is unitary.
Since $M_{\wh f}$ is a positive contraction, the quadratic form $q\mapsto\Avg_{\wh f}[\wh{fq}]$ is therefore $2$-Lipschitz on the unit sphere of $L^2(\unif)$.
It follows that
\[
\begin{aligned}
\Avg_{\wh f}[\wh{fq}]
&\leq \Avg_{\wh f}[\wh{fv}]+2\|q-v\|_{L^2(\unif)}\\
&\leq \frac12+0.01+2(0.21)+o(1)\\
&=0.93+o(1).
\end{aligned}
\]
Combining the two cases, for all sufficiently large $n$ we have
\[
\Avg_f[g]+\Avg_{\wh f}[\wh g]
\leq \max\{2-\kappa,1.93+o(1)\}
\leq1.999.
\]
This holds simultaneously for every $g$ compatible with $f$, and the failure probability is at most $\exp(-c'N^{1/3})\leq2^{-\Omega(n)}$.
\end{proof}

\section{The $A$-$B$ partition and comparing Dirichlet forms on $A$}
\label{sec:AB-split}

We now return to the proof of \Cref{thm:full-uncertainty}. As outlined in the technical overview, the next step is to split $\{0,1\}^n$ into two sets, a large one $A$ such that the $\eps$-net argument works on functions supported on $A$, and a small one $B$ where the $\eps$-net argument is false but we can still bound $\Avg$ for functions supported here by other means.
In particular, for a universal constant $T>1$ to be chosen later, set
\[A=\left\{x:\mu(x)\in\left[\frac{1}{TN}, \frac{T}{N}\right]\right\}\qquad\text{and}\qquad B=A^c\,.\]
Define $q_A=q\cdot\1_A$ and $q_B=q\cdot\1_B$.

The idea is that $\Inf_\mu$ should be comparable (up to $T$) with $\Inf_{\unif}$ on $A$, so the covering number facts that hold for low-influence functions on the uniform-measure hypercube should hold also for $q_A$.
Unfortunately, it is not always true that $q_A$ is low influence, even though $q$ is promised to be.
Recall that in the win-win argument, we are only given $\Inf_\mu(q)\leq \kappa n$.
The only immediate fact we can conclude about $q_A$ from the given bound is that 
\[\Inf_{\unif}^{\Int (A)}(q_A)\simeq_T\Inf_\mu^{\Int(A)}(q_A)=\Inf_\mu^{\Int(A)}(q)\overset{\text{(inclusion)}}\leq \Inf_\mu(q)\leq \kappa\,,\]
where for any function $h$,
\[\Inf_\mu^{\Int(A)}(h):=\sum_{\substack{x,y\in A\\y\sim x}}\frac{2\mu(x)\mu(y)}{\mu(x)+\mu(y)}\left|\frac{h(x)-h(y)}{2}\right|^2\,.\]
Unfortunately, without using more about $A$, having bounded internal influence $\Inf_\mu^{\Int(A)}$ is far from enough to get a covering number bound.
A simple counterexample is $A$ that is an independent set on the hypercube graph (no interior); a more general one is $A$ any collection of very many small, disjoint subgraphs $A_j$, $j=1,\ldots, J$.
As long as $q$ is constant on each of the $A_j$'s, $\Inf^{\Int(A)}_\unif(q_A)$ will always be zero, even though such functions consitute a $J$-dimensional subspace.

But these examples merely show we must use more about $A$.
The main purpose of this section is to show that with high probability over $\mu$, the geometry of $A$ is such that a bound $\Inf_\unif^{\Int(A)}(q_A)$ indeed suffices to control the metric entropy of the $q_A$'s.
\begin{theorem}[Metric entropy of $A$-part of low-energy $q$]
\label{thm:metric-entropy}
    Let
    \[Q_\kappa=\Big\{q_A:\E_\mu q= 0 , \|q\|_{L^2(\mu)}=1,\Inf_\mu^{\Int(A)}(q)\leq \kappa n\Big\}\]
    Then with high probability over $\mu\sim \bs \mu$, for all $\eps\in (0,1]$,
    \[\log\mc N_\eps(Q_\kappa, L^2(\mu))\leq N^{H(\gamma)+o(1)}\log(C/\eps) \qquad\text{for}\qquad\gamma:=C\frac{T^5\kappa}{\eps^2}\,.\]
\end{theorem}

Here and throughout, $\mc N_\eps(S,L^2(\mu))$ denotes the $\eps$-covering number of a set $S$ with respect to the metric induced by $L^2(\mu)$.

The main idea in the proof is to construct a \textit{linear extension operator} $\Ext_A$ so that $(\Ext_Aq_A)_A=q_A$ and
\[\Inf_{\unif}(\Ext_A q_A)\;\lesssim_T\;\Inf_{\unif}^{\Int(A)}(q_A)\overset{\text{(by the above)}}\leq \kappa n\,.\]
We may then apply standard dimension-counting argument for low-influence function to the $\Ext_Aq_A$'s.

\bigskip
To begin, for any $b\in B$ write $\mc N_A(b)=\{x\in A:x\sim b\}$.
Then define $\Ext_A$ via
\[
(\mathrm{Ext}_A h)(x)=
\begin{cases}
h(x), & x\in A,\\[1mm]
\displaystyle \frac{1}{|\mc N_A(b)|}\sum_{u\in \mc N_A(b)} h(u), & x=b\in B.
\end{cases}
\]
for any $h:\{0,1\}^n\to \C$.
Clearly, $(\Ext_Ah)_A=h_A)$.
Our analysis of $\Ext_A$ draws inspiration from the literature on canonical paths and the ``local escape property'' of \cite{Huang2025}.
We will use the canonical paths employed in \cite{Huang2025}; let us recall these and then assemble some facts about them.

\begin{definition}[Canonical $i$-paths, adapted from \cite{Huang2025}]
Let $(u,v)\in \{0,1\}^n$.
For $i\in [n]$ define the $i$\textsuperscript{th} canonical path, $P_i(u,v)$ to be as follows.
\begin{itemize}
    \item Let $D=\{j:u_j\neq v_j\}$.
    \item Begin at $u$ and flip coordinate $i$
    \item Scan rightward (cyclically) and flip each bit in $D$.
    \item If $i\notin D$, flip $i$ again at the end.
\end{itemize}
\end{definition}

\begin{fact}[{Adapted from \cite[Prop. 11]{Huang2025}}]
\label{prop:canonical-paths}
    For any two vertices $x,y\in\{0,1\}^n$, the paths $P_i(x,y)$ are simple, pairwise internally disjoint, and of length $|P_i(x,y)|\leq d(x,y)+2$.
\end{fact}

We need a version of the local escape property of \cite{Huang2025} for general $A$ (not just the thresholds $[\frac1{11N},\frac{5}{N}]$).
\begin{prop}[{Local escape property, generalized from \cite[Lemma 17]{Huang2025}}]
    For any $T\geq 12$
    there are $\alpha=\alpha(T),\beta=\beta(T) > 0$ such that for all $n$ large enough, with probability at least $1-2^{-\beta n}$,
    \begin{itemize}
        \item each $x\in \{0,1\}^n$ has at least $\alpha n$ neighbors in $A$;
        \item for all $x,y$ with $d(x,y)\leq 3$, at least $\alpha n$ of the  $i$-paths $P_i(x,y)$ have all internal vertices in $A$.
    \end{itemize}
\end{prop}
\begin{proof}
	Let us use a normalized exponential model for $\mu$.
	That is, define $Z_x$, $x\in \{0,1\}^n$ to be iid $\mathrm{Exp}(1)$ and let $\mu(x)=Z_x/(\sum_yZ_y)$.
	Let
	\[E=\{N/2\leq \tsum_xZ_x\leq 2N\}\]
	and as usual, $\Pr[E]\geq 1-\exp(-\Omega(N))$.
	Let us define a ``surrogate'' $A$ via
	\[\wt A=\{x: 2/T\leq Z_x\leq T/2\}\]
	and note that if $E$ obtains, $\wt A\subseteq A$, because
	\[\frac{1}{TN}\leq \mu(x)=\frac{Z_x}{\sum_y Z_y}\leq \frac TN\,.\]
	Also note for any fixed $x$,
	\[p_{\wt A}:=\Pr\!\big[x\in \wt A\,\big]=e^{-2/T}-e^{-T/2},\]
	so for all $T\geq 12$,
	\[p_{\wt A}^4>\frac12.\]
	\medskip
	
	\noindent\textit{Proof of (i).}
	Denote the number of $\wt A$-neighbors of $x$ by $G_x=|\{j\in[n]:x+e_j\in \wt A\}|$.
	By independence, $G_x\sim \mathrm{Bin}(n,p_{\wt A})$ and so for the appropriate choice of $\alpha_1$,
	\[\Pr[G_x< \alpha_1 n]\leq e^{-c_1 n}\]
	for $c_1>\log 2$.
	(This is possible because $p_{\wt A}> 1/2$.)
	Union-bounding over $x$ gives 
	\[\Pr[\exists x: G_x< \alpha_1 n]\leq 1-2^ne^{-c_1n}\leq2^{\beta_1n}\]
	for an appropriate choice of $\beta_1$ and all $n$ large enough.
	
	\medskip
	\noindent\textit{Proof of (ii)}.
	Fix $x,y\in\{0,1\}^n$ with $d(x,y)\leq 3$.
	By \Cref{prop:canonical-paths} the paths $P_i(x,y)$, $i\in [n]$ are internally disjoint and have length at most $d(x,y)+2\leq 5$.
	So each path has at most 4 internal vertices, disjoint from all the others.
	Let $J_i(x,y)$ indicate the event that all internal vertices of $P_i(x,y)$ lie in $\wt A$.
	By disjointness, the $J_i$'s are independent and
	\[\Pr[J_i(x,y)=1]\geq q_{\wt A}^4>\frac12\,.\]
	Thus, with $K(x,y):=\sum_iJ_i(x,y)$, we may choose an $\alpha_2$ so that
	\[\Pr[K(x,y)<\alpha_2n]\leq e^{-c_2n}\]
	where again $c_2>\log 2$.
	
	For the final union bound, note that the number of pairs $x,y$ of distance at most 3 is upper-bounded by
	\[2^n\tsum_{k=0}^3\binom{n}{k}\leq 2^{n+o(n)}\,.\]
	Thus
	\[\Pr[\exists x,y: d(x,y)\leq 3, K(x,y)< \alpha_2n]\leq 2^{n+o(n)}2^{-c_2n}\leq 2^{-\beta_2 n}\]
	for an appropriate choice of $\beta_2$ and $n$ large enough.
	
	On $E$ we have $\wt A\subseteq A$, so all paths counted by the $K(x,y)$'s are entirely within $A$.
	Combining these union bounds with $\Pr[E]\leq e^{-\Omega(N)}$ and taking $\alpha=\min\{\alpha_1,\alpha_2\}, \beta=\min\{\beta_1,\beta_2\}$ finishes the proposition.
\end{proof}

With our extensions of \cite{Huang2025} established, we are almost prepared to study $\Ext_A$.
We just need a small combinatorial lemma about the number of canonical paths using a fixed edge:
\begin{lem}
    \label{lem:path-count-lem}
    For any fixed edge $e$ in the hypercube graph and any distance $\ell\geq 1$,
    \[|\{(u,v,i):d(u,v)=\ell, e\in P_i(u,v)\}|\leq 6(2n)^\ell.\]
\end{lem}
\begin{proof}
    Let $D=\{j:u_j\neq v_j\}$.
    For $e$ to lie on the path between $u$ and $v$, it must be along a coordinate direction $i_e \in \{i\}\cup D$.
    So the already the number of admissible pairs $(i,D)$ for $e$ is at most $3n^\ell$:
    \begin{itemize}
        \item If $i\notin D$, either $i_e=i$ and there are $\binom{n-1}{\ell}\leq n^\ell$ consistent pairs $(i,D)$, or $i_e\in D$, in which case there are $n\binom{n-2}{\ell-1}\leq n^\ell$ consistent pairs.
        \item If $i\in D$, then there are $\binom{n-1}{\ell-1}\ell\leq n^\ell$ possibilities.
    \end{itemize}
    Now observe that the $(i,D)$'s can arise only from a very limited number of paths containing $e$; in particular because $i_e\in\{i\}\cup D$ and the walk never alters bits in $F:=(\{i\}\cup D)^c$, $e$ fixes the bits in $F$ for both $u$ and $v$.
    There are thus at most $2^{r+1}$ starting vertices $u$ compatible with $(e,i,D)$, and the information $(e,i,D,u)$ also fixes the ending vertex $v$.
    So the final count is $6 (2n)^\ell$.
\end{proof}

\begin{prop}[Extension under local escape]
\label{prop:extension-local-escape}
Let $A\subseteq\{0,1\}^n$ and $B:=A^c$ so that local escape holds.
Then
\[
\Inf_{\mathrm{unif}}(\mathrm{Ext}_A h)
\le C_\alpha \, \Inf_{\mathrm{unif}}^{\mathrm{int}(A)}(h)
\qquad \forall h:A\to\C.
\]
Moreover, we may take $C_\alpha \leq 1100/\alpha^3$ provided $\alpha < 1$.
\end{prop}

\begin{proof}
For notational convenience define the unweighted Dirichlet energies
\[
\mathscr{E}(h):=\sum_{x\sim y} |h(x)-h(y)|^2,
\qquad
\mathscr{E}_A(h):=\sum_{\substack{x\sim y\\ x,y\in A}} |h(x)-h(y)|^2.
\]
It suffices to prove
\[
\mathscr{E}(\mathrm{Ext}_A h)\le C_\alpha \mathscr{E}_A(h),
\]
since $\Inf_{\mathrm{unif}}=(4N)^{-1}\mathscr{E}$.

Let us decompose the hypercube edges in the natural way, into the $AA$, $AB$, and $BB$ edges.
Label their contributions to $\mathscr{E}(\Ext_Ah)$ by $I_{AA},I_{AB},$ and $I_{BB}$ respectively.

\medskip
\noindent\textit{$AA$ edges.}
The contribution is exactly $\mathscr{E}_A(h)$:
\[I_{AA}:=\mathscr{E}_A(h).\]

\medskip
\noindent\textit{$AB$ edges.}
Consider some $b\in B$.
The total contribution of $b$ to $AB$ edges is
\[
\sum_{a\in \mc N_A(b)} |\Ext_A h(b)-h(a)|^2
=
\frac{1}{|\mc N_A(b)|}\sum_{\{a,a'\}\subset \mc N_A(b)}|h(a)-h(a')|^2,
\]
(one may recognize the identity for variance).

Now, every pair $u,v\in \mc N_A(b)$ has $d(u,v)=2$.
For any path $P=(x_0=u,x_1,\ldots, x_{\ell-1},x_\ell=v)$ in the hypercube, we may of course write
\[h(u)-h(v)=\sum_{j=1}^\ell\big(h(x_{j-1})-h(x_j)\big),\]
so
\[|h(u)-h(v)|^2\leq \ell\sum_{j=1}^\ell |h(x_{j-1})-h(x_j)|^2\,.\]
Averaging over the $\geq \alpha n$ good canonical paths $\mc P^A_{a,a'}$ of length at most $5$ between $u$ and $v$ yields
\[
|g(u)-g(v)|^2
\le
\frac{5}{\alpha n}
\sum_{P\in\mc{P}_{a,a'}}
\sum_{e\in P} |h(e_\text{start})-h(e_\text{end})|^2,
\]
so in total the $AB$ contribution to $\mathscr{E}(\Ext_Ag)$ is
\[I_{AB}=\sum_{b\in B}\sum_{a\in \mc N_A(b)}|\Ext_Ah(b)-h(a)|^2\leq\frac{5}{\alpha^2n^2}\sum_{b\in B}\sum_{\{a,a'\}\subset \mc N_A(b)}\sum_{P\in \mc P^A_{a,a'}}\sum_{e\in P}|h(e_\mathrm{start})-h(e_\mathrm{end})|^2\,,\]
where we have used $|\mc N_A(b)|\ge \alpha n$.

Note that any $\{a,a'\}\in A$, $d(a,a')=2$, has at most two $b\in B$ such that $\{a,a'\}\subseteq \mc N_A(b)$.
Thus
\[I_{AB}\leq \frac{10}{\alpha^2n^2}\sum_{\substack{a,a'\in A\\d(a,a')=2}}\sum_{P\in\mc P^A_{a,a'}}\sum_{e\in P}|h(e_\text{start})-h(e_\text{end})|^2\,.\]
Reordering the sum and overcounting (forgetting that the paths must be good) and then applying the counting lemma, we have
\begin{align*}
I_{AB} &\leq \frac{10}{\alpha^2n^2}\sum_{e\subset A}\left|\{(a,a',i):d(a,a')= 2,e\in P_i(a,a')\}\right|\cdot|h(e_\text{start})-h(e_\text{end})|^2\\
&\leq \frac{120n^2}{\alpha^2n^2}\sum_{x\sim y\in A}|h(x)-h(y)|^2\tag{\Cref{lem:path-count-lem}}\\
&= \frac{120}{\alpha^2}\mathscr{E}_A(h)\,.
\end{align*}

\medskip
\noindent\textit{$BB$ edges.}
If $b\sim b'$
\begin{align*}
|\Ext_A h(b)-\Ext_A h(b')|^2&= \left|\frac{1}{|\mc N_A(b)|\cdot |\mc N_A(b')|}\sum_{u\in \mc N_A(b)}\sum_{v\in \mc N_A(b')}\left(h(u)-h(v)\right)\right|^2\\
&\leq \frac{1}{\alpha^2n^2}\sum_{u\in \mc N_A(b)}\sum_{v\in \mc N_A(b')}\left|h(u)-h(v)\right|^2.
\end{align*}
Now $d(u,v)\in\{1,3\}$.

For the $d(u,v)=1$ pairs, we get a contribution of at most
\begin{align*}
    \frac{1}{\alpha^2n^2}\sum_{b\sim b'\in B}\sum_{\substack{
    u\in \mc N_A(b)\\
    v\in \mc N_A(b')\\
    u\sim v}}|h(u)-h(v)|^2=\frac{1}{\alpha^2n^2}\sum_{\substack{
    u\in \mc N_A(b)\\
    v\in \mc N_A(b')\\
    u\sim v}}\sum_{b\sim b'\in B}|h(u)-h(v)|^2\leq \frac{1}{\alpha^2n}\mathscr{E}_A(h)\,.
\end{align*}

For the $d(u,v)=3$ part, we follow the idea from the $I_{AB}$ edges: average over the $\geq \alpha n$ canonical good paths of length at most $5$ connecting $u$ and $v$:
\[
|h(u)-h(v)|^2
\le
\frac{5}{\alpha n}
\sum_{P\in\mc{P}^A_{u,v}}
\sum_{e\in P} |h(e_\text{start}-h(e_\text{end})|^2.
\]
Note that any $\{a,a'\}\in A, d(a,a')=3$, has at most four $BB$ edges $b\sim b'\in B$ with $a\in \mc N_A(b)$ and $a' \in \mc N_A(b')$.
Using this observation and applying \Cref{lem:path-count-lem}, we get a total contribution to $I_{BB}$ by the $d(u,b)$ part of at most
\[\frac{20}{\alpha^3n^3}\cdot48n^3\sum_{x\sim y\in A}|h(x)-h(y)|^2=\frac{960}{\alpha^3}\mathscr{E}_A.
\]

Combining the $AA$, $AB$, and $BB$ contributions gives
\[
\mathscr{E}(\Ext_A h)\leq \left(\frac{960}{\alpha^3}+\frac{120}{\alpha^2}+1+o(1)\right)\mathscr{E}_A(h)\,. \qedhere
\]
\end{proof}
\medskip

We are finally able to give the metric entropy bound.
\begin{proof}[Proof of \Cref{thm:metric-entropy}]
Consider $q_A$ for $q\in Q_\kappa$.
By \Cref{prop:extension-local-escape},
\[
\mc E_{\mathrm{unif}}(\Ext_A q_A)
\le
C_\alpha \mc E_{\mathrm{unif}}^{\mathrm{int}(A)}(q_A).
\]
Also, on every edge inside $A$,
\[
\frac{\mu(x)\mu(y)}{\mu(x)+\mu(y)}\ge \frac{1}{2TN},
\]
hence
\[
\mc E_{\mathrm{unif}}^{\mathrm{int}(A)}(q_A)
\le
T\mc E_\mu^{\mathrm{int}(A)}(q)
\le
T\kappa.
\]
Therefore
\[
\mc E_{\mathrm{unif}}(\Ext_A q_A)\le C_\alpha T\kappa.
\]

Let $(\Ext_A q_A)_{\le d}$ be the projection of $\Ext_Aq_A$ onto Fourier degree at most $d$.
Since the normalized cube Laplacian has eigenvalue $|S|/n$
on Fourier level $S$,
\[
\|\Ext_A q_A-(\Ext_A q_A)_{\le d}\|_{L^2(\mathrm{unif})}^2
\le
\frac{n}{d}\mc E_{\mathrm{unif}}(\Ext_A q_A)
\le
C_\alpha T\frac{n}{d}\kappa.
\]
Restricting back to $A$ and using $\mu(x)\le T/N$ on $A$,
\[
\|q_A-((\Ext_A q_A)_{\le d})_A\|_{L^2(\mu)}^2
\le
T\|\Ext_A q_A-(\Ext_A q_A)_{\le d}\|_{L^2(\unif)}^2
\le
C_{\alpha}T^2\frac{n}{d}\kappa.
\]
Choose
\[
d= C_{\alpha}T^2 n\, \frac{\kappa}{\eps^2}.
\]
Then every $q_A$ for $q\in Q_\kappa$ is within $\eps$ in $L^2(\mu)$ of the restriction
to $A$ of a degree-$\le d$ Fourier polynomial.

The dimension of the degree-$\le d$ Fourier space is
\[
D=\sum_{j\le d}\binom{n}{j}\le N^{H(d/n)+o(1)},
\]
where $H(\cdot)$ is binary entropy.
A standard volumetric bound gives an $\eps$-net of cardinality at most $(C/\eps)^D$,
hence
\[
\log \mc N_\eps( Q_\kappa,L^2(\mu))
\le
d\log(C/\eps)
\le
N^{H(\gamma)+o(1)}\log(C/\eps),
\]
with $\gamma=d/n= C_\alpha T^2\kappa/\eps^2\leq 1100T^2\kappa/(\eps^2\alpha^3)$.
\end{proof}

A corollary to this general metric entropy argument is a slightly more technical statement, that guarantees a net comprising polynomials with a certain flatness guarantee that will be needed in the sequel.
\begin{cor}[Compatible flat net for the centered $A$-supported class]
\label{cor:flat-A-net}
Fix $T\geq12$, $\kappa>0$, and $0<\eps\leq1$.
For a sufficiently large universal constant $C>0$, the following holds with probability $1-2^{-\Omega(n)}$ over $\mu\sim\bs\mu$, for all sufficiently large $n$.

Let
\[
    \mathcal Q_{\kappa,A}
    :=
    \left\{
        q:
        \operatorname{supp}(q)\subseteq A,\ 
        \E_\mu q=0,\ 
        \|q\|_{L^2(\mu)}=1,\ 
        \Inf_\mu^{\Int(A)}(q)\leq\kappa n
    \right\}.
\]
Suppose that
\[
    0<\gamma\leq 1/3 \qquad \text{and} \qquad H(\gamma)<\frac13
    \qquad\text{for}\qquad
    \gamma
    :=
    C\frac{T^5\kappa}{\eps^2}.
\]
Then there exists a finite family
\[
    \mathcal Q_{\kappa,A}^{(\eps)}
\]
such that every $q\in\mathcal Q_{\kappa,A}$ has some
$v\in\mathcal Q_{\kappa,A}^{(\eps)}$ satisfying
\[
    \|q-v\|_{L^2(\mu)}
    \leq
    \eps.
\]
Every $v\in\mathcal Q_{\kappa,A}^{(\eps)}$ is $\mu$-compatible and satisfies
\[
    \max_x\mu(x)|v(x)|^2
    \leq
    \|\mu\|_\infty^{1/3}.
\]
Moreover,
\[
    \log
    \left|
        \mathcal Q_{\kappa,A}^{(\eps)}
    \right|
    \leq
    N^{H(\gamma)+o(1)}
    \log\left(
        \frac C\eps
    \right).
\]
\end{cor}
\begin{proof}
Work on the local escape event.
By monotonicity of $A(T)$ in $T$, the event for $T=12$ supplies an absolute constant $\alpha>0$ valid for every $T\geq12$.

Fix $q\in\mathcal Q_{\kappa,A}$ and put $h=\Ext_Aq$.
Jensen's inequality and $|\mathcal N_A(b)|\geq\alpha n$ give
\[
\|h\|_{L^2(\unif)}^2
\leq (1+\alpha^{-1})\|q\|_{L^2(\unif)}^2
\leq C_0T.
\]
Indeed, when summing the extension values over $b\in B$, each $a\in A$ occurs in at most $n$ neighbor sets.
Moreover, \Cref{prop:extension-local-escape} and comparison of the edge weights on $A$ give
\[
\Inf_{\unif}[h]
\leq C_0T\Inf_\mu^{\Int(A)}[q]
\leq C_0T\kappa n,
\]
where $C_0$ is absolute.

Set $d=\lfloor\gamma n\rfloor$, and let $p=h_{\leq d}$ be the orthogonal projection onto Fourier degree at most $d$.
The Fourier spectral decomposition gives
\[
\begin{aligned}
\|q-\1_Ap\|_{L^2(\mu)}^2
&\leq T\|h-p\|_{L^2(\unif)}^2\\
&\leq \frac{T}{d+1}\Inf_{\unif}[h]\\
&\leq \frac{C_0T^2\kappa n}{d+1}
\leq \frac{\eps^2}{64},
\end{aligned}
\]
provided the universal constant $C$ in $\gamma$ is sufficiently large.

To enforce compatibility, define
\[
r=\1_A\left(p-\frac{\E_\mu[\1_Ap]}{\mu(A)}\right),
\qquad
u_q=\frac{r}{\|r\|_{L^2(\mu)}}.
\]
Centering on $A$ is an orthogonal projection in $L^2(\mu)$ and fixes $q$.
Hence
\[
\|q-r\|_{L^2(\mu)}\leq\eps/8,
\qquad
\|r\|_{L^2(\mu)}\geq1-\eps/8\geq7/8,
\]
and therefore
\[
\E_\mu u_q=0,\qquad
\|u_q\|_{L^2(\mu)}=1,\qquad
\|q-u_q\|_{L^2(\mu)}\leq\eps/4.
\]

Let
\[
D=\sum_{j=0}^{d}\binom nj\leq N^{H(\gamma)}.
\]
Since the Walsh characters are uniformly bounded and orthonormal in $L^2(\unif)$,
\[
\|p\|_\infty^2
\leq D\|p\|_{L^2(\unif)}^2
\leq C_0TD.
\]
The centering constant has modulus at most $\|p\|_\infty$.
Using the lower bound on $\|r\|_{L^2(\mu)}$ and the support of $u_q$ in $A$, we obtain
\[
\max_x\mu(x)|u_q(x)|^2
\leq C_1T^2\frac{D}{N}
\leq C_1T^2N^{-2/3}
\leq N^{-1/3}
\leq\|\mu\|_\infty^{1/3}
\]
for all sufficiently large $n$.

Finally, the family
\[
\mathcal V=\{u_q:q\in\mathcal Q_{\kappa,A}\}
\]
lies on the unit sphere of a complex linear subspace of dimension at most $D$: the image of the degree-$\leq d$ Fourier space under restriction to $A$ and centering.
Choose a maximal $\eps/2$-separated subset of $\mathcal V$ as $\mathcal Q_{\kappa,A}^{(\eps)}$.
The volumetric packing bound gives
\[
\left|\mathcal Q_{\kappa,A}^{(\eps)}\right|
\leq (1+4/\eps)^{2D}.
\]
Every $q\in\mathcal Q_{\kappa,A}$ is within $\eps/4$ of $u_q$ and hence within $3\eps/4$ of this net.
Its points are compatible and flat because they belong to $\mathcal V$, and
\[
\log\left|\mathcal Q_{\kappa,A}^{(\eps)}\right|
\leq 2D\log(1+4/\eps)
\leq N^{H(\gamma)+o(1)}\log(C/\eps).
\]
The construction depends only on $\mu$.
\end{proof}

\section{Concentration of $\|\wh{M_{\wh f}}\|_\op$ for the small subspace $B$}
\label{sec:op-norm-conc}

Now that we have been able to analyze the contribution from the ``good'' $A$-part, we will focus on the ``bad'' $B$-part. It will be convenient to work in the $g$-picture rather than the $q$-picture.
Recall that $\wh M_{\wh f}$ is defined via
\[\langle g, \wh{M_{\wh f}} \,g\rangle = \Avg_{\wh f}[\wh g].\]

This section is devoted to a proof of the following concentration bounds.
\begin{theorem}
    \label{thm:B-op-norm-bounds}
    For any $\eps>0$, there is a $T\geq 1$ such that with probability $1-\exp(-\Omega_\eps(n))$ over Haar-random $f$,
    \begin{align}
    \label{ineq:BB}
    &\|\Pi_B\wh{M_{\wh f}}\Pi_B\|_\op\leq \frac12+\eps\\
    \label{ineq:AB}
    \text{and}\qquad &\|\Pi_A\wh{M_{\wh f}}\Pi_B\|_\op=\|\Pi_B\wh{M_{\wh f}}\Pi_a\|_\op\leq \frac14+\eps\,,
    \end{align}
    where $\Pi_B$ is the coordinate projection onto $B=B(T)=\{x:\mu_f(x)\notin [1/{TN},T/N]\}$ and $\Pi_A$ is onto $A=B^c$.
\end{theorem}


The proof of \Cref{thm:B-op-norm-bounds} proceeds roughly as follows.
Let us start with the bound for $\Pi_B \wh{M_{\wh f}}\Pi_B$ and begin by fixing a $\mu\sim \bs \mu$.
With high probability, $\mu$ is well-behaved in a certain sense required by comparisons below.
Now, notice that because we only need the good event in \Cref{thm:B-op-norm-bounds} to occur with probability $1-\exp(-\Omega(n))$, it suffices to consider tracial moments of $\Pi_B \wh{M_{\wh f}}\Pi_B$ only up to order $\mc O(n)$, rather than $N=2^n$ (here the tracial moments are integrated over $f\sim F_\mu$ for the fixed $\mu$).
Next, we observe that from the perspective of moments of this relatively low order, when $\mu$ is ``nice,'' the Fourier coefficients $\wh f(s)$ are close to i.i.d., and in particular they are well approximated by a product of i.i.d. Gaussians which are independent of $\mu$ entirely.
In this i.i.d. Gaussian model we can then uncondition on $\mu$, which leads to $B$ varying as a random subset independent of the Gaussianized $\wh f(s)$'s.
This means that in $\wh{M_{\wh f}}$, the projector $\Pi_B$ becomes a \textit{uniformly random coordinate projection} of a certain size.
Thus, fixing our Gaussianized $\wh{M_{\wh f}}$ and allowing $B$ to vary, $\Pi_B \wh{M_{\wh f}}\Pi_B$ is (after centering) precisely a \textit{random paving}, for which we appeal to the work of Tropp \cite{MR2379999}.
Provided certain technical conditions hold, \cite{MR2379999} states that sufficiently sparse random coordinate projections of an entrywise small matrix have small operator norm.

The proof of the $BB$ bound, \Cref{ineq:BB}, is the subject of \Cref{sec:BB} and in particular \Cref{prop:BB-sufficient}.
The argument for the bound on $\|\Pi_A \wh{M_{\wh f}}\Pi_B\|_\op$, given in \Cref{sec:AB}, is nearly identical except that we pass to the square of $\Pi_A \wh{M_{\wh f}}\Pi_B$, controlled by (a centered version of) $\Pi_B \wh{M_{\wh f}}^2 \Pi_B$.
This object is then treated with the same argument as above.

Several technical steps are required to execute these proofs, but we can borrow much of the machinery developed in \Cref{sec:fixed-q-conc}, especially the regularization approach and corresponding derivative bounds on terms in $\wh{M_{\wh f}}$.

\subsection{Technical preparations}

\subsubsection{Regularity of $B$}

\begin{prop}
\label{prop:bad-set-regularity}
Set
\[
    \beta_T
    :=
    \Pr\bigl[\mathrm{Exp}(1)\notin[1/T,T]\bigr]
    =
    1-e^{-1/T}+e^{-T}=\E|B|/N.
\]

Then, with probability at least \(1-2^{-\Omega(n)}\),
\[
    \frac{|B|}{N}=\beta_T+o(1),
\]
Moreover, the law of \(B\) is exchangeable under permutations of the cube.
Consequently, conditional on \(|B|=k\), the set \(B\) is uniformly distributed
among all \(k\)-subsets of \(\{0,1\}^n\).
In particular, for every \(\beta>0\), if \(T\) is chosen sufficiently large,
then with probability at least \(1-2^{-\Omega(n)}\),
\[
    |B|\le \beta N.
\]
\end{prop}

\subsubsection{Random paving}

We will need the following random paving estimate of Tropp~\cite{MR2379999}.
Let us call by a \emph{uniform coordinate projection of density $\eta$} the random diagonal matrix $R$ where the diagonal is distributed uniformly over $0$-$1$ vectors with 1-density $\eta$.
For any matrix $A$, the projection $RAR$ is a \textit{random paving}.

\begin{theorem}[Moments of random pavings {\cite[Prop. 4 \& Thm. 5]{MR2379999}}]
\label{thm:tropp-random-paving}
Fix $\alpha>0$ and  $\eps\in(0,1)$.
There exists
$\beta_{\max}=\beta_{\max}(\eps,\alpha)>0$ such that the following holds for any sufficiently-large integer $N$.

If $A$ is a deterministic $N\times N$ matrix satisfying
\[
    \max_{x,y}|A_{xy}|\le \frac{1}{(\log N)^{1+\alpha}}\|A\|_\op,
\]
and $R$ is a uniform coordinate projection of density at most $\beta_{\max}$,
then for all even integers $p$ with $2\le p\le 2\lceil \log N\rceil$,
\[
    \left(
        \E_R\|RAR\|_{\op}^{p}
    \right)^{1/p}
    \le \eps\|A\|_{\op} .
\]
\end{theorem}

\subsection{Proof of  \Cref{ineq:BB}, the $BB$ bound}
\label{sec:BB}

\paragraph{Step 1: Regularize.}
Recall that
\[\wh{M_{\wh f}} = H^{\otimes n}\cdot \frac{1}{2n}\sum_{i,s}\frac{\ketbra{u_{s,i}}{u_{s,i}}}{\|u_{s,i}\|_2^2}\cdot H^{\otimes n},\]
where $\ket{u_{s,i}}=\wh f(s)\ket{s} + \wh f(s+i)\ket{s+i}$.
Similarly to what was done in \Cref{sec:concentration}, define for $\eta>0$ the regularized operator
\[\wh M_{\wh f}^{(\eta)} = H^{\otimes n}\cdot \frac{1}{2n}\sum_{i,s}\frac{\ketbra{u_{s,i}}{u_{s,i}}}{\|u_{s,i}\|_2^2+\eta/N}\cdot H^{\otimes n}\,.\]
One may calculate that
\begin{align*}
    0\preceq \wh M_{\wh f} - \wh M_{\wh f}^{(\eta)}\preceq H^{\otimes n}\cdot \frac{1}{2n}\sum_{i,s}\frac{(\eta/N) I_{s,s+i}}{\|u_{s,i}\|_2^2+\eta/N}\cdot H^{\otimes n}=: D^{(\eta)}_{\wh f}
\end{align*}
where $I_{s,s+i}$ denotes the all-zeros matrix except with $I$ on the $(s,s+i)$ submatrix.

Thus
\[\|\Pi_B\wh{M_{\wh f}}\Pi_B\|_\op \leq \|\Pi_B\wh{M}^{(\eta)}_{\wh f}\Pi_B\|_\op+\|\Pi_BD_{\wh f}^{(\eta)}\Pi_B\|_\op\]
and \eqref{ineq:BB} follows if we show:

\begin{prop}
\label{prop:BB-sufficient}
For any $\eps>0$, there is a threshold $T>0$ and smoothing parameter $\eta>0$ such that with $B=B(T)$ defined as in \Cref{thm:B-op-norm-bounds}, with probability $1-\exp(-\Omega(n))$ over Haar-random $f$,
\begin{align}
    \|\Pi_B\wh{M}^{(\eta)}_{\wh f}\Pi_B\|_\op &\leq \frac12+\eps\\
    \text{and}\qquad \|\Pi_BD^{(\eta)}_{\wh f}\Pi_B\|_\op &\leq \eps\,.
\end{align}
\end{prop}

To accomplish this, we will compare both $\wh M^{(\eta)}_{\wh f}$ and $D^{(\eta)}_{\wh f}$ to Gaussian versions, decouple from $B$, and apply random paving.

\paragraph{Step 2: Gaussianize.}
With $\bs\gamma=(\gamma_s)_s$, $s\in\{0,1\}^n$, a family of i.i.d. standard complex Gaussians with variance $1/N$, define $\wh{M}^{(\eta)}_{\bs \gamma}$ and $D_{\bs\gamma}^{(\delta)}$ in analogy to $\wh{M}^{(\eta)}_{\wh f}$ and $D^{(\eta)}_{\wh f}$.
That is, define $\ket{g_{s,i}}=\gamma_s\ket{s}+\gamma_{s+i}\ket{s+i}$ and 
\[\wh{M}^{(\eta)}_{\bs \gamma}=H^{\otimes n}\cdot \frac{1}{2n}\sum_{i,s}\frac{\ketbra{g_{s,i}}{g_{s,i}}}{\|g_{s,i}\|_2^2+\eta/N}\cdot H^{\otimes n}\qquad\text{and}\qquad D_{\bs\gamma}^{(\eta)}=H^{\otimes n}\cdot \frac{1}{2n}\sum_{i,s}\frac{(\eta/N) I_{s,s+i}}{\|g_{s,i}\|_2^2+\eta/N}\cdot H^{\otimes n}\,.\]
Also define the scalars
\begin{equation}
\label{eq:m-and-d_eta-defn}
m_\eta =\E\frac{|\gamma_1|^2}{|\gamma_1|^2+|\gamma_2|^2+\eta/N}\qquad\text{and}\qquad d_\eta = \E\frac{\eta/N}{|\gamma_1|^2+|\gamma_2|^2+\eta/N}=1-2m_\eta\,,
\end{equation}
which will eventually be used to center the Gaussianized operators.
Then we have the following.


\begin{prop}[$BB$ Gaussianization]
\label{prop:BB-gaussianization}
    Fix $\eta>0$ and a measure $\mu$.
    Then there exists $C_\eta>0$ such that for any even $p\geq 8$ satisfying
    \[\Bias(\mu)^{2/p}\leq \frac12\,,\]
    the differences
\[
\begin{aligned}
&
\left|
    \left(\E\tr\big|\Pi_B \big(\wh{M}^{(\eta)}_{\wh f}-m_\eta I\big)\Pi_B\big|^p\right)^{1/p}
    -
    \left(\E\tr\big|\Pi_B \big(\wh{M}^{(\eta)}_{\bs \gamma}-m_\eta I\big)\Pi_B\big|^p\right)^{1/p}
\right|
\end{aligned}
\]
and
\[
\begin{aligned}
&
\left|
    \left(\E\tr\big|\Pi_B \big(D^{(\eta)}_{\wh f}-d_\eta I\big)\Pi_B\big|^p\right)^{1/p}
    -
    \left(\E\tr\big|\Pi_B \big(D^{(\eta)}_{\bs \gamma}-d_\eta I\big)\Pi_B\big|^p\right)^{1/p}
\right|
\end{aligned}
\]
are at most
\[C_\eta N^{1/p}\left[p^3\|\mu\|_{3/2}^{1/2}+\left(\frac{|B|}{N}+\max_{y\in \{0,1\}^n}\mu(B+y)\right)^{1/2}+\Bias(\mu)^{3/p}\right]\,.\]
Here $B+y$ denotes the set $\{b+y:b\in B\}$.
\end{prop}

\begin{proof}
Throughout the proof, $C_\eta$ denotes a constant depending only on $\eta$,
whose value may change from line to line.

Let us begin with the $\wh{M}^{(\eta)}_{\wh f}$ estimate.
We will proceed in two steps: first we do a Gaussian replacement of the Steinhaus-random vector $f$ via a Lindeberg-type argument.
Then, we move to Fourier space and consider the quantity of interest as a function of the random vector $\big(\wh f(s)\big)_s$.
With this perspective, we argue that small tracial moments of the corresponding quantity behave the same as if the Fourier coefficients $\big(\wh f(s)\big)_s$ were actually i.i.d. Gaussians.

\paragraph{Steinhaus to physical Gaussians.}
Let $\bs\xi=(\xi_x)_x$, where $\xi_x$, $x\in\{0,1\}^n$ are independent complex Gaussians
with variances $\mu(x)$ respectively.
We compare
$\wh M_{\wh f}^{(\eta)}$ and $\wh M_{\wh{\bs\xi}}^{(\eta)}$ by replacing the
physical coordinates one at a time.

Choose a physical coordinate $x$, fix a $\tilde z\in\C^{2^n}$, and let
$z=z(\zeta)$ be obtained from $\tilde z$ by replacing its $x$-th
coordinate with $\sqrt{\mu(x)}\,\zeta$.  Define
\[
    W_x(\zeta)
    :=
    \Pi_B
    \left(
        \wh M_{\wh{z(\zeta)}}^{(\eta)}-m_\eta I
    \right)
    \Pi_B.
\]
Note that $0\preceq
    \wh M_{\wh{z(\zeta)}}^{(\eta)}
    \preceq I$, so \[\|W_x(\zeta)\|_\op\leq1\,.\]
    
Since
\[
    \frechet_\zeta\widehat{z(\zeta)}(s)[v]
    =
    \sqrt{\frac{\mu(x)}{N}}
    (-1)^{x\cdot s}v,
\]
\Cref{lem:master-frechet} applied with regularization parameter $\eta/N$,
along with the chain rule, give for $r=1,2,3$,
\[
\begin{aligned}
    \left\|
        \frechet_\zeta^rW_x(\zeta)
    \right\|
    &\leq
    \frac1n\sum_i
    \max_{s:s_i=0}
    \left\|
        \frechet_\zeta^r
        \frac{\ketbra{u_{s,i}(\zeta)}{u_{s,i}(\zeta)}}
        {\|u_{s,i}(\zeta)\|_2^2+\eta/N}
    \right\| \\
    &\lesssim_r
    \left(\frac{N}{\eta}\right)^{r/2}
    \left(\frac{\mu(x)}{N}\right)^{r/2} \\
    &=
    \eta^{-r/2}\mu(x)^{r/2},
\end{aligned}
\]
where $\ket{u_{s,i}(\zeta)}
    :=
    \widehat{z(\zeta)}(s)\ket{s}
    +
    \widehat{z(\zeta)}(s+i)\ket{s+i}$.

Let $p$ be even.
Differentiating $\tr\left[W_x(\zeta)^p\right]$ three times and applying Schatten Hölder to the three possible
types of product-rule terms gives
\begin{equation}
\begin{aligned}
\left\|
    \frechet_\zeta^3\tr\left[W_x(\zeta)^p\right]
\right\|
\leq
C_\eta\mu(x)^{3/2}
\Big[
&
pN^{1/p}\|W_x(\zeta)\|_{S_p}^{p-1}\\
&+
p^2N^{2/p}\|W_x(\zeta)\|_{S_p}^{p-2}\\
&+
p^3N^{3/p}\|W_x(\zeta)\|_{S_p}^{p-3}
\Big],
\end{aligned}
\label{ineq:refined-third-derivative}
\end{equation}
where $\|\cdot\|_p$ denotes the $p$th Schatten norm.

We now perform the Lindeberg replacement at the level of the $p$-th moment roots.
Enumerate the physical coordinates as $x_1,\ldots,x_N$, and let
$f^{(k)}$ be the hybrid vector in which the first $k$ coordinates use
Steinhaus variables and the remaining coordinates use Gaussian variables; so for example $f^{(0)}=\bs\xi$ and $f^{(N)}=f$.

Choose $K_\eta$ sufficiently large and set
\[
    \Lambda
    :=
    K_\eta p^3\|\mu\|_{3/2}^{1/2}.
\]
For $0\leq k\leq N$, put
\[
    R_k
    :=
    \left(
        N\Lambda^p
        +
        \E\tr\left|
            \Pi_B
            \left(
                \wh M_{\widehat{f^{(k)}}}^{(\eta)}-m_\eta I
            \right)
            \Pi_B
        \right|^p
    \right)^{1/p},
\]
and let
\[
    R_*:=\max_{0\leq k\leq N}R_k.
\]
In particular,
\begin{equation}
    R_*\geq N^{1/p}\Lambda.
\label{ineq:root-floor}
\end{equation}

Fix one replacement coordinate $x$.
The first
derivative bound above gives
\begin{equation}
    \|W_x(t\zeta)-W_x(0)\|_{S_p}
    \leq
    C_\eta N^{1/p}\sqrt{\mu(x)}\,|\zeta|.
\label{ineq:segment-lipschitz}
\end{equation}
Moreover, Minkowski's inequality and the Gaussian moment bound
$(\E|g|^p)^{1/p}\lesssim\sqrt p$ give
\begin{equation}
\begin{aligned}
&
\left(
    N\Lambda^p
    +
    \E\|W_{x,0}\|_{S_p}^p
\right)^{1/p}\\
&\qquad\leq
R_*
+
C_\eta N^{1/p}\sqrt{p\mu(x)}\\
&\qquad\leq
R_*
\left(
    1+\frac{C_\eta}{K_\eta p^{5/2}}
\right),
\end{aligned}
\label{ineq:zeroed-root}
\end{equation}
where we also used $\sqrt{\mu(x)}
    \leq
    \|\mu\|_{3/2}^{1/2}$.

The pointwise quantity
\[
    \left(
        N\Lambda^p+\|W_{x,0}\|_{S_p}^p
    \right)^{1/p}
\]
is at least $N^{1/p}\Lambda$.  Combining this observation with
\eqref{ineq:segment-lipschitz} gives, for $j=1,2,3$,
\[
\begin{aligned}
\|W_x(t\zeta)\|_{S_p}^{p-j}
&\leq
\left(
    N\Lambda^p+\|W_{x,0}\|_{S_p}^p
\right)^{(p-j)/p}
\exp\left(
    \frac{C_\eta p\sqrt{\mu(x)}}{\Lambda}
    |\zeta|
\right).
\end{aligned}
\]
Since
\[
    \frac{p\sqrt{\mu(x)}}{\Lambda}
    \leq
    \frac{1}{K_\eta p^2},
\]
\eqref{ineq:zeroed-root} and the exponential moments of a standard complex
Gaussian imply
\begin{equation}
    \E\left[
        |\zeta|^3
        \sup_{0\leq t\leq1}
        \|W_x(t\zeta)\|_{S_p}^{p-j}
    \right]
    \leq
    C_\eta R_*^{p-j},
    \qquad
    j=1,2,3,
\label{ineq:taylor-segment-root}
\end{equation}
for both a Steinhaus variable and a standard complex Gaussian.

Taylor expansion of $|W_x(\cdot)|^p$ to second order at zero, together with the matching first two
real moments of the Steinhaus and complex Gaussian variables, now gives
\begin{equation}
\begin{aligned}
&
\left|
    \E\tr\left|
        \Pi_B
        \left(
            \wh M_{\widehat{f^{(k)}}}^{(\eta)}-m_\eta I
        \right)
        \Pi_B
    \right|^p
    -
    \E\tr\left|
        \Pi_B
        \left(
            \wh M_{\widehat{f^{(k-1)}}}^{(\eta)}-m_\eta I
        \right)
        \Pi_B
    \right|^p
\right|\\
&\qquad\leq
C_\eta\mu(x_k)^{3/2}
\left[
    pN^{1/p}R_*^{p-1}
    +
    p^2N^{2/p}R_*^{p-2}
    +
    p^3N^{3/p}R_*^{p-3}
\right].
\end{aligned}
\label{ineq:one-coordinate-root-lindeberg}
\end{equation}

Let $k_*$ be such that $R_{k_*}=R_*$.  Telescoping from either endpoint to
$k_*$ and summing \eqref{ineq:one-coordinate-root-lindeberg} over the physical
coordinates gives, for $j\in\{0,N\}$,
\begin{equation}
\begin{aligned}
R_*-R_j
\leq
C_\eta N^{1/p}
\left[
    p\sum_x\mu(x)^{3/2}
    +
    \frac{p^2}{\Lambda}\sum_x\mu(x)^{3/2}
    +
    \frac{p^3}{\Lambda^2}\sum_x\mu(x)^{3/2}
\right].
\end{aligned}
\label{ineq:max-hybrid-root}
\end{equation}
Using the definition of $\Lambda$, the fact that
$\sum_x\mu(x)^{3/2}\leq1$, and
\[
    \|\mu\|_{3/2}^{1/2}
    =
    \left(\sum_x\mu(x)^{3/2}\right)^{1/3},
\]
we obtain
\[
    R_*-R_j
    \leq
    C_\eta N^{1/p}p^3\|\mu\|_{3/2}^{1/2}.
\]
Removing the temporary floor $N\Lambda^p$ and using the last estimate for both
endpoints gives the desired comparison to physical Gaussians:
\begin{equation}
\begin{aligned}
&
\left|
    \left(
        \E\tr\left|
            \Pi_B
            \left(
                \wh M_{\wh f}^{(\eta)}-m_\eta I
            \right)
            \Pi_B
        \right|^p
    \right)^{1/p}
    -
    \left(
        \E\tr\left|
            \Pi_B
            \left(
                \wh M_{\wh{\bs\xi}}^{(\eta)}-m_\eta I
            \right)
            \Pi_B
        \right|^p
    \right)^{1/p}
\right|\\
&\qquad\leq
C_\eta N^{1/p}p^3\|\mu\|_{3/2}^{1/2}.
\end{aligned}
\label{ineq:root-lindeberg}
\end{equation}

\paragraph{Physical Gaussians to independent Fourier Gaussians.}

Next we pass from $\wh{M}^{(\eta)}_{\wh{\bs\xi}}$ to
$\wh{M}^{(\eta)}_{\bs\gamma}$ by Gaussian interpolation.
Let $X(\tau)$, $0\leq \tau\leq 1$ be a Gaussian vector interpolating between
$\big(\wh{\bs\xi}(s)\big)_s$ at $X(0)$ and the $\big(\bs\gamma(s)\big)_s$ at $X(1)$.
That is, with $\Gamma_0$ the covariance of $\wh{\bs\xi}$ and $\Gamma_1=N^{-1}I$ the
covariance of $\bs\gamma$, define $X(\tau)$ to have
covariance $
    (1-\tau)\Gamma_0+\tau\Gamma_1$.
Also define
\[
    \Delta_{st}
    :=
    (\Gamma_1-\Gamma_0)_{st}.
\]
The covariances of $\wh{\bs\xi}$ are $1/N$ on the diagonal, while for
$s\neq t$,
\[
    |(\Gamma_0)_{st}|=
    \left|
        \E\wh{\bs\xi}(s)
        \overline{\wh{\bs\xi}(t)}
    \right|=
    \left|
        \frac1N
        \sum_x
        \mu(x)(-1)^{(s+t)\cdot x}
    \right|\le
    \frac{\Bias(\mu)}{N}.
\]
Thus
\begin{equation}
\label{ineq:covariance-entry-bound}
    \Delta_{ss}=0
    \qquad\text{and}\qquad
    |\Delta_{st}|
    \le
    \frac{\Bias(\mu)}{N}
    \quad\text{for all }s\neq t.
\end{equation}

Define
\[
    W:=W(\tau)
    :=
    \Pi_B
    \left(
        \wh M^{(\eta)}_{X(\tau)}-m_\eta I
    \right)
    \Pi_B.
\]
Since $W(\tau)$ is self-adjoint and $p$ is even,
\[
    \tr|W(\tau)|^p
    =
    \tr[W(\tau)^p].
\]
The fundamental theorem of calculus and Gaussian covariance interpolation give
\begin{align*}
&
\E\tr\left|
    \Pi_B
    \left(
        \wh M^{(\eta)}_{\bs\gamma}-m_\eta I
    \right)
    \Pi_B
\right|^p
-
\E\tr\left|
    \Pi_B
    \left(
        \wh M^{(\eta)}_{\wh{\bs\xi}}-m_\eta I
    \right)
    \Pi_B
\right|^p
\\
&\qquad=
\int_0^1
\frac{d}{d\tau}
\E\tr[W(\tau)^p]
\,d\tau
\\
&\qquad=
\int_0^1
\E
\sum_{s,t}
\Delta_{st}
\partial_s\bar\partial_t
\tr[W(\tau)^p]
\,d\tau
\\
&\qquad=
\int_0^1
\E
\sum_{s\neq t}
\Delta_{st}
\partial_s\bar\partial_t
\tr[W(\tau)^p]
\,d\tau.
\end{align*}
Here $\partial_s$ and $\bar\partial_t$ denote the Wirtinger derivatives in the
Fourier coordinates:
\[
    \partial_s
    :=
    \frac{\partial}{\partial z_s},
    \qquad
    \bar\partial_t
    :=
    \frac{\partial}{\partial\overline z_t}.
\]

We now pursue a uniform bound on the interpolation integrand.
By the product rule and cyclicity of trace,
\begin{align}
\label{eq:second-derivative-trace-power}
    \partial_s\bar\partial_t\tr[W^p]
    &=
    p\underbrace{
        \tr\left[
            (\partial_s\bar\partial_tW)W^{p-1}
        \right]
    }_{=:\,T_1(s,t)}+
    p\sum_{\substack{\ell,\ell'\geq0\\\ell+\ell'=p-2}}
    \underbrace{
        \tr\left[
            (\partial_sW)W^\ell
            (\bar\partial_tW)W^{\ell'}
        \right]
    }_{=:\,T_2(s,t,\ell)}.
\end{align}
We first control the covariance-weighted $T_1$ terms.
We then treat separately the interior $T_2$ terms, for which
$1\le\ell\le p-3$, and the two endpoint terms $\ell=0,p-2$.

\paragraph{The $T_1$ terms.}
Recall that
\[
    W(\tau)
    =
    \Pi_BH^{\otimes n}
    \left(
        \frac{1}{2n}
        \sum_{i,s}
        \frac{\ketbra{u_{s,i}}{u_{s,i}}}
        {\|u_{s,i}\|_2^2+\eta/N}
        -
        m_\eta I
    \right)
    H^{\otimes n}\Pi_B,
\]
where $\ket{u_{s,i}}
    =
    X(\tau)_s\ket{s}
    +
    X(\tau)_{s+e_i}\ket{s+e_i}$.
Since $\Delta_{ss}=0$, only mixed derivatives with respect to distinct Fourier coordinates contribute, and for a fixed direction $i$, such a mixed derivative is nonzero only when the two coordinates are the endpoints of a common $i$-edge.
Grouping the two directed copies of each $i$-edge, we therefore have
\begin{align*}
\sum_{r,t}
\Delta_{rt}\partial_r\bar\partial_tW
&\;=\;
\Pi_BH^{\otimes n}
\Bigg[
\frac1n\sum_i
\bigoplus_{s:s_i=0}
\Bigg(
\Delta_{s,s+e_i}
\partial_s\bar\partial_{s+e_i}
\frac{\ketbra{u_{s,i}}{u_{s,i}}}
{\|u_{s,i}\|_2^2+\eta/N}
\\
&\hspace{15em}
+
\Delta_{s+e_i,s}
\partial_{s+e_i}\bar\partial_s
\frac{\ketbra{u_{s,i}}{u_{s,i}}}
{\|u_{s,i}\|_2^2+\eta/N}
\Bigg)
\Bigg]
H^{\otimes n}\Pi_B.
\end{align*}
By \Cref{lem:master-frechet}, each local mixed derivative in this display has operator norm at most $C_\eta N$.
Using the direct-sum structure, unitarity of $H^{\otimes n}$, contractivity of $\Pi_B$, and $|\Delta_{rt}|\le\Bias(\mu)/N$, we obtain
\begin{align}
\label{ineq:weighted-mixed-derivative-bound}
\left\|
    \sum_{r,t}
    \Delta_{rt}\partial_r\bar\partial_tW
\right\|_\op
&\le
\frac1n\sum_i
\max_{s:s_i=0}
\Bigg[
|\Delta_{s,s+e_i}|
\left\|
\partial_s\bar\partial_{s+e_i}
\frac{\ketbra{u_{s,i}}{u_{s,i}}}
{\|u_{s,i}\|_2^2+\eta/N}
\right\|_\op
\nonumber\\
&\hspace{8em}
+
|\Delta_{s+e_i,s}|
\left\|
\partial_{s+e_i}\bar\partial_s
\frac{\ketbra{u_{s,i}}{u_{s,i}}}
{\|u_{s,i}\|_2^2+\eta/N}
\right\|_\op
\Bigg]
\nonumber\\
&\le
C_\eta\Bias(\mu).
\end{align}
Consequently, Schatten Hölder gives
\begin{align}
\label{ineq:T1-refined}
\left|
    \sum_{s,t}
    \Delta_{st}T_1(s,t)
\right|
&=
\left|
    \tr\left[
        \left(
            \sum_{s,t}
            \Delta_{st}\partial_s\bar\partial_tW
        \right)
        W^{p-1}
    \right]
\right|
\nonumber\\
&\le
\left\|
    \sum_{s,t}
    \Delta_{st}\partial_s\bar\partial_tW
\right\|_{S_p}
\left\|W^{p-1}\right\|_{S_{p/(p-1)}}
\nonumber\\
&\le
C_\eta
N^{1/p}
\Bias(\mu)
\|W\|_{S_p}^{p-1}.
\end{align}

\paragraph{The $T_2$ terms.}
Fix $\ell$ and $\ell'=p-2-\ell$.
For this fixed choice,
\begin{align*}
\sum_{s,t}\Delta_{st}T_2(s,t,\ell)
&=
\sum_{s,t}\Delta_{st}
\tr\left[
    (\partial_sW)W^\ell
    (\bar\partial_tW)W^{\ell'}
\right]\\
&=
\sum_{s,t}\Delta_{st}
\tr\left[
    (\partial_sM^{(\eta)}_{X(\tau)})
    \underbrace{
        H^{\otimes n}\Pi_BW^\ell\Pi_BH^{\otimes n}
    }_{=:\,Q}
    (\bar\partial_tM^{(\eta)}_{X(\tau)})
    \underbrace{
        H^{\otimes n}\Pi_BW^{\ell'}\Pi_BH^{\otimes n}
    }_{=:\,Q'}
\right].
\end{align*}

Expand the Fourier-side operator into its $2\times2$ edge blocks:
\[
    M^{(\eta)}_{X(\tau)}
    =
    \frac1{2n}
    \sum_i
    \sum_r
    \sum_{a,b\in\{0,1\}}
    c_{r,a,b,i}
    \ketbra{r+ae_i}{r+be_i}.
\]
The derivative $\partial_s c_{r,a,b,i}$ vanishes unless $r=s$ or $r=s+e_i$.
Define
\[
\begin{aligned}
    c'_{s,a,b,i}
    &:=
    \partial_s c_{s,a,b,i}
    +
    \partial_s c_{s+e_i,a+1,b+1,i},
    \\
    \widetilde c'_{t,c,d,j}
    \text{and}\qquad &:=
    \bar\partial_t c_{t,c,d,j}
    +
    \bar\partial_t c_{t+e_j,c+1,d+1,j},
\end{aligned}
\]
where the additions in $a,b,c,d$ are modulo $2$.
With these definitions,
\[
    \partial_sM^{(\eta)}_{X(\tau)}
    =
    \frac1{2n}
    \sum_i
    \sum_{a,b\in\{0,1\}}
    c'_{s,a,b,i}
    \ketbra{s+ae_i}{s+be_i},
\]
and
\[
    \bar\partial_tM^{(\eta)}_{X(\tau)}
    =
    \frac1{2n}
    \sum_j
    \sum_{c,d\in\{0,1\}}
    \widetilde c'_{t,c,d,j}
    \ketbra{t+ce_j}{t+de_j}.
\]
By \Cref{lem:master-frechet},
\begin{equation}
\label{ineq:raw-fourier-derivative-coefficients}
    |c'_{s,a,b,i}|,
    |\widetilde c'_{t,c,d,j}|
    \le
    C_\eta\sqrt N.
\end{equation}
Substituting these expansions into the preceding covariance-weighted sum gives
\begin{align}
\label{eq:weighted-T2-matrix-entry-expansion}
\sum_{s,t}\Delta_{st}T_2(s,t,\ell)
&=
\frac1{4n^2}
\sum_{i,j}
\sum_{a,b,c,d\in\{0,1\}}
\sum_{s,t}
\Delta_{st}
c'_{s,a,b,i}
\widetilde c'_{t,c,d,j}
\nonumber\\
&\qquad\qquad\cdot
\tr\left[
    \ketbra{s+ae_i}{s+be_i}
    Q
    \ketbra{t+ce_j}{t+de_j}
    Q'
\right]
\nonumber\\
&=
\frac1{4n^2}
\sum_{i,j}
\sum_{a,b,c,d}
\sum_{s,t}
\Delta_{st}
c'_{s,a,b,i}
\widetilde c'_{t,c,d,j}
\nonumber\\
&\qquad\qquad\cdot
\bra{t+de_j}Q'\ket{s+ae_i}
\bra{s+be_i}Q\ket{t+ce_j}.
\end{align}

\paragraph{Interior $T_2$ terms.}
Suppose that
\[
    1\le\ell\le p-3.
\]
Let $S_u$ denote translation in the Fourier coordinates:
\[
    S_u\ket{s}
    :=
    \ket{s+u}.
\]
Let
\[
    \mathcal S_\Delta(A)
    :=
    \Delta\circ A
\]
denote Schur multiplication by $\Delta$.
Writing $\chi_x(s)
    :=
    (-1)^{x\cdot s}$,
we have the exact identity
\begin{equation}
\label{eq:schur-walsh-representation}
    \mathcal S_\Delta(A)
    =
    \frac1N
    \sum_x
    \left(
        \frac1N-\mu(x)
    \right)
    \operatorname{diag}(\chi_x)\,
    A\,
    \operatorname{diag}(\chi_x).
\end{equation}
Indeed, the $(s,t)$ entry of the right-hand side is
\[
    \frac1N
    \sum_x
    \left(
        \frac1N-\mu(x)
    \right)
    (-1)^{x\cdot(s+t)}
    A_{st}
    =
    \Delta_{st}A_{st}.
\]
Since
\[
    \sum_x
    \left|
        \frac1N-\mu(x)
    \right|
    \le2,
\]
we obtain
\[
    \|\mathcal S_\Delta\|_{S_\infty\to S_\infty}
    \le
    \frac2N.
\]
Moreover, since a Schur multiplier is diagonal in the matrix-unit basis,
\[
    \|\mathcal S_\Delta\|_{S_2\to S_2}
    =
    \max_{s,t}|\Delta_{st}|
    \le
    \frac{\Bias(\mu)}N.
\]
Interpolation between these two bounds, together with duality, gives
\begin{equation}
\label{ineq:schur-Sq-bound}
    \|\mathcal S_\Delta(A)\|_{S_q}
    \le
    \frac2N
    \Bias(\mu)^{
        2\min\left\{
            \frac1q,
            1-\frac1q
        \right\}
    }
    \|A\|_{S_q}
\end{equation}
for every $1\le q\le\infty$.

For fixed $i,j,a,b,c,d$, the inner sum over $s,t$ in
\eqref{eq:weighted-T2-matrix-entry-expansion} is exactly
\begin{equation}
\label{eq:exact-schur-factorization}
\begin{aligned}
\tr\Bigg[
&
\operatorname{diag}
    \bigl(
        \widetilde c'_{t,c,d,j}
    \bigr)_t
S_{de_j}^*
Q'
S_{ae_i} \cdot
\mathcal S_\Delta\left(
    \operatorname{diag}
        \bigl(
            c'_{s,a,b,i}
        \bigr)_s
    S_{be_i}^*
    Q
    S_{ce_j}
\right)
\Bigg].
\end{aligned}
\end{equation}
Indeed, expanding the trace in \eqref{eq:exact-schur-factorization} gives
\[
\begin{aligned}
\sum_{s,t}
&
\Delta_{st}
c'_{s,a,b,i}
\widetilde c'_{t,c,d,j}
\bra{t+de_j}Q'\ket{s+ae_i}
\bra{s+be_i}Q\ket{t+ce_j},
\end{aligned}
\]
which is the corresponding term in
\eqref{eq:weighted-T2-matrix-entry-expansion}.

We apply Schatten Hölder with the conjugate exponents
\[
    \frac{p}{\ell'}
    \qquad\text{and}\qquad
    \frac{p}{\ell+2}.
\]
The shifts are unitary, and $W=\Pi_BW\Pi_B$, so
\begin{align}
\label{ineq:first-interior-factor-bound}
&
\left\|
    \operatorname{diag}
        \bigl(
            \widetilde c'_{t,c,d,j}
        \bigr)_t
    S_{de_j}^*
    Q'
    S_{ae_i}
\right\|_{S_{p/\ell'}}
\le
C_\eta\sqrt N
\|Q'\|_{S_{p/\ell'}}
=
C_\eta\sqrt N
\|W\|_{S_p}^{\ell'}.
\end{align}
Similarly,
\begin{align}
\label{ineq:second-interior-factor-bound}
&
\left\|
    \operatorname{diag}
        \bigl(
            c'_{s,a,b,i}
        \bigr)_s
    S_{be_i}^*
    Q
    S_{ce_j}
\right\|_{S_{p/(\ell+2)}}
\le
C_\eta\sqrt N
\|W^\ell\|_{S_{p/(\ell+2)}}
\le
C_\eta\sqrt N
N^{2/p}
\|W\|_{S_p}^{\ell}.
\end{align}
The last inequality is the finite-dimensional Schatten embedding
\[
    \|W^\ell\|_{S_{p/(\ell+2)}}
    \le
    N^{2/p}
    \|W\|_{S_p}^{\ell}.
\]
Applying Schatten Hölder to
\eqref{eq:exact-schur-factorization}, and then using
\eqref{ineq:schur-Sq-bound},
\eqref{ineq:first-interior-factor-bound}, and
\eqref{ineq:second-interior-factor-bound}, gives
\[
\begin{aligned}
&
\left|
\tr\Bigg[
\operatorname{diag}
    \bigl(
        \widetilde c'_{t,c,d,j}
    \bigr)_t
S_{de_j}^*
Q'
S_{ae_i}
\cdot
\mathcal S_\Delta\left(
    \operatorname{diag}
        \bigl(
            c'_{s,a,b,i}
        \bigr)_s
    S_{be_i}^*
    Q
    S_{ce_j}
\right)
\Bigg]
\right|
\\
&\qquad\le
C_\eta
\Bias(\mu)^{
    2\min\left\{
        \frac{\ell+2}{p},
        \frac{\ell'}{p}
    \right\}
}
N^{2/p}
\|W\|_{S_p}^{p-2}.
\end{aligned}
\]
Here the factor $1/N$ from the Schur-multiplier estimate cancels the two
$\sqrt N$ factors from the differentiated edge coefficients.
The prefactor $1/(4n^2)$ in
\eqref{eq:weighted-T2-matrix-entry-expansion} cancels the $n^2$ choices of
$i,j$, up to an absolute constant, while the choices of $a,b,c,d$ contribute
another absolute constant.
Consequently,
\begin{align}
\label{ineq:interior-one-sided-schur}
&
\left|
    \sum_{s,t}
    \Delta_{st}T_2(s,t,\ell)
\right|\le
C_\eta
N^{2/p}
\Bias(\mu)^{
    2\min\left\{
        \frac{\ell+2}{p},
        \frac{\ell'}{p}
    \right\}
}
\|W\|_{S_p}^{p-2}.
\end{align}

Since $\Delta$ is symmetric,
\[
    \tr\left[A\mathcal S_\Delta(B)\right]
    =
    \tr\left[\mathcal S_\Delta(A)B\right].
\]
We may therefore apply the Schur-multiplier estimate to the factor containing
$Q'$ instead.
Taking the better of the two bounds gives
\begin{equation}
\label{ineq:interior-schur-bound}
    \left|
        \sum_{s,t}
        \Delta_{st}T_2(s,t,\ell)
    \right|
    \le
    C_\eta
    N^{2/p}
    \Bias(\mu)^{\theta_\ell}
    \|W\|_{S_p}^{p-2},
\end{equation}
where
\[
\theta_\ell
:=
\max\left\{
    2\min\left(
        \frac{\ell+2}{p},
        \frac{\ell'}{p}
    \right),
    2\min\left(
        \frac{\ell'+2}{p},
        \frac{\ell}{p}
    \right)
\right\}.
\]
As $\ell$ ranges over $1,\ldots,p-3$, each power
$\Bias(\mu)^{2k/p}$ with $k\ge3$ occurs at most four times among the
quantities $\Bias(\mu)^{\theta_\ell}$.
Therefore,
\begin{align}
\label{ineq:interior-exponent-sum}
    \sum_{\ell=1}^{p-3}
    \Bias(\mu)^{\theta_\ell}\le
    4
    \sum_{k=3}^{\infty}
    \Bias(\mu)^{2k/p}=
    \frac{
        4\Bias(\mu)^{6/p}
    }{
        1-\Bias(\mu)^{2/p}
    }\le
    8\Bias(\mu)^{6/p}.
\end{align}

\paragraph{Endpoint $T_2$ terms.}
It remains to treat $\ell=0$ and $\ell=p-2$.
When $\ell=0$,
\[
    Q
    =
    H^{\otimes n}\Pi_BH^{\otimes n},
\qquad\text{and}\qquad
    Q'
    =
    H^{\otimes n}\Pi_BW^{p-2}\Pi_BH^{\otimes n}.
\]
For arbitrary shifts $u,v\in\{0,1\}^n$, one has
\begin{equation}
\label{ineq:endpoint-schur-bound}
\left\|
    \mathcal S_\Delta\left(
        S_u^*
        H^{\otimes n}\Pi_BH^{\otimes n}
        S_v
    \right)
\right\|_\op
\le
\frac1N
\left(
    \frac{|B|}{N}
    +
    \max_y\mu(B+y)
\right).
\end{equation}
Indeed,
\[
\begin{aligned}
&
\left(
    S_u^*
    H^{\otimes n}\Pi_BH^{\otimes n}
    S_v
\right)_{s,t}
=
\frac1N
\sum_{y\in B}
(-1)^{y\cdot(u+v)}
(-1)^{y\cdot(s+t)},
\end{aligned}
\]
while
\[
    \Delta_{s,t}
    =
    \frac1N
    \sum_x
    \left(
        \frac1N-\mu(x)
    \right)
    (-1)^{x\cdot(s+t)}.
\]
Their Schur product is Walsh-diagonalizable, with eigenvalues
\[
    \frac1N
    \sum_x
    \left(
        \frac1N-\mu(x)
    \right)
    \mathbf 1_B(z+x)
    (-1)^{(z+x)\cdot(u+v)}.
\]
The absolute value of every such eigenvalue is at most
\[
    \frac1N
    \left(
        \frac{|B|}{N}
        +
        \mu(B+z)
    \right),
\]
which proves \eqref{ineq:endpoint-schur-bound}.

Schur multiplication commutes with multiplication by a diagonal matrix on the
left, so
\[
\begin{aligned}
&
\left\|
    \mathcal S_\Delta\left(
        \operatorname{diag}
            \bigl(
                c'_{s,a,b,i}
            \bigr)_s
        S_{be_i}^*
        Q
        S_{ce_j}
    \right)
\right\|_\op\le
C_\eta\sqrt N
\frac1N
\left(
    \frac{|B|}{N}
    +
    \max_y\mu(B+y)
\right).
\end{aligned}
\]
The other factor in
\eqref{eq:exact-schur-factorization} satisfies
\[
\begin{aligned}
&
\left\|
    \operatorname{diag}
        \bigl(
            \widetilde c'_{t,c,d,j}
        \bigr)_t
    S_{de_j}^*
    Q'
    S_{ae_i}
\right\|_{S_1}\le
C_\eta\sqrt N
\|W^{p-2}\|_{S_1}.
\end{aligned}
\]
The two $\sqrt N$ factors again cancel the $1/N$ in
\eqref{ineq:endpoint-schur-bound}.
Moreover,
\[
    \|W^{p-2}\|_{S_1}
    \le
    N^{2/p}
    \|W\|_{S_p}^{p-2}.
\]
Applying Schatten Hölder in
\eqref{eq:exact-schur-factorization}, and then summing over
$i,j,a,b,c,d$, gives
\begin{equation}
\label{ineq:endpoint-T2-bound}
\begin{aligned}
&
\left|
    \sum_{s,t}
    \Delta_{st}T_2(s,t,0)
\right|
\le
C_\eta
N^{2/p}
\left(
    \frac{|B|}{N}
    +
    \max_y\mu(B+y)
\right)
\|W\|_{S_p}^{p-2}.
\end{aligned}
\end{equation}
The endpoint $\ell=p-2$ is identical, with the two differentiated factors
interchanged.

We now combine the $T_1$ terms, the interior $T_2$ terms, and the endpoint
$T_2$ terms.
Define
\[
    R(\tau)
    :=
    \left(
        \E\tr|W(\tau)|^p
    \right)^{1/p}.
\]
For $j=1,2$, Hölder's inequality gives
\[
    \E\|W(\tau)\|_{S_p}^{p-j}
    \le
    R(\tau)^{p-j}.
\]
Using \eqref{ineq:T1-refined},
\eqref{ineq:interior-exponent-sum}, and
\eqref{ineq:endpoint-T2-bound} in the Gaussian interpolation identity gives
\begin{equation}
\label{ineq:root-covariance-differential}
\begin{aligned}
\left|
    \frac{d}{d\tau}
    \E\tr|W(\tau)|^p
\right|
\le
C_\eta p\Bigg[
&
N^{1/p}
\Bias(\mu)
R(\tau)^{p-1}
\\
&+
N^{2/p}
\left(
    \frac{|B|}{N}
    +
    \max_y\mu(B+y)
    +
    \Bias(\mu)^{6/p}
\right)
R(\tau)^{p-2}
\Bigg].
\end{aligned}
\end{equation}

We now convert this differential inequality for the $p$-th moment into one
for its $p$-th root.
At every point where $R(\tau)>0$,
\eqref{ineq:root-covariance-differential} gives
\[
\begin{aligned}
    |R'(\tau)|
    \le
    C_\eta\Bigg[
    N^{1/p}\Bias(\mu)+
    \frac{
        N^{2/p}
        \left(
            \frac{|B|}{N}
            +
            \max_y\mu(B+y)
            +
            \Bias(\mu)^{6/p}
        \right)
    }{
        R(\tau)
    }
    \Bigg].
\end{aligned}
\]
Consequently,
\begin{align*}
&
\left|
    \frac{d}{d\tau}
    \left(
        R(\tau)^2
        +
        N^{2/p}
        \left[
            \frac{|B|}{N}
            +
            \max_y\mu(B+y)
            +
            \Bias(\mu)^{6/p}
        \right]
    \right)^{1/2}
\right|
\\
&\qquad\le
C_\eta
N^{1/p}
\left[
    \Bias(\mu)
    +
    \left(
        \frac{|B|}{N}
        +
        \max_y\mu(B+y)
        +
        \Bias(\mu)^{6/p}
    \right)^{1/2}
\right].
\end{align*}
(At points where $R(\tau)=0$, one may apply the same argument after replacing
$\E\tr|W(\tau)|^p$ by
$\E\tr|W(\tau)|^p+\delta^p,$
and then let $\delta\downarrow0$.)
Integrating over $\tau\in[0,1]$ gives the same bound for the difference of the
square-root expressions at $\tau=0$ and $\tau=1$.
For every $\tau$,
\[
\begin{aligned}
0
&\le
\left(
    R(\tau)^2
    +
    N^{2/p}
    \left[
        \frac{|B|}{N}
        +
        \max_y\mu(B+y)
        +
        \Bias(\mu)^{6/p}
    \right]
\right)^{1/2}
-
R(\tau)
\\
&\le
N^{1/p}
\left(
    \frac{|B|}{N}
    +
    \max_y\mu(B+y)
    +
    \Bias(\mu)^{6/p}
\right)^{1/2}.
\end{aligned}
\]
It follows that
\[
\begin{aligned}
    |R(1)-R(0)|
    \le
    C_\eta
    N^{1/p}
    \Bigg[
    \Bias(\mu)+
    \left(
        \frac{|B|}{N}
        +
        \max_y\mu(B+y)
        +
        \Bias(\mu)^{6/p}
    \right)^{1/2}
    \Bigg].
\end{aligned}
\]
Since
\[
    \Bias(\mu)
    \le
    \Bias(\mu)^{3/p}
\]
and
\[
\begin{aligned}
&
\left(
    \frac{|B|}{N}
    +
    \max_y\mu(B+y)
    +
    \Bias(\mu)^{6/p}
\right)^{1/2}\le
\left(
    \frac{|B|}{N}
    +
    \max_y\mu(B+y)
\right)^{1/2}
+
\Bias(\mu)^{3/p},
\end{aligned}
\]
we conclude that
\begin{equation}
\label{ineq:root-covariance-transfer}
\begin{aligned}
&
\left|
    \left(
        \E\tr\left|
            \Pi_B
            \left(
                \wh M_{\wh{\bs\xi}}^{(\eta)}-m_\eta I
            \right)
            \Pi_B
        \right|^p
    \right)^{1/p}
    -
    \left(
        \E\tr\left|
            \Pi_B
            \left(
                \wh M_{\bs\gamma}^{(\eta)}-m_\eta I
            \right)
            \Pi_B
        \right|^p
    \right)^{1/p}
\right|
\\
&\qquad\le
C_\eta
N^{1/p}
\left[
    \left(
        \frac{|B|}{N}
        +
        \max_y\mu(B+y)
    \right)^{1/2}
    +
    \Bias(\mu)^{3/p}
\right].
\end{aligned}
\end{equation}
Combining \eqref{ineq:root-lindeberg} and
\eqref{ineq:root-covariance-transfer} proves the asserted estimate for
$\wh M^{(\eta)}$.

Finally, consider $D^{(\eta)}$.
On each Fourier edge, its local block is
\[
    Q_\eta(u)
    =
    \frac{\eta/N}{\|u\|_2^2+\eta/N}I_2
    =
    \left(
        1-
        \tr
        \frac{uu^*}{\|u\|_2^2+\eta/N}
    \right)I_2.
\]
Its first three derivatives have the same scales as those of the regularized
projector block:
\[
    \|D^rQ_\eta(u)\|
    \lesssim_r
    \left(
        \frac{N}{\eta}
    \right)^{r/2},
    \qquad
    r=1,2,3.
\]
Consequently, the root-level Lindeberg argument, the covariance-weighted
$T_1$ estimate, the interior $T_2$ argument, and the endpoint $T_2$ argument
all apply identically to
\[
    \Pi_B
    \left(
        D^{(\eta)}-d_\eta I
    \right)
    \Pi_B.
\]
This proves the asserted estimate for $D^{(\eta)}$ and completes the proof.
\end{proof}

\paragraph{Step 3: Random pavings.}
We now apply \Cref{thm:tropp-random-paving} to $\wh M^{(\eta)}_{\bs \gamma}$ to complete the argument for the $BB$ case.

\begin{prop}[Gaussian random-paving bounds for the \(BB\) objects]
\label{prop:BB-gaussian-paving}
For every $\eps,\eta>0$, there exists a density threshold
$\beta_{\max}>0$ such that the following holds for $n$ large enough.

Let \(\bs\gamma=(\gamma_s)_{s\in\{0,1\}^n}\) be i.i.d. complex Gaussians with variance $1/N$ and let
$B\subseteq\{0,1\}^n$ be an exchangeable random subset, independent of $\bs \gamma$.
Then for every even $p\le n$, 
\[
    \E \left[\mathbf{1}_{|B|/N\leq \beta_{\max}}\cdot\tr\left|\Pi_B \left(\wh{M}^{(\eta)}_{\bs \gamma}-m_\eta I\right)\Pi_B\right|^p 
\right]^{1/p}
\le
N^{1/p}\eps
\]
and
\[
    \E \left[\mathbf{1}_{|B|/N\leq \beta_{\max}}\cdot\tr\left|
            \Pi_B\left(D_{\bs \gamma}^{(\eta)}-d_\eta I\right)\Pi_B
        \right|^p
\right]^{1/p}
\le
N^{1/p}\eps.
\]
\end{prop}

\begin{proof}
    Write
    \[\wh{M}^{(\eta)}_{\bs \gamma}-m_\eta I\;=\;H^{\otimes n}\left(\frac{1}{n}\sum_{i}P_{i}\right)H^{\otimes n}, \qquad P_i = \sum_{s:s_i=0}\ketbra{s_{[1,i-1]}}{s_{[1,i-1]}}\otimes P_{i,s}\otimes \ketbra{s_{[i+1,n]}}{s_{[i+1,n]}}\,,\]
    
    \[\text{and}\qquad P_{i,s} = \frac{\ketbra{g_{s,i}}{g_{s,i}}}{\|\ket{g_{s,i}}\|_2^2+\delta/N}-m_\eta I\,,\]
    where $\ket{g_{s,i}}=\gamma_{s}\ket{0}+\gamma_{s+i}\ket{1}$.
        
    Fix $i\in[n]$ and for convenience reorder the $n$ tensor factors so that $i$ is last.
    Then $P_i$ is a direct sum of $2\times 2$ matrices with operator norm $\mc O_\delta(1)$,
    \[P_i=\sum_{s:s_i=0} \ketbra{s_{[n]\backslash i}}{s_{[n]\backslash i}}\otimes P_{i,s}\,,\]
    so $\|P_i\|_\op=\|H^{\otimes n}P_iH^{\otimes n}\|_\op=\mc O_\delta(1)$ and indeed $\|\wh{M}^{(\eta)}_{\bs \gamma}-m_\eta I\|_\op=\mc O_\delta(1).$
    
    To see that $H^{\otimes n}P_iH^{\otimes n}$ is entrywise small, note that, schematically
    \begin{align*}
    H^{\otimes n}P_iH^{\otimes n}&=\sum_{s:s_i=0} \Big(H^{\otimes n-1}\ketbra{s_{[n]\backslash i}}{s_{[n]\backslash i}}H^{\otimes n-1}\Big)\otimes \Big(HP_{i,s}H\Big)\\
    &=\frac{2}{N}\sum_{s:s_i=0} \Big(\pm1  \text{ sign matrix}\Big)\otimes\Big(HP_{i,s}H\Big).
    \end{align*}
    Each $2\times 2$ matrix in the family $\{P_{i,s}\}_s$ is independent, mean-zero, and bounded, so we see from this display that each entry of $H^{\otimes n}P_iH^{\otimes n}$ is an average of $N/2$-many $\mc O_\delta(1)$-bounded, centered random variables.

    By Bernstein's inequality, this means there is a universal constant $c'=c'(\delta)$ such that for every fixed pair $x,y$,
    \[
        \Pr\left[
            |(H^{\otimes n}P_iH^{\otimes n})_{xy}|>\frac{1}{n^2}
        \right]
        \le \exp\left(-c'\frac{N}{n^4}\right)\,.
    \]
    Taking a union
    bound over all $x,y\in\{0,1\}^n$ and all $i\in[n]$, the event
    \[
        E_M := \left\{\max_{i\in[n]}\max_{x,y}|(H^{\otimes n}P_iH^{\otimes n})_{xy}|
        \le
        \frac{1}{n^2}\right\}
    \]
    holds with probability $1-\exp(-N/\poly(n))$.
    When $E$ obtains, we get
    \[
        \max_{x,y}\left|\left(\wh{M}^{(\eta)}_{\bs \gamma}-m_\eta I\right)_{xy}\right|
        \le
        \frac1n\sum_i\max_{x,y}|(H^{\otimes n}P_iH^{\otimes n})_{xy}|
        \leq \frac{1}{n^2}
    \]
    for all sufficiently large $n$.
    
    The same argument applies to the centered dominator $D_{\bs \gamma}^{(\delta)}-d_\eta I$.
    Indeed,
    \[
        D_{\bs \gamma}^{(\delta)}-d_\eta I
        =
        \frac1n\sum_{i=1}^n Q_i,
    \]
    where \(Q_i\) is the physical-side conjugate of the Fourier-side direct sum over
    \(i\)-edges of the centered scalar block
    \[
        \left(
            \frac{\delta/N}{\|\ket{g_{s,i}}\|_2^2+\delta/N}-d_\eta
        \right)I_2.
    \]
    In each $Q_i$ these blocks are i.i.d., mean-zero, and bounded.
    Hence, on an event $E_D$ analogous to $E_M$, also holding with probability $\exp(-N/\poly(n))$,
    \[
        \left\|D_{\bs \gamma}^{(\delta)}-d_\eta I\right\|_{\op}=\mc O_\eta(1)
        \qquad\text{and}\qquad
        \max_{x,y}
        \left|
            \left(D_{\bs\gamma}^{(\delta)}-d_\eta I\right)_{xy}
        \right|
        \leq\frac{1}{n^2}
    \]
    for all sufficiently large $n$.

    To finish the argument, we apply
    \Cref{thm:tropp-random-paving} to the deterministic matrices
    $\wh{M}^{(\eta)}_{\bs \gamma}-m_\eta I$ and
    $D_{\bs \gamma}^{(\eta)}-d_\eta I$.
    The application is identical for each, so we explain the former.
    Let $E=E_M\cap E_D$ and let $C_\eta\leq \mc O_\eta(1)$ be such that $\|\wh M^{(\eta)}_{\bs \gamma}\|_\op\leq C_\eta$ pointwise.
    Let $\eps$ be as in the statement of \Cref{prop:BB-gaussian-paving} and put $\eps'=\eps/C_\eta$.
    Then according to \Cref{thm:tropp-random-paving} there is a density threshold $\beta_{\max}=\beta_{\max}(\eps'/2,1)$
    such that for all even $2\leq p\leq n$,
    \begin{align*}
        \E\left[\1_{|B|/N\leq \beta_{\max}}\cdot\big\|\Pi_B \big(\wh{M}^{(\eta)}_{\bs \gamma}-m_\eta I\big)\Pi_B\big\|^p_\op\right] &= \sum_{b=0}^{\lfloor\beta_{\max}N\rfloor}\E\left[\mathbf{1}_{|B| = b\wedge E}\cdot \big\|\Pi_B \big(\wh{M}^{(\eta)}_{\bs \gamma}-m_\eta I\big)\Pi_B\big\|^p_\op\right]\\
        &\hspace{3em} +\E\left[\mathbf{1}_{|B|/N\leq \beta_{\max}\wedge E^c}\cdot \big\|\Pi_B \big(\wh{M}^{(\eta)}_{\bs \gamma}-m_\eta I\big)\Pi_B\big\|^p_\op\right]\\
        &\leq (\eps/2)^p+\mc O_\eta(1)^p\cdot \exp(-N/\poly(n))\\
        &\leq \eps^p\,.
    \end{align*}
    for $n$ large enough.
    We conclude that for all even $2\leq p\leq n$,
    \begin{align*}\E\Big[\1_{|B|/N\leq \beta_{\max}}&\cdot\tr\big\lvert\Pi_B \big(\wh{M}^{(\eta)}_{\bs \gamma}-m_\eta I\big)\Pi_B\big\rvert^p\Big]^{1/p}\\
    &\leq N^{1/p}\E\left[\1_{|B|/N\leq \beta_{\max}}\big\|\Pi_B \big(\wh{M}^{(\eta)}_{\bs \gamma}-m_\eta I\big)\Pi_B\big\|^p_\op\right]^{1/p}
    \leq N^{1/p}\eps\,,
    \end{align*}
    as desired.
\end{proof}

\paragraph{Step 4: Combine bounds.}
It remains to combine \Cref{prop:BB-gaussianization} and \Cref{prop:BB-gaussian-paving} to prove the sufficient condition (\Cref{prop:BB-sufficient}) for the $BB$ bound, \Cref{ineq:BB}.

\begin{proof}[Proof of \Cref{prop:BB-sufficient}]
Fix $\eps>0$ and choose $\eta>0$ sufficiently small so that
\[
    d_\eta\le\frac{\eps}{2}.
\]
Recall from \eqref{eq:m-and-d_eta-defn} that $m_\eta\leq m_\eta+d_\eta/2=1/2$.

Let $C_\eta$ be the constant from \Cref{prop:BB-gaussianization}. Choose a
constant $c=c(\eta,\eps)\in(0,1)$ sufficiently small that
\[
    2C_\eta
    \exp\left(-\frac{3\log 2}{2c}\right)
    \le
    \frac{\eps}{8}
    \exp\left(-\frac{5\log 2}{4c}\right).
\]
Apply \Cref{prop:BB-gaussian-paving} with regularization parameter $\eta$ and
paving target
\[
    \frac{\eps}{8}
    \exp\left(-\frac{5\log 2}{4c}\right),
\]
and let $\beta_{\max}>0$ be the resulting density threshold.

By \Cref{prop:bad-set-regularity}, there are deterministic quantities
$r_T\to0$ as $T\to\infty$ such that, with probability
$1-\exp(-\Omega(n))$,
\[
    \|\mu\|_{3/2}^{3/2}\le N^{-1/2+o(1)},
    \qquad
    \Bias(\mu)\le N^{-1/2+o(1)},
\]
and
\[
    \frac{|B|}{N}+\max_y\mu(B+y)\le r_T.
\]
Choose $T$ sufficiently large that
\[
    r_T\le\beta_{\max}
\]
and
\[
    C_\eta\sqrt{r_T}
    \le
    \frac{\eps}{8}
    \exp\left(-\frac{5\log 2}{4c}\right).
\]
Let $E$ denote the event on which these profile bounds hold.

For each $n$, let $p$ be the largest even integer satisfying $p\le cn$.
For all sufficiently large $n$, we have $8\le p\le n$. Moreover, on event
$E$,
\[
    \Bias(\mu)^{2/p}
    \le
    \exp\left(-\frac{\log 2-o(1)}{c}\right)
    \le\frac12.
\]
Also,
\[
    \|\mu\|_{3/2}^{1/2}
    =
    \left(\|\mu\|_{3/2}^{3/2}\right)^{1/3}
    \le
    N^{-1/6+o(1)}.
\]
Thus, after increasing $n$ if necessary,
\[
    C_\eta p^3\|\mu\|_{3/2}^{1/2}
    \le
    \frac{\eps}{8}
    \exp\left(-\frac{5\log 2}{4c}\right)
\]
and
\[
    C_\eta\Bias(\mu)^{3/p}
    \le
    \frac{\eps}{8}
    \exp\left(-\frac{5\log 2}{4c}\right).
\]

The following argument applies identically to the centered regularized
$M$-operator and to the centered dominator. Write $W_{\wh f}$ and
$W_{\bs\gamma}$ for either of the pairs
\[
\begin{aligned}
W_{\wh f}
&=
\Pi_B\left(\wh M_{\wh f}^{(\eta)}-m_\eta I\right)\Pi_B,
&
W_{\bs\gamma}
&=
\Pi_B\left(\wh M_{\bs\gamma}^{(\eta)}-m_\eta I\right)\Pi_B,
\end{aligned}
\]
or
\[
\begin{aligned}
W_{\wh f}
&=
\Pi_B\left(D_{\wh f}^{(\eta)}-d_\eta I\right)\Pi_B,
&
W_{\bs\gamma}
&=
\Pi_B\left(D_{\bs\gamma}^{(\eta)}-d_\eta I\right)\Pi_B.
\end{aligned}
\]

Conditional on $\mu\in E$, \Cref{prop:BB-gaussianization} gives
\[
\begin{aligned}
\left(
    \E_\omega\tr|W_{\wh f}|^p
\right)^{1/p}
&\le
\left(
    \E_{\bs\gamma}\tr|W_{\bs\gamma}|^p
\right)^{1/p}\\
&\quad+
\frac{3\eps}{8}N^{1/p}
\exp\left(-\frac{5\log 2}{4c}\right).
\end{aligned}
\]
Taking the $L_p$-norm over $\mu$ on $E$ and applying Minkowski's
inequality,
\[
\begin{aligned}
\left(
    \E_\mu\mathbf 1_{E}
    \E_\omega\tr|W_{\wh f}|^p
\right)^{1/p}
&\le
\left(
    \E_\mu\mathbf 1_{E}
    \E_{\bs\gamma}\tr|W_{\bs\gamma}|^p
\right)^{1/p}\\
&\quad+
\frac{3\eps}{8}N^{1/p}
\exp\left(-\frac{5\log 2}{4c}\right).
\end{aligned}
\]
Since $E$ implies $|B|/N\le\beta_{\max}$,
\Cref{prop:BB-gaussian-paving} gives
\[
    \left(
        \E_\mu\mathbf 1_{E}
        \E_{\bs\gamma}\tr|W_{\bs\gamma}|^p
    \right)^{1/p}
    \le
    \frac{\eps}{8}N^{1/p}
    \exp\left(-\frac{5\log 2}{4c}\right).
\]
Consequently,
\[
    \left(
        \E_\mu\mathbf 1_{E}
        \E_\omega\tr|W_{\wh f}|^p
    \right)^{1/p}
    \le
    \frac{\eps}{2}N^{1/p}
    \exp\left(-\frac{5\log 2}{4c}\right).
\]

Markov's inequality therefore gives
\[
\begin{aligned}
\Pr\left[
    E
    \ \wedge\
    \|W_{\wh f}\|_\op\ge\frac{\eps}{2}
\right]
&\le
\left(\frac{2}{\eps}\right)^p
\E_\mu\mathbf 1_{E}
\E_\omega\tr|W_{\wh f}|^p\\
&\le
N\exp\left(-\frac{5\log 2}{4c}p\right)\\
&=
\exp\left(-\frac{\log 2}{4}n+O_c(1)\right).
\end{aligned}
\]
Since $\Pr[E^c]\leq\exp(-\Omega(n))$, applying the estimate to both
centered objects and taking a union bound shows that, with probability
$1-\exp(-\Omega(n))$,
\[
    \left\|
        \Pi_B\left(
            \wh M_{\wh f}^{(\eta)}-m_\eta I
        \right)\Pi_B
    \right\|_\op
    \le\frac{\eps}{2}\qquad\text{and}\qquad
    \left\|
        \Pi_B\left(
            D_{\wh f}^{(\eta)}-d_\eta I
        \right)\Pi_B
    \right\|_\op
    \le\frac{\eps}{2}.
\]

Therefore,
\[
\begin{aligned}
\left\|\Pi_B\wh M_{\wh f}^{(\eta)}\Pi_B\right\|_\op
&\le
m_\eta+\frac{\eps}{2}\\
&\le
\frac12+\eps,
\end{aligned}
\]
and
\[
\begin{aligned}
\left\|\Pi_BD_{\wh f}^{(\eta)}\Pi_B\right\|_\op
&\le
d_\eta+\frac{\eps}{2}\\
&\le
\eps.
\end{aligned}
\]
This proves the proposition.
\end{proof}

\subsection{Proof of  \Cref{ineq:AB}, the $AB$, $BA$ bounds}
\label{sec:AB}
The proof of \Cref{ineq:AB} goes much the same way as the $BB$ bound, except we analyze the Hermitian square.
We describe it only briefly.
For notational convenience define
\[
    R_{\wh f}^{(\eta)}
    :=
    \wh M_{\wh f}^{(\eta)}-m_\eta I.
\]
Since $A\cap B=\varnothing$,
\[
    \Pi_A\wh M_{\wh f}^{(\eta)}\Pi_B
    =
    \Pi_AR_{\wh f}^{(\eta)}\Pi_B.
\]
Consequently,
\begin{align}
\label{ineq:AB-square-reduction}
\left\|
    \Pi_A\wh M_{\wh f}^{(\eta)}\Pi_B
\right\|_\op^2
&=
\left\|
    \Pi_BR_{\wh f}^{(\eta)}
    \Pi_A
    R_{\wh f}^{(\eta)}\Pi_B
\right\|_\op
\le
\left\|
    \Pi_B\left(R_{\wh f}^{(\eta)}\right)^2\Pi_B
\right\|_\op.
\end{align}

Define also
\[
    R_{\bs\gamma}^{(\eta)}
    :=
    \wh M_{\bs\gamma}^{(\eta)}-m_\eta I.
\]

\begin{lem}[Gaussian square mean]
\label{lem:AB-gaussian-square-mean}
There is a scalar $b_{n,\eta}$ such that
\[
    \E_{\bs\gamma}
    \left(R_{\bs\gamma}^{(\eta)}\right)^2
    =
    b_{n,\eta}I.
\]
Moreover,
\[
    \lim_{\eta\downarrow0}
    \lim_{n\to\infty}
    b_{n,\eta}
    =
    b,
\]
where
\[
    b
    =
    \frac{\pi^2}{3}-\frac{13}{4}
    <
    \frac1{16}.
\]
\end{lem}

\begin{proof}
Write
\[
    R_{\bs\gamma}^{(\eta)}
    =
    H^{\otimes n}
    \left(
        \frac1n\sum_{i=1}^nP_i^{(\eta)}
    \right)
    H^{\otimes n},
\]
where
\[
    P_i^{(\eta)}
    =
    \bigoplus_{s:s_i=0}P_{i,s}^{(\eta)}
\]
and
\[
    P_{i,s}^{(\eta)}
    :=
    \frac{\ketbra{g_{s,i}}{g_{s,i}}}
    {\|g_{s,i}\|_2^2+\eta/N}
    -
    m_\eta I_2,
    \qquad
    \ket{g_{s,i}}
    :=
    \gamma_s\ket0+\gamma_{s+e_i}\ket1.
\]

The law of $\bs\gamma$ is invariant under translations of the cube and under
independent phase rotations of its coordinates.
It follows that, for some scalars $a_\eta$ and $b_\eta$,
\[
    \E\left(P_i^{(\eta)}\right)^2
    =
    a_\eta I
\]
and, whenever $i\neq j$,
\[
    \E P_i^{(\eta)}P_j^{(\eta)}
    =
    b_\eta I.
\]
Moreover, $\|P_i^{(\eta)}\|_\op\le1$, so $0\le a_\eta\le1$.

To identify $b_\eta$, fix distinct $i,j$ and a vertex $s$.
The $i$-edge and $j$-edge through $s$ meet only at $s$, and hence the
$s$-th diagonal entry of $P_i^{(\eta)}P_j^{(\eta)}$ is the product of the
corresponding diagonal entries.
Thus, if $X,Y,Z$ are independent $\operatorname{Exp}(1)$ random variables,
\[
    b_\eta
    =
    \E
    \left[
        \left(
            \frac{X}{X+Y+\eta}-m_\eta
        \right)
        \left(
            \frac{X}{X+Z+\eta}-m_\eta
        \right)
    \right],
\]
where
\[
    m_\eta
    =
    \E\frac{X}{X+Y+\eta}.
\]

Consequently,
\begin{align*}
\E_{\bs\gamma}
\left(R_{\bs\gamma}^{(\eta)}\right)^2
&=
H^{\otimes n}
\E\left[
    \left(
        \frac1n\sum_iP_i^{(\eta)}
    \right)^2
\right]
H^{\otimes n}=
\left(
    \frac{a_\eta}{n}
    +
    \frac{n-1}{n}b_\eta
\right)I.
\end{align*}
We may therefore take
\[
    b_{n,\eta}
    :=
    \frac{a_\eta}{n}
    +
    \frac{n-1}{n}b_\eta.
\]

Since $a_\eta$ is bounded,
\[
    \lim_{n\to\infty}b_{n,\eta}
    =
    b_\eta.
\]
Also, by dominated convergence,
\[
    m_\eta\longrightarrow\frac12
\]
and
\[
    b_\eta
    \longrightarrow
    \E
    \left[
        \left(
            \frac{X}{X+Y}-\frac12
        \right)
        \left(
            \frac{X}{X+Z}-\frac12
        \right)
    \right]
    =
    \E\frac{X^2}{(X+Y)(X+Z)}
    -
    \frac14.
\]

The normalized vector
\[
    \frac{(X,Y,Z)}{X+Y+Z}
\]
is uniform on the two-dimensional simplex.
Therefore,
\begin{align*}
\E\frac{X^2}{(X+Y)(X+Z)}
&=
2\int_0^1
\int_0^{1-x}
\frac{x^2}{(x+y)(1-y)}
\,dy\,dx\\
&=
-4\int_0^1
\frac{x^2\log x}{1+x}
\,dx\\
&=
4\sum_{k=0}^\infty
\frac{(-1)^k}{(k+3)^2}\\
&=
\frac{\pi^2}{3}-3.
\end{align*}
It follows that
\[
    b
    =
    \frac{\pi^2}{3}-3-\frac14
    =
    \frac{\pi^2}{3}-\frac{13}{4}
    <
    \frac1{16}.\qedhere
\]
\end{proof}

We claim that, for every fixed $\eps>0$, after choosing $T$ sufficiently
large, with probability $1-\exp(-\Omega(n))$,
\begin{equation}
\label{ineq:regularized-square-concentration}
    \left\|
        \Pi_B
        \left(
            \left(R_{\wh f}^{(\eta)}\right)^2-b_{n,\eta}I
        \right)
        \Pi_B
    \right\|_\op
    \le
    \eps.
\end{equation}
The proof follows the same Gaussian-transfer and random-paving argument
as in the $BB$ case.
To make the extension explicit, write $P_i(z)$
for the centered Fourier-side edge operators from the preceding lemma,
with general Fourier input $z$, and set
\[
K_{ij}(z):=P_i(z)P_j(z)-c_{ij}I,
\qquad
c_{ij}:=
\begin{cases}
a_\eta,&i=j,\\
b_\eta,&i\ne j.
\end{cases}
\]
Then
\[
(R_z^{(\eta)})^2-b_{n,\eta}I
=
H^{\otimes n}
\left(\frac1{n^2}\sum_{i,j}K_{ij}(z)\right)
H^{\otimes n}.
\]
For fixed $i,j$, $K_{ij}$ is a direct sum over cosets of
$\operatorname{span}_{\mathbb F_2}\{e_i,e_j\}$, with blocks of size
at most four, each depending only on the coordinates in its coset.

The product rule and \Cref{lem:master-frechet} give the same derivative
bounds, up to constants depending on $\eta$.
Thus the physical-coordinate
replacement argument is unchanged.
For covariance interpolation, expand
the derivatives into matrix units within these blocks.
The earlier covariance-interpolation argument applies with ordered pairs of directions replacing single directions, and shifts in $\operatorname{span}\{e_i,e_j\}$ replacing shifts in $\{0,e_i\}$.
The averaging factors cancel the additional direction sums, and the endpoint Schur-multiplier estimate already holds for arbitrary shifts.
Consequently, the Gaussian comparison has the same error bound.

Finally, in the independent Gaussian model, the blocks of each fixed $K_{ij}$ are independent, centered, and uniformly bounded.
After Hadamard conjugation, Bernstein's inequality and a union bound over $i,j$ and matrix entries give the same entrywise estimate required for random paving.
Thus the proofs of \Cref{prop:BB-gaussianization} and \Cref{prop:BB-gaussian-paving} apply with only constant-factor changes and yield \eqref{ineq:regularized-square-concentration}.

Combining \eqref{ineq:AB-square-reduction} and
\eqref{ineq:regularized-square-concentration}, we obtain
\[
    \left\|
        \Pi_A\wh M_{\wh f}^{(\eta)}\Pi_B
    \right\|_\op
    \le
    \sqrt{b_{n,\eta}+\eps}.
\]
Choosing $\eta$ sufficiently small, then $n$ sufficiently large, and taking
the tolerance in \eqref{ineq:regularized-square-concentration} sufficiently
small gives
\[
    \left\|
        \Pi_A\wh M_{\wh f}^{(\eta)}\Pi_B
    \right\|_\op
    \le
    \frac14+\frac{\eps}{2}.
\]

It remains only to remove the regularization.
Put
\[
    E_{\wh f}^{(\eta)}
    :=
    \wh M_{\wh f}-\wh M_{\wh f}^{(\eta)}.
\]
Since
\[
    0\preceq E_{\wh f}^{(\eta)}
    \preceq D_{\wh f}^{(\eta)}
    \qquad\text{and}\qquad
    E_{\wh f}^{(\eta)}\preceq I,
\]
we have
\begin{align*}
\left\|
    \Pi_AE_{\wh f}^{(\eta)}\Pi_B
\right\|_\op^2
&=
\left\|
    \Pi_BE_{\wh f}^{(\eta)}
    \Pi_A
    E_{\wh f}^{(\eta)}\Pi_B
\right\|_\op\\
&\le
\left\|
    \Pi_B\left(E_{\wh f}^{(\eta)}\right)^2\Pi_B
\right\|_\op\\
&\le
\left\|
    \Pi_BD_{\wh f}^{(\eta)}\Pi_B
\right\|_\op.
\end{align*}
Applying \Cref{prop:BB-sufficient} with a sufficiently small error parameter
therefore gives
\[
    \left\|
        \Pi_AE_{\wh f}^{(\eta)}\Pi_B
    \right\|_\op
    \le
    \frac{\eps}{2}
\]
with probability $1-\exp(-\Omega(n))$.
Hence
\[
    \left\|
        \Pi_A\wh M_{\wh f}\Pi_B
    \right\|_\op
    \le
    \frac14+\eps,
\]
which proves \Cref{ineq:AB}.

\section{Finishing the proof}
\label{sec:finishing}

\begin{proof}[Proof of \Cref{thm:full-uncertainty}]
Fix $\eps<1/20$ and choose $T>0$ large enough to satisfy the conditions in
\Cref{thm:B-op-norm-bounds}; \textit{i.e.,} so that the errors in the $AB$-
and $BB$-bounds are at most $\eps$.  Then choose $\kappa>0$ small enough that
\[
    2\sqrt{\frac{CT^5\kappa}{\eps^3}}<\frac13
    \qquad\text{and}\qquad
    \kappa<\frac14-5\eps,
\]
where $C$ is the constant in \Cref{thm:metric-entropy}.

If $\Avg_f[g]\le 1-\kappa$, we are done, so assume
$\Avg_f[g]>1-\kappa$, \textit{i.e.},
\[
    \Inf_\mu[q]<\kappa n.
\]

It will be convenient to work with a shifted $q$ (equiv. $g$) so that its $A$-part is centered with respect to $\mu$.
To that end, let $A=B^c$, and define
\[
    \alpha
    :=
    \frac{\langle \1,q\1_A\rangle_\mu}
    {\langle \1,\1_A\rangle_\mu}.
\]
Then define the shifted versions of $g$ and $q$ by:
\begin{align*}
    &g^\perp:=g-\alpha f,
    \qquad
    g_A^\perp:=g^\perp\1_A,
    \qquad
    g_B^\perp:=g^\perp\1_B\\
    \text{and}\qquad &q^\perp:=q-\alpha \1,
    \qquad
    q_A^\perp:=q^\perp\1_A,
    \qquad
    q_B^\perp:=q^\perp\1_B
\end{align*}
Note that since $g=fq$, $\langle\1,q\rangle_\mu=0$, and
$\|q\|_{L^2(\mu)}=1$, we have
\[
    \langle f,g^\perp\rangle
    =
    \langle\1,q^\perp\rangle_\mu
    =
    -\alpha
\qquad\text{and}\qquad
    \|g^\perp\|_2^2
    =
    \|q^\perp\|_{L^2(\mu)}^2
    =
    1+|\alpha|^2.
\]
Moreover, the definition of $\alpha$ gives
\[
    \langle f,g_A^\perp\rangle
    =
    \E_\mu q_A^\perp
    =
    0.
\]

Since $f$ is a $+1$ eigenvector of $\wh M_{\wh f}$,
\begin{align*}
\Avg_{\wh f}[\wh g]
&=
\langle g,\wh M_{\wh f}g\rangle\\
&=
\left\langle
    g^\perp+\alpha f,
    \wh M_{\wh f}(g^\perp+\alpha f)
\right\rangle\\
&=
\langle g^\perp,\wh M_{\wh f}g^\perp\rangle
+
2\Re\left(
    \overline\alpha\langle f,g^\perp\rangle
\right)
+
|\alpha|^2\\
&=
\langle g^\perp,\wh M_{\wh f}g^\perp\rangle
-
|\alpha|^2.
\end{align*}
Decomposing $g^\perp$, we get
\begin{equation}
    \label{eq:gperp-decomp}
\langle g^\perp,\wh M_{\wh f}g^\perp\rangle
=
\langle g_A^\perp,\wh M_{\wh f}g_A^\perp\rangle
+
2\Re\langle g_A^\perp,\wh M_{\wh f}g_B^\perp\rangle +
\langle g_B^\perp,\wh M_{\wh f}g_B^\perp\rangle.
\end{equation}
We bound these three terms separately.
Note that our choice of $\alpha$ ensures that
$q_A^\perp$ is centered with respect to $\mu$, as required for the
$AA$-part argument, while one appreciates that the $AB$- and $BB$-operator bounds of \Cref{thm:B-op-norm-bounds} require no
centeredness.

Indeed, by \Cref{thm:B-op-norm-bounds} and our choice of $T$, with probability
$1-\exp(-\Omega(n))$ we have
\[
\left|
    \langle g_A^\perp,\wh M_{\wh f}g_B^\perp\rangle
\right|
\le
\left(
    \frac14+\eps
\right)
\|g_A^\perp\|_2\|g_B^\perp\|_2
\qquad\text{and}\qquad
\langle g_B^\perp,\wh M_{\wh f}g_B^\perp\rangle
\le
\left(
    \frac12+\eps
\right)
\|g_B^\perp\|_2^2.
\]

\paragraph{$AA$-part.}
For the $AA$ part, $\langle g_A^\perp,\wh M_{\wh f}g_A^\perp\rangle$, we pass to the $q$ picture.
Put
\[
    \wt M
    :=
    \operatorname{Diag}(1/f)\,
    \wh M_{\wh f}\,
    \operatorname{Diag}(f).
\]
This is well-defined with probability one.
The map $u\mapsto fu$ is an
isometry from $L^2(\mu)$ to $\ell_2$, and
\[
    \langle u,\wt Mv\rangle_\mu
    =
    \langle fu,\wh M_{\wh f}fv\rangle.
\]
Consequently, $\wt M$ is self-adjoint and contractive on $L^2(\mu)$, and
\[
    \langle g_A^\perp,\wh M_{\wh f}g_A^\perp\rangle
    =
    \langle q_A^\perp,\wt Mq_A^\perp\rangle_\mu.
\]

On the internal edges of $A$, the functions $q_A^\perp$ and
$q-\alpha\1$ agree.  Therefore
\begin{align*}
\Inf_\mu^{\Int(A)}[q_A^\perp]
&=
\Inf_\mu^{\Int(A)}[q-\alpha\1]\\
&\le
\Inf_\mu[q-\alpha\1]\\
&=
\Inf_\mu[q]\\
&\le
\kappa n.
\end{align*}

We claim that, with high probability,
\begin{equation}
\label{eq:A-part-bound}
\langle g_A^\perp,\wh M_{\wh f}g_A^\perp\rangle=\langle q_A^\perp,\wt Mq_A^\perp\rangle_\mu
\;\leq\;
\left(
    \frac12+2\eps+o(1)
\right)
\|q_A^\perp\|_{L^2(\mu)}^2
+
\eps.
\end{equation}
This follows by considering two cases.

\medskip
\noindent\textit{Case I: Small norm.}
If $\|q_A^\perp\|_{L^2(\mu)}^2\le\eps$,
then contractivity of $\wt M$ gives
\[
    \langle q_A^\perp,\wt Mq_A^\perp\rangle_\mu
    \le
    \|q_A^\perp\|_{L^2(\mu)}^2
    \le
    \eps,
\]
so \eqref{eq:A-part-bound} holds.

\medskip
\noindent\textit{Case II: Large norm.}
Suppose instead that $\|q_A^\perp\|_{L^2(\mu)}^2>\eps$.
Then
\[
    \tilde{q}^\perp_A:=\frac{q_A^\perp}{\|q_A^\perp\|_{L^2(\mu)}}
\]
satisfies
\[\E_\mu \tilde{q}^\perp_A=0,\qquad\|\tilde{q}^\perp_A\|_{L^2(\mu)}=1,\qquad\text{and}\qquad
    \Inf_\mu^{\Int(A)}
    \left[
        \tilde{q}^\perp_A
    \right]
    \le
    \frac{\kappa}{\eps}n.
\]



On the high-probability event over $\mu$ on which
\Cref{cor:flat-A-net} and \Cref{thm:avg-hat-fixed-q} both hold, let
\[
    \mathcal Q_{\kappa/\eps,A}^{(\eps/4)}
\]
be the compatible flat net supplied by \Cref{cor:flat-A-net}.
Every $
    v\in\mathcal Q_{\kappa/\eps,A}^{(\eps/4)}$
is $\mu$-compatible and satisfies
\[
    \max_x\mu(x)|v(x)|^2
    \leq
    \|\mu\|_\infty^{1/3}.
\]
Moreover, by the definition of $\wt M$, $
    \langle v,\wt Mv\rangle_\mu
    =
    \Avg_{\wh f}[\widehat{fv}]$.
We may therefore apply \Cref{thm:avg-hat-fixed-q} to every point of the net.
Taking the deviation parameter there to be $\eps/2$ and applying a union
bound gives
\begin{align*}
&
\Pr_{f\sim F_\mu}\left[
    \exists
    v\in\mathcal Q_{\kappa/\eps,A}^{(\eps/4)}
    \text{ such that }
    \langle v,\wt Mv\rangle_\mu
    >
    \frac12+\frac{\eps}{2}+o(1)
\right]
\\
&\qquad\leq
\exp\left(
    N^{
        H\left(
            CT^5\kappa/\eps^3
        \right)
        +o(1)
    }
    \log\left(
        \frac{C}{\eps}
    \right)
    -
    c\eps^6N^{1/3-o(1)}
\right).
\end{align*}
Here and below, the universal constants absorb the fixed factors arising from
using $\kappa/\eps$ as the energy threshold and $\eps/4$ as the net radius.
By our choice of $\kappa$,
\[
    H\left(
        \frac{CT^5\kappa}{\eps^3}
    \right)
    <
    \frac13.
\]
Consequently, the preceding probability is
\[
    \exp\left(
        -N^{1/3-o(1)}
    \right),
\]
and in particular is at most $\exp(-\Omega(n))$.

Condition on the complementary event.
Since
\[
    \tilde q_A^\perp
    \in
    \mathcal Q_{\kappa/\eps,A},
\]
we may choose
\[
    v\in
    \mathcal Q_{\kappa/\eps,A}^{(\eps/4)}
\]
such that
\[
    \left\|
        \tilde q_A^\perp-v
    \right\|_{L^2(\mu)}
    \leq
    \frac{\eps}{4}.
\]
Since $\tilde q_A^\perp$ and $v$ both have unit $L^2(\mu)$ norm and
$\|\wt M\|_{\op,L^2(\mu)}\leq1$,
\begin{align*}
\left\langle
    \tilde q_A^\perp,
    \wt M\tilde q_A^\perp
\right\rangle_\mu
&\leq
\langle v,\wt Mv\rangle_\mu
+
\left|
    \left\langle
        \tilde q_A^\perp-v,
        \wt M\tilde q_A^\perp
    \right\rangle_\mu
\right|
+
\left|
    \left\langle
        v,
        \wt M\left(
            \tilde q_A^\perp-v
        \right)
    \right\rangle_\mu
\right|
\\
&\leq
\frac12+\frac{\eps}{2}+o(1)
+
2
\left\|
    \tilde q_A^\perp-v
\right\|_{L^2(\mu)}
\\
&\leq
\frac12+\eps+o(1).
\end{align*}
Multiplying by
\[
    \|q_A^\perp\|_{L^2(\mu)}^2
\]
gives
\[
\langle q_A^\perp,\wt Mq_A^\perp\rangle_\mu
\leq
\left(
    \frac12+\eps+o(1)
\right)
\|q_A^\perp\|_{L^2(\mu)}^2,
\]
which is stronger than \eqref{eq:A-part-bound}.


\paragraph{Combining the bounds.}
Intersecting the good events for the $AA$, $AB$, and $BB$ bounds and substituting into \eqref{eq:gperp-decomp}, we obtain, with probability
$1-\exp(-\Omega(n))$,
\begin{align*}
\langle g^\perp,\wh M_{\wh f}g^\perp\rangle
&\le
\begin{pmatrix}
\|g_A^\perp\|_2 &
\|g_B^\perp\|_2
\end{pmatrix}
\begin{pmatrix}
\frac12+2\eps+o(1) & \frac14+\eps\\
\frac14+\eps & \frac12+\eps
\end{pmatrix}
\begin{pmatrix}
\|g_A^\perp\|_2\\
\|g_B^\perp\|_2
\end{pmatrix}
+
\eps\\
&\le
\left(
    \frac34+3\eps+o(1)
\right)
\left(
    \|g_A^\perp\|_2^2+\|g_B^\perp\|_2^2
\right)
+
\eps\\
&=
\left(
    \frac34+3\eps+o(1)
\right)
\left(
    1+|\alpha|^2
\right)
+
\eps.
\end{align*}
Here the maximum eigenvalue of the displayed $2\times2$ matrix is at most its
largest row sum, namely $\frac34+3\eps+o(1)$.

Returning to the earlier identity,
\begin{align*}
\Avg_{\wh f}[\wh g]
&=
\langle g^\perp,\wh M_{\wh f}g^\perp\rangle
-
|\alpha|^2\\
&\le
\frac34+4\eps+o(1)
-
\left(
    \frac14-3\eps-o(1)
\right)
|\alpha|^2\\
&\le
\frac34+5\eps
\end{align*}
for all sufficiently large $n$.
By our choice of $\kappa$, this means $\Avg_{\wh f}[\wh g]< 1-\kappa$.

We have therefore shown that, with probability $1-\exp(-\Omega(n))$ over $f$, for every
admissible $g$,
\[
    \min\!\big\{\Avg_f[g],
    \Avg_{\wh f}[\wh g]\big\}\le1-\kappa.
\]
This proves \Cref{thm:full-uncertainty}.
\end{proof}

\newpage

\bibliographystyle{alpha}
\bibliography{references}

@article{li2025universal,
  title={Universal and Efficient Quantum State Verification via Schmidt Decomposition and Mutually Unbiased Bases},
  author={Li, Yunting and Zhu, Huangjun},
  journal={arXiv preprint arXiv:2506.19809},
  year={2025}
}

@inproceedings{CLSW26,
  title={The Power of Two Bases: Robust and copy-optimal certification of nearly all quantum states with few-qubit measurements},
  author={Coladangelo, Andrea and Li, Jerry and Slote, Joseph and Wu, Ellen},
  booktitle={Proceedings of the 58th Annual ACM Symposium on Theory of Computing},
  pages={1771--1776},
  year={2026}
}

@book{LevinPeres2017,
  author    = {Levin, David A. and Peres, Yuval},
  title     = {Markov Chains and Mixing Times},
  edition   = {Second},
  publisher = {American Mathematical Society},
  address   = {Providence, RI},
  year      = {2017},
  month     = oct,
  isbn      = {9781470442323},
  doi       = {10.1090/mbk/107},
  note      = {With contributions by Elizabeth L. Wilmer}
}

@INPROCEEDINGS{7148912,
  author={Benedetto, John J. and Koprowski, Paul J.},
  booktitle={2015 International Conference on Sampling Theory and Applications (SampTA)}, 
  title={Graph theoretic uncertainty principles}, 
  year={2015},
  volume={},
  number={},
  pages={357-361},
  doi={10.1109/SAMPTA.2015.7148912}}

@book{ODonnell2014,
  author    = {Ryan O'Donnell},
  title     = {Analysis of Boolean Functions},
  publisher = {Cambridge University Press},
  year      = {2014},
  doi       = {10.1017/CBO9781139814782},
  isbn      = {9781107038325}
}

@article {MR2379999,
    AUTHOR = {Tropp, Joel A.},
     TITLE = {The random paving property for uniformly bounded matrices},
   JOURNAL = {Studia Math.},
  FJOURNAL = {Studia Mathematica},
    VOLUME = {185},
      YEAR = {2008},
    NUMBER = {1},
     PAGES = {67--82},
      ISSN = {0039-3223,1730-6337},
   MRCLASS = {46B09 (15A52 46B20 60E15)},
  MRNUMBER = {2379999},
MRREVIEWER = {Sasha\ Sodin},
       DOI = {10.4064/sm185-1-4},
       URL = {https://doi-org.offcampus.lib.washington.edu/10.4064/sm185-1-4},
}

@article{WEISSLER1980218,
title = {Logarithmic Sobolev inequalities and hypercontractive estimates on the circle},
journal = {Journal of Functional Analysis},
volume = {37},
number = {2},
pages = {218-234},
year = {1980},
issn = {0022-1236},
doi = {https://doi.org/10.1016/0022-1236(80)90042-7},
url = {https://www.sciencedirect.com/science/article/pii/0022123680900427},
author = {Fred B Weissler}
}

@misc{rapaport2025negativemomentssteinhaussums,
      title={Negative Moments of Steinhaus Sums}, 
      author={Martin Rapaport and Tomasz Tkocz and Isabella Wu},
      year={2025},
      eprint={2512.09077},
      archivePrefix={arXiv},
      primaryClass={math.PR},
      url={https://arxiv.org/abs/2512.09077}, 
}

@inproceedings{guptaHeODonnell2025,
author = {Gupta, Meghal and He, William and O'Donnell, Ryan},
title = {Few Single-Qubit Measurements Suffice to Certify Any Quantum State},
year = {2026},
isbn = {9798400725364},
publisher = {Association for Computing Machinery},
address = {New York, NY, USA},
url = {https://doi.org/10.1145/3798129.3800726},
doi = {10.1145/3798129.3800726},
booktitle = {Proceedings of the 58th Annual ACM Symposium on Theory of Computing},
pages = {54–60},
numpages = {7},
location = {Salt Lake City, UT, USA},
series = {STOC '26}
}

@misc{sater2025efficientcertificationintractablequantum,
      title={Efficient certification of intractable quantum states with few Pauli measurements}, 
      author={Sami Abdul Sater and Maxime Garnier and Thierry Martinez and Harold Ollivier and Ulysse Chabaud},
      year={2025},
      eprint={2511.07300},
      archivePrefix={arXiv},
      primaryClass={quant-ph},
      url={https://arxiv.org/abs/2511.07300}, 
}

@article{Huang2025,
author={Huang, Hsin-Yuan
and Preskill, John
and Soleimanifar, Mehdi},
title={Certifying almost all quantum states with few single-qubit measurements},
journal={Nature Physics},
year={2025},
month={Sep},
day={08},
issn={1745-2481},
doi={10.1038/s41567-025-03025-1},
url={https://doi.org/10.1038/s41567-025-03025-1}
}

@article{KlieschRoth,
  title = {Theory of Quantum System Certification},
  author = {Kliesch, Martin and Roth, Ingo},
  journal = {PRX Quantum},
  volume = {2},
  issue = {1},
  pages = {010201},
  numpages = {53},
  year = {2021},
  month = {Jan},
  publisher = {American Physical Society},
  doi = {10.1103/PRXQuantum.2.010201},
  url = {https://link.aps.org/doi/10.1103/PRXQuantum.2.010201}
}

@article{ZhuHayashi,
  title = {Optimal verification and fidelity estimation of maximally entangled states},
  author = {Zhu, Huangjun and Hayashi, Masahito},
  journal = {Phys. Rev. A},
  volume = {99},
  issue = {5},
  pages = {052346},
  numpages = {10},
  year = {2019},
  month = {May},
  publisher = {American Physical Society},
  doi = {10.1103/PhysRevA.99.052346},
  url = {https://link.aps.org/doi/10.1103/PhysRevA.99.052346}
}

@inproceedings{Badescu2019,
  author       = {Costin Badescu and
                  Ryan O'Donnell and
                  John Wright},
  editor       = {Moses Charikar and
                  Edith Cohen},
  title        = {Quantum state certification},
  booktitle    = {Proceedings of the 51st Annual {ACM} {SIGACT} Symposium on Theory
                  of Computing, {STOC} 2019, Phoenix, AZ, USA, June 23-26, 2019},
  pages        = {503--514},
  publisher    = {{ACM}},
  year         = {2019},
  url          = {https://doi.org/10.1145/3313276.3316344},
  doi          = {10.1145/3313276.3316344},
  bibsource    = {dblp computer science bibliography, https://dblp.org}
}

@inbook{ledoux,
  title = {Concentration of measure and logarithmic Sobolev inequalities},
  ISBN = {9783540484073},
  ISSN = {1617-9692},
  url = {http://dx.doi.org/10.1007/BFb0096511},
  DOI = {10.1007/bfb0096511},
  booktitle = {Séminaire de Probabilités XXXIII},
  publisher = {Springer Berlin Heidelberg},
  author = {Ledoux,  Michel},
  year = {1999},
  pages = {120–216}
}

@article{du2025certifying,
  title={Certifying localizable quantum properties with constant sample complexity},
  author={Du, Zhenyu and Liu, Jinchang and Huber, Elias X and Liu, Zi-Wen and Ma, Xiongfeng},
  journal={arXiv preprint arXiv:2509.17580v3},
  year={2026}
}
\newpage
\appendix

\crefalias{section}{appendix}



\section{Notational glossary}
\hypertarget{gloss}{}
\fancyfoot[R]{}
\label{sec:gloss}

\begin{notation}
    \item[$F_\mu$]
    The set of $f:\{0,1\}^n\to \C$ such that $\mu_f=\mu$, where $\mu_f(x)=|f(x)|^2$.
    \item[\(\eta\)]
    Dimensionless smoothing parameter used in the regularization.

    \item[\(\delta\)]
    Dimensional smoothing parameter, related to \(\eta\) by
    \(\delta=\eta/N\).

    \item[\(\beta\)]
    Density of the bad set \(B\subseteq\{0,1\}^n\).

    \item[\(\beta_{\max}\)]
    Density threshold at which Tropp's random-paving result applies.


    \item[\(\Bias(\mu)\)]
    Largest magnitude among $\wh \mu(s)$, $s\neq 0^n$.

    \item[$\frechet$]
    Fréchet derivative.

    \item[\(m_\eta\)]
    Centering constant used in the Gaussianization of
    \(\widehat{M}_{\widehat{f}}\).

    \item[\(d_\eta\)]
    Centering constant used in the Gaussianization of
    \(D_{\widehat{f}}\).
\end{notation}

\section{Large covering number but small \texorpdfstring{$\Inf_\mu$}{influence}}
\label{app:badspace}
\begin{prop}
For every $\kappa>0$, with probability at least $1-\exp\bigl(-\Omega_\kappa(2^n)\bigr)$,
there exists a subspace $V$ of functions on $\{0,1\}^n$ of dimension $\Omega_\kappa(2^n)$
such that every $q\in V$, $\|q\|_{L^2(\mu)}=1$, satisfies
\[
\Inf_\mu[f]\leq \kappa n.
\]
\end{prop}

The idea is to construct a subspace of functions with basis $q_x:=\mathbf 1_{\{x\}}/\sqrt{\mu(x)}$ for $x$ with $\mu(x)$ large yet $\mu(y)$ small for all neighbors $y\sim x$.
Because the smaller weight on an edge governs the conductance (edge weight), $\Inf_\mu(q_x)$ will stay small  for these $x$.

\begin{proof}
It will be convenient to work with the Exponential model for the Porter--Thomas distribution.
Write
\[
\mu(x)=\frac{Z_x}{\sum_z Z_z},
\]
where the $Z_x$, $x\in\{0,1\}^n$ are independent $\mathrm{Exp}(1)$ random variables.

Let $E,O\subset\{0,1\}^n$ be all even and odd strings respectively.
For $x\in E$, set
\[
R_x=\sum_{y\sim x} Z_y.
\]
Since every $y\in O$ has exactly $n$ neighbors in $E$,
\[
\tsum_{x\in E} R_x
=
n\tsum_{y\in O} Z_y.
\]

By a Chernoff bound,
\[
\Pr\left(\sum_{y\in O} Z_y>\frac{3N}{4}\right)
\leq \exp(-cN).
\]
On the complementary event, the number of vertices $x\in E$ for which
$R_x>2n$ is at most
\[
\frac{1}{2n}\sum_{x\in E}R_x
=
\frac12\sum_{y\in O}Z_y
\leq \frac{3N}{8}.
\]
Since $|E|=N/2$, there are therefore at least $N/8$ vertices $x\in E$
satisfying
\[
R_x\leq 2n.
\]

Condition on the variables $(Z_y)_{y\in O}$. For the vertices satisfying
$R_x\leq 2n$, the events
\[
Z_x\geq \frac{2}{\kappa}
\]
are independent.
Define
\[
p_\kappa
=
\Pr\left(Z\geq \frac{2}{\kappa}\right)>0
\]
for a fresh $Z\sim\mathrm{Exp}(1)$.
Another Chernoff bound shows that, except with probability
$\exp(-\Omega_\kappa(N))$, at least $p_\kappa N/16$ vertices satisfy
both
\begin{equation}
\label{eq:sad-event}
R_x\leq 2n
\qquad\text{and}\qquad
Z_x\geq \frac{2}{\kappa}.
\end{equation}

Let $S\subseteq E$ be the set of such vertices. For every $x\in S$,
\[
\sum_{y\sim x}
\frac{\mu(y)}{\mu(x)+\mu(y)}
=
\sum_{y\sim x}
\frac{Z_y}{Z_x+Z_y}
\leq
\frac{R_x}{Z_x}
\leq \kappa n.
\]

Now let
\[
V=\{q:\operatorname{supp}(q)\subseteq S\}.
\]
Then
\[
\dim V=|S|\geq \frac{p_\kappa}{16}N.
\]
Moreover, since $S$ is contained in one parity class, it is an
independent set.
Hence, for every $q\in V$, conditioned on the event in \eqref{eq:sad-event}, we get
\[
\operatorname{Inf}_\mu[q]
=
\sum_{x\in S}
\mu(x)|q(x)|^2
\sum_{y\sim x}
\frac{\mu(y)}{\mu(x)+\mu(y)}
\leq
\kappa n\sum_{x\in S}\mu(x)|q(x)|^2
\leq
\kappa n\|f\|_{L^2(\mu)}^2=\kappa n.
\]
as desired.
\end{proof}

\section{Relating $\Avg_f$ to $\Inf_\mu$}
\label{app:identity}
Here we provide a proof of the identity 
$$\Avg_f[g] = 1-\frac1n\Inf_{\mu_f}\left[\frac{g}{f}\right] \,.$$

We first use the following elementary identity. For any \(a,b\geq 0\)
and \(u,v\in\mathbb C\),
\[
\frac{|au+bv|^2}{a+b}
=
a|u|^2+b|v|^2-\frac{ab}{a+b}|u-v|^2.
\]
Indeed,
\begin{align*}
|au+bv|^2
&=a^2|u|^2+b^2|v|^2+2ab\operatorname{Re}(u\overline v),
\end{align*}
and
\begin{align*}
(a+b)\bigl(a|u|^2+b|v|^2\bigr)-ab|u-v|^2
&=(a^2+ab)|u|^2+(ab+b^2)|v|^2
   -ab\bigl(|u|^2+|v|^2-2\operatorname{Re}(u\overline v)\bigr)
\\
&=a^2|u|^2+b^2|v|^2+2ab\operatorname{Re}(u\overline v)
\\
&=|au+bv|^2.
\end{align*}
Dividing by \(a+b\) gives the claimed identity.

Let $q = g/f$ (here we are assuming $f$ is nonzero everywhere, which happens with probability $1$ for a Haar-random $f$).
We now apply the above identity with
\[
a=|f(x)|^2,\qquad
b=|f(x+i)|^2,\qquad
u=q(x),\qquad
v=q(x+i).
\]
This gives
\begin{align*}
\frac{
\left|
  |f(x)|^2q(x)+|f(x+i)|^2q(x+i)
\right|^2
}{
  |f(x)|^2+|f(x+i)|^2
}
&=
|f(x)|^2|q(x)|^2
+
|f(x+i)|^2|q(x+i)|^2
\\
&\hspace{3em}
-
\frac{|f(x)|^2|f(x+i)|^2}
     {|f(x)|^2+|f(x+i)|^2}
\left|q(x)-q(x+i)\right|^2.
\end{align*}
Therefore,
\begin{align*}
\Avg_f[g] &=\frac{1}{2n}\sum_i\sum_x\frac{\Big|\overline{f(x)}g(x)+\overline{f(x+i)}g(x+i)\Big|^2}{|f(x)|^2+|f(x+i)|^2} \\
&= \frac{1}{2n}
\sum_{i}
\sum_{x}
\frac{
\left|
  |f(x)|^2q(x)+|f(x+i)|^2q(x+i)
\right|^2
}{
  |f(x)|^2+|f(x+i)|^2
}
\\
&=
\frac{1}{2n}
\sum_{i,x}
\Big(
 |f(x)|^2|q(x)|^2
 +
 |f(x+i)|^2|q(x+i)|^2
\Big)
\\
&\hspace{3em}
-
\frac{1}{2n}
\sum_{i,x}
\frac{|f(x)|^2|f(x+i)|^2}
     {|f(x)|^2+|f(x+i)|^2}
\left|q(x)-q(x+i)\right|^2.
\end{align*}
Now, notice that
\begin{align*}
\frac{1}{2n}
\sum_{i,x}
\Big(
 |f(x)|^2|q(x)|^2
 +
 |f(x+i)|^2|q(x+i)|^2
\Big)
&=
\frac{1}{2n}\sum_i
2\sum_x |f(x)|^2|q(x)|^2
\\
&=
\sum_x |f(x)|^2|q(x)|^2
\\
&=
\sum_x |g(x)|^2
=1.
\end{align*}
Substituting this into the previous expression, we get
\begin{align*}
&\frac{1}{2n}
\sum_i\sum_x
\frac{
\left|
  |f(x)|^2q(x)+|f(x+i)|^2q(x+i)
\right|^2
}{
  |f(x)|^2+|f(x+i)|^2
}
\\
&=
1-
\frac{1}{2n}
\sum_i\sum_x
\frac{|f(x)|^2|f(x+i)|^2}
     {|f(x)|^2+|f(x+i)|^2}
\left|q(x)-q(x+i)\right|^2 \\
&= 1-\frac1n\Inf_{\mu_f}\left[\frac{g}{f}\right] \,,
\end{align*}
as desired.


\end{document}